\documentclass[reqno]{article}

\usepackage[utf8]{inputenc}   
\usepackage[T1]{fontenc}      

\usepackage{amsmath,amsthm} 
\usepackage{amssymb,mathrsfs} 
\usepackage{amssymb}
\usepackage{amsfonts}
\usepackage{upgreek}
\usepackage{nicefrac}
\usepackage{geometry}
\usepackage{authblk}
\usepackage[numbers]{natbib}
\usepackage{graphicx}
 \usepackage{color}
\usepackage[colorlinks = true, citecolor = blue]{hyperref}
\usepackage{algpseudocode,algorithm,algorithmicx}
\algnewcommand{\Inputs}[1]{%
  \State \textbf{Inputs:}
  \Statex \hspace*{\algorithmicindent}\parbox[t]{.8\linewidth}{\raggedright #1}
}
\algnewcommand{\Initialize}[1]{%
  \State \textbf{Initialize:}
  \Statex \hspace*{\algorithmicindent}\parbox[t]{.8\linewidth}{\raggedright #1}
}
\algnewcommand{\Outputs}[1]{%
  \State \textbf{Outputs:}
  \Statex \hspace*{\algorithmicindent}\parbox[t]{.8\linewidth}{\raggedright #1}
}

\usepackage{stmaryrd}

\usepackage{enumitem}
\usepackage{url}
\usepackage{graphicx} 
\usepackage{lmodern} 
\usepackage{xcolor}
\usepackage{bbm}

\usepackage{ifthen}
\usepackage{xargs}

\usepackage[textwidth=1.8cm]{todonotes}

\usepackage{aliascnt}
\usepackage{cleveref}
\usepackage{autonum}
\makeatletter
\newtheorem{theorem}{Theorem}
\crefname{theorem}{theorem}{Theorems}
\Crefname{Theorem}{Theorem}{Theorems}

\newtheorem*{lemma_nonumber*}{Lemma}

\newaliascnt{lemma}{theorem}
\newtheorem{lemma}[lemma]{Lemma}
\aliascntresetthe{lemma}
\crefname{lemma}{lemma}{lemmas}
\Crefname{Lemma}{Lemma}{Lemmas}

\newaliascnt{corollary}{theorem}

\aliascntresetthe{corollary}
\crefname{corollary}{corollary}{corollaries}
\Crefname{Corollary}{Corollary}{Corollaries}

\newaliascnt{proposition}{theorem}
\newtheorem{proposition}[proposition]{Proposition}
\aliascntresetthe{proposition}
\crefname{proposition}{proposition}{propositions}
\Crefname{Proposition}{Proposition}{Propositions}

\newaliascnt{definition}{theorem}

\aliascntresetthe{definition}
\crefname{definition}{definition}{definitions}
\Crefname{Definition}{Definition}{Definitions}

\newaliascnt{remark}{theorem}

\aliascntresetthe{remark}
\crefname{remark}{remark}{remarks}
\Crefname{Remark}{Remark}{Remarks}

\crefname{example}{example}{examples}
\Crefname{Example}{Example}{Examples}

\crefname{figure}{figure}{figures}
\Crefname{Figure}{Figure}{Figures}

\newtheorem{assumption}{\textbf{A}\hspace{-3pt}}
\crefformat{assumption}{{\textbf{A}}#2#1#3}

\newtheorem{assumptionH}{\textbf{H}\hspace{-3pt}}
\Crefname{assumptionH}{\textbf{H}\hspace{-3pt}}{\textbf{H}\hspace{-3pt}}
\crefname{assumptionH}{\textbf{H}}{\textbf{H}}

\newtheorem{assumptionL}{\textbf{L}\hspace{-3pt}}
\Crefname{assumptionL}{\textbf{L}\hspace{-3pt}}{\textbf{L}\hspace{-3pt}}
\crefname{assumptionL}{\textbf{L}}{\textbf{L}}

\Crefname{assumptionQ}{\textbf{Q}\hspace{-3pt}}{\textbf{Q}\hspace{-3pt}}
\crefname{assumptionQ}{\textbf{Q}}{\textbf{Q}}

\newcommand{\tetaa}[1]{\teta_{#1}^{(a)}}
\newcommand{\tetab}[1]{\teta_{#1}^{(b)}}
\newcommand{\tetac}[1]{\teta_{#1}^{(c)}}
\newcommand{\tetad}[1]{\teta_{#1}^{(d)}}

\newcommand{\mttdeux}{\mathtt{m}_1}
\newcommand{\tmttdeux}{\tilde{\mtt}_1}
\newcommand{\mtttrois}{\mathtt{m}_2}

\newcommand{\Rdeux}{R_1}

\newcommand{\tE}{\tilde{E}}
\newcommand{\expo}{a}
\newcommand{\bartheta}{\hat{\theta}}
\newcommand{\Rtheta}{M_{\Theta}}

\newcommand{\tup}[1]{\textup{#1}}
\newcommand{\Kker}{\mathrm{K}}
\newcommand{\Rker}{\mathrm{R}}
\newcommand{\Pker}{\mathrm{P}}

\newcommand{\LUx}{\mathtt{L}}

\newcommand{\X}{\tilde{X}}
\newcommand{\F}{\tilde{\mathcal{F}}}

\newcommand{\poiss}[1]{\hat{H}_{\gamma_{#1},\theta_{#1}}}
\newcommand{\tpoiss}[1]{\hat{H}_{\gamma_{#1},\ttheta_{#1}}}

\newcommand{\tkernel}[1]{\Kker_{\gamma_{#1}, \ttheta_{#1}}}
\newcommand{\poissp}[2]{\hat{H}_{#1,#2}}

\newcommand{\Rgt}{\Rker_{\gamma, \theta}}

\def\y{{\boldsymbol y}}

\def\L{L_f}
\newcommandx{\norm}[2][1=]{\ifthenelse{\equal{#1}{}}{\left\Vert #2 \right\Vert}{\left\Vert #2 \right\Vert^{#1}}}
\newcommandx{\normLigne}[2][1=]{\ifthenelse{\equal{#1}{}}{\Vert #2 \Vert}{\Vert #2\Vert^{#1}}}

\def\bfn{\mathbf{n}}

\def\bfDd{\mathbf{D}_{\mathrm{d}}}
\def\bfDc{\mathbf{D}_{\mathrm{c}}}
\def\Psibf{\boldsymbol{\Psi}}
\def\Xibf{\boldsymbol{\Xi}}
\def\Lambdabf{\boldsymbol{\Lambda}}

\def\mtt{\mathtt{m}}

\def\msa{\mathsf{A}}

\def\msu{\mathsf{U}}

\def\msy{\mathsf{Y}}

\def\mcbb{\mathcal{B}}  
\newcommand{\mcb}[1]{\mathcal{B}(#1)}

\def\mcf{\mathcal{F}}

\def\rset{\mathbb{R}}

\def\nset{\mathbb{N}}
\def\nsets{\mathbb{N}^*}

\def\Rset{\mathbb{R}}

\def\rmd{\mathrm{d}}

\def\rme{\mathrm{e}}

\def\mrc{\mathrm{C}}

\def\rmc{\mathrm{C}}

\newcommand{\R}{\mathbb R}

\newcommand{\A}{\mathcal A}

\newcommandx{\functionspace}[2][1=+]{\mathbb{F}_{#1}(#2)}
\newcommand{\argmax}{\operatorname*{arg\,max}}
\newcommand{\argmin}{\operatorname*{arg\,min}}

\newcommandx{\VarDeux}[3][3=]{\operatorname{Var}^{#3}_{#1}\left\{#2 \right\}}

\newcommand{\1}{\mathbbm{1}}
\newcommand{\Unbf}{\boldsymbol{1}}

\newcommand{\LeftEqNo}{\let\veqno\@@leqno}

\newcommand{\ceil}[1]{\left\lceil #1 \right\rceil}

\newcommand{\N}{\ensuremath{\mathbb{N}}}

\newcommand{\PE}{\mathbb{E}}
\newcommand{\PP}{\mathbb{P}}

\newcommand{\abs}[1]{\left\vert #1 \right\vert}

\newcommand{\tvnorm}[1]{\| #1 \|_{\mathrm{TV}}}

\newcommandx{\Vnorm}[2][1=V]{\| #2 \|_{#1}}
\newcommandx{\VnormEq}[2][1=V]{\left\| #2 \right\|_{#1}}

\newcommand{\parenthese}[1]{\left(#1 \right)}
\newcommand{\parentheseLigne}[1]{(#1 )}
\newcommand{\parentheseDeux}[1]{\left[ #1 \right]}
\newcommand{\defEns}[1]{\left\lbrace #1 \right\rbrace }

\newcommand{\ps}[2]{\left\langle#1,#2 \right\rangle}

\newcommand{\proba}[1]{\mathbb{P}\left( #1 \right)}

\newcommand{\probaLigne}[1]{\mathbb{P}( #1 )}
\newcommandx\probaMarkovTilde[2][2=]
{\ifthenelse{\equal{#2}{}}{{\widetilde{\mathbb{P}}_{#1}}}{\widetilde{\mathbb{P}}_{#1}\left[ #2\right]}}

\newcommand{\expe}[1]{\PE \left[ #1 \right]}

\newcommand{\expeLigne}[1]{\PE [ #1 ]}

\newcommand{\bigO}{\ensuremath{\mathcal O}}

\newcommand{\plusinfty}{+\infty}

\newcommand\numberthis{\addtocounter{equation}{1}\tag{\theequation}}

\def\ie{\textit{i.e.}}

\def\eqsp{\;}
\newcommand{\coint}[1]{\left[#1\right)}
\newcommand{\ocint}[1]{\left(#1\right]}
\newcommand{\ooint}[1]{\left(#1\right)}
\newcommand{\ccint}[1]{\left[#1\right]}

\renewcommand{\iint}[2]{\left\lbrace #1,\ldots,#2\right\rbrace}

\newcommandx{\weight}[2][2=n]{\omega_{#1,#2}^N}

\newcommand{\boule}[2]{\operatorname{B}(#1,#2)}
\newcommand{\ball}[2]{\operatorname{B}(#1,#2)}
\newcommand{\boulefermee}[2]{\overline{\mathrm{B}}(#1,#2)}
\newcommand{\cball}[2]{\boulefermee{#1}{#2}}

\def\as{almost surely}

\newcommandx\sequence[3][2=,3=]
{\ifthenelse{\equal{#3}{}}{\ensuremath{\{ #1_{#2}\}}}{\ensuremath{\{ #1_{#2}, \eqsp #2 \in #3 \}}}}
\newcommandx\sequenceD[3][2=,3=]
{\ifthenelse{\equal{#3}{}}{\ensuremath{\{ #1_{#2}\}}}{\ensuremath{( #1)_{ #2 \in #3} }}}

\newcommandx{\sequencen}[2][2=n\in\N]{\ensuremath{\{ #1_n, \eqsp #2 \}}}
\newcommandx\sequenceDouble[4][3=,4=]
{\ifthenelse{\equal{#3}{}}{\ensuremath{\{ (#1_{#3},#2_{#3}) \}}}{\ensuremath{\{  (#1_{#3},#2_{#3}), \eqsp #3 \in #4 \}}}}
\newcommandx{\sequencenDouble}[3][3=n\in\N]{\ensuremath{\{ (#1_{n},#2_{n}), \eqsp #3 \}}}

\newcommand{\wrt}{with respect to}

\def\iid{i.i.d.}

\def\eg{e.g.}

\newcommand{\opnorm}[1]{{\left\vert\kern-0.25ex\left\vert\kern-0.25ex\left\vert #1 
    \right\vert\kern-0.25ex\right\vert\kern-0.25ex\right\vert}}

\def\generator{\mathcal{A}}

\def\Id{\operatorname{Id}}
\def\Idd{\operatorname{I}}

\newcommandx{\CPE}[3][1=]{{\mathbb E}_{#1}\left[\left. #2 \middle \vert #3 \right. \right]} 
\newcommandx{\CPELigne}[3][1=]{{\mathbb E}_{#1}[ #2 \vert #3  ]} 
\newcommandx{\CPVar}[3][1=]{\mathrm{Var}^{#3}_{#1}\left\{ #2 \right\}}
\newcommand{\CPP}[3][]
{\ifthenelse{\equal{#1}{}}{{\mathbb P}\left(\left. #2 \, \right| #3 \right)}{{\mathbb P}_{#1}\left(\left. #2 \, \right | #3 \right)}}

\newcommandx{\osc}[2][1=]{\mathrm{osc}_{#1}(#2)}

\def\Id{\operatorname{Id}}

\def\domain{\mathrm{D}}

\def\transpose{\operatorname{T}}

\def\y{y}

\def\bgamma{\bar{\gamma}}

\def\tr{\tilde{r}}

\def\tx{\tilde{x}}

\def\tX{\tilde{X}}

\def\tE{\tilde{E}}

\def\teta{\tilde{\eta}}
\newcommand{\ttheta}{\tilde{\theta}}

\def\sign{\operatorname{sign}}

\newcommand{\ensemble}[2]{\left\{#1\,:\eqsp #2\right\}}

\def\mrc{\mathrm{C}}

\newcommand\coupling[2]{\Gamma(\mu,\nu)}

\newcommand{\complementary}{\mathrm{c}}

\def\interior{\mathrm{int}}

\def\vareps{\varepsilon}
\def\varespilon{\varepsilon}

\newcommandx{\KL}[2]{\text{KL}\left( #1 | #2 \right)}
\def\vol{\operatorname{Vol}}

\newcommand*\Let[1]{\State #1 }

\newcommand\fraca[2]{(#1)/(#2)}

\def\ccur{\mathtt{c}}
\def\Ve{V_{\rme}}

\def\Ber{\mathrm{Ber}}
\def\loiGauss{\mathrm{N}}
\def\train{\mathrm{train}}
\def\test{\mathrm{test}}

\def\measMat{\mathbf{M}}
\def\Psimatrix{\boldsymbol{\Psi}}
\def\MAP{\mathrm{MAP}}
\def\true{\mathrm{true}}

\title{Efficient stochastic optimisation by unadjusted Langevin Monte Carlo. Application to maximum marginal likelihood and empirical Bayesian estimation.}

\author[1]{Valentin De Bortoli \footnote{Email: valentin.debortoli@cmla.ens-cachan.fr}}
\author[1]{Alain Durmus \footnote{Email: alain.durmus@cmla.ens-cachan.fr} }
\author[2]{Marcelo Pereyra \footnote{Email: m.pereyra@hw.ac.uk} }
\author[2]{Ana Fernandez Vidal \footnote{Email: af69@hw.ac.uk}}

\affil[1]{CMLA - \'Ecole normale supérieure Paris-Saclay, CNRS, Université Paris-Saclay, 94235 Cachan, France.}
\affil[2]{School of Mathematical and Computer Sciences, Heriot Watt University \& Maxwell Institute for Mathematical Sciences.}

\begin{document}

\maketitle

\begin{abstract}
	Stochastic approximation methods play a central role in maximum likelihood estimation problems involving intractable likelihood functions, such as marginal likelihoods arising in problems with missing or incomplete data, and in parametric empirical Bayesian estimation. Combined with Markov chain Monte Carlo algorithms, these stochastic optimisation methods have been successfully applied to a wide range of problems in science and industry. However, this strategy scales poorly to large problems because of methodological and theoretical difficulties related to using high-dimensional Markov chain Monte Carlo algorithms within a stochastic approximation scheme. This paper proposes to address these difficulties by using unadjusted Langevin algorithms to construct the stochastic approximation. This leads to a highly efficient stochastic optimisation methodology with favourable convergence properties that can be quantified explicitly and easily checked.
The proposed methodology is demonstrated with three experiments, including a challenging application to high-dimensional statistical audio analysis and a sparse Bayesian logistic regression with random effects problem.   
\end{abstract}


\section{Introduction}
Maximum likelihood estimation (MLE) is central to modern statistical science. It is a cornerstone of frequentist inference \citep{Casella2002}, and also plays a fundamental role in parametric empirical Bayesian inference \citep{carlin2000empirical, casella1985introduction}. For simple statistical models, MLE can be performed analytically and exactly. However, for most models, it requires using numerical computation methods, particularly optimisation schemes that iteratively seek to maximise the likelihood function and deliver an approximate solution. Following decades of active research in computational statistics and optimisation, there are now several computationally efficient methods to perform MLE in a wide range of classes of models \citep{gentle2012handbook, boyd2004convex}. 

In this paper we consider MLE in models involving incomplete or ``missing'' data, such as hidden, latent or unobserved variables, and focus on Expectation Maximisation (EM) optimisation methods \citep{Dempster1977}, which are the predominant strategy in this setting . While the original EM optimisation methodology involved deterministic steps, modern EM methods are mainly stochastic \citep{Robert}. In particular, they typically rely on a Robbins-Monro stochastic approximation (SA) scheme that uses a Monte Carlo stochastic simulation algorithm to approximate the gradients that drive the optimisation procedure \citep{robbins1951stochastic, delyon1999, kushner2003stochastic, fort2011convergence}. In many cases, SA methods use Markov chain Monte Carlo (MCMC) algorithms, leading to a powerful general methodology which is simple to implement, has a detailed convergence theory \citep{atchade2017perturbed}, and can address a wide range of moderately low-dimensional models. Alternatively, some stochastic EM schemes use a Gibbs sampling algorithm \citep{Casella2001}, however this requires running several fully converged MCMC chains and can be significantly more computationally expensive as a result.



The expectations and demands on SA methods constantly rise as we seek to address larger problems and provide stronger theoretical guarantees on the solutions delivered. Unfortunately, existing SA methodology and theory do not scale well to large problems. The reasons are twofold. First, the family of MCMC kernels driving the SA scheme needs to satisfy uniform geometric ergodicity conditions that are usually difficult to verify for high-dimensional MCMC kernels. Second, the existing theory requires using asymptotically exact MCMC methods. In practice, these are usually high-dimensional Metropolis-Hastings methods such as the Metropolis-adjusted Langevin algorithm \citep{roberts:tweedie:1996} or Hamiltonian Monte Carlo \citep{Girolami2011,durmus2017convergence}, which are difficult to calibrate within the SA scheme to achieve a prescribed acceptance rate. For these reasons, practitioners rarely use SA schemes in high-dimensional settings.

In this paper, we propose to address these limitations by using inexact MCMC methods to drive the SA scheme, particularly unadjusted Langenvin algorithms, which have easily verifiable geometric ergodicity conditions, and are easy to calibrate \citep{durmus2017unadjusted, dalalyan2017theoretical}. This will allow us to design a high-dimensional stochastic optimisation scheme with favourable convergence properties that can be quantified explicitly and easily checked.

Our contributions are structured as follows: Section 2 formalises the class of MLE problems considered and presents the proposed stochastic optimisation method, which is based on a SA approach driven by an unadjusted Langevin algorithm. Section 3 presents three numerical experiments that demonstrate the proposed methodology in a variety of scenarios. Detailed theoretical convergence results for the method are reported in Section 4, which also describes a generalisation of the proposed methodology and theory to other inexact Markov kernels. The online supplementary material includes additional theoretical results and some details on computational aspects.


\section{The stochastic optimisation via unadjusted Langevin method}
\label{sec:sto_optim_Langevin}
The proposed Stochastic Optimisation via Unadjusted Langevin (SOUL) method is useful for solving maximum likelihood estimation problems involving intractable likelihood functions. The method is a SA iterative scheme that is driven by an unadjusted Langevin MCMC algorithm. Langevin algorithms are very efficient in high dimensions and lead to an SA scheme that inherits their favourable convergence properties.

\subsection{Maximum marginal likelihood estimation}
Let $\Theta$ be a convex closed set in $\rset^{d_{\Theta}}$. The proposed optimisation method is well-suited for solving maximum likelihood estimation problems of the form
\begin{equation}\label{eq:def_theta_star}
\begin{split}
  \theta^{\star} \in \argmax_{\theta \in \Theta}\,\, \log p(y|\theta) - g(\theta)\, ,
  \end{split}
\end{equation}
where the parameter of interest $\theta$ is related to the observed data $y \in \msy$ by a likelihood function $p(y,x|\theta)$ involving an unknown quantity $x \in \mathbb{R}^d$, which is removed from the model by marginalisation. More precisely, we consider problems where the resulting marginal likelihood
$$
p(y|\theta) = \int_{\rset^d}  p(y,x|\theta)\rmd x \, ,
$$
is computationally intractable, and focus on models where the dimension of $x$ is large, making the computation of \eqref{eq:def_theta_star} even more difficult. For completeness, we allow the use of a penalty function $g : \Theta \rightarrow \mathbb{R}$, or set $g=0$ to recover the standard maximum likelihood estimator. 

As mentioned previously, the maximum marginal likelihood estimation problem \eqref{eq:def_theta_star} arises in problems involving latent or hidden variables \citep{Dempster1977}. It is also central to parametric empirical Bayes approaches that base their inferences on the pseudo-posterior distribution $p(x|y,\theta^{\star}) = {p(y,x|\theta^{\star})}/{p(y|\theta^{\star})}$ \citep{carlin2000empirical}. Moreover, the same optimisation problem also arises in hierarchical Bayesian maximum-a-posteriori estimation of $\theta$ given $y$, with marginal posterior $p(\theta|y)\propto p(y|\theta)p(\theta)$ where $p(\theta)$ denotes the prior for $\theta$; in that case $g(\theta) = -\log p(\theta)$ \citep{Casella2002}.

Finally, in this paper we assume that $\log p(y,x|\theta)$ is continuously differentiable \wrt~$x$ and $\theta$, and that $g$ is also continuously differentiable \wrt~$\theta$
. A generalisation of the proposed methodology to non-smooth models is presented in a forthcoming paper \citep{vidal:debortoli:pereyra:durmus:2019} that focuses on non-smooth statistical imaging models.

\subsection{Stochastic approximation methods}
\label{sec:eb-inference-based}

The scheme we propose to solve the optimisation problem \eqref{eq:def_theta_star} is derived in the SA framework \citep{delyon1999}, which we recall below.

Starting from any $\theta_0 \in \Theta$, SA schemes seek to solve \eqref{eq:def_theta_star} iteratively by computing a sequence $(\theta_n)_{n \in \nset}$ associated with the recursion
\begin{equation}
  \label{eq:theta_it}
\theta_{n+1} = \Pi_{\Theta}[\theta_n + \delta_{n+1}(\Delta_{\theta_n} - \nabla g(\theta_n))] \eqsp ,
\end{equation}
where $\Delta_{\theta_n}$ is some estimator of the intractable gradient $\theta \mapsto \nabla_{\theta}\log p(y|\theta)$ at $\theta_n$, $\Pi_{\Theta}$ denotes the projection onto $\Theta$, and $(\delta_n)_{n \in \nsets} \in (\rset_+^*)^{\nsets}$ is a sequence of stepsizes. From an optimisation viewpoint, iteration \eqref{eq:theta_it} is a stochastic generalisation of the projected gradient ascent iteration \citep{boyd2004convex} for models with intractable gradients. For $n \in \nset$, Monte Carlo estimators $\Delta_{\theta_n}$ for $\nabla_{\theta}\log p(y|\theta)$ at $\theta_n$ are derived from the identity
\begin{align}
  \nabla_{\theta} \log p(y|\theta) &= \int_{\rset^d} \frac{\nabla_{\theta} p(x,y|\theta)}{p(x,y|\theta)} p(x|y,\theta) \rmd x\\
  \label{eq:def_nabla_theta}
&= \int_{\rset^d} \nabla_{\theta} \log p(x,y|\theta)p(x|y,\theta) \rmd x  \eqsp,
\end{align}
which suggests to consider 
\begin{equation}
  \label{eq:delta_n}
  \Delta_{\theta_n} = \frac{1}{m_n} \sum_{k=1}^{m_n} \nabla_{\theta} \log p(X_k^n,y |\theta_n)\eqsp ,
\end{equation}
where $(m_n)_{n \in \N}\in \left(\nsets\right)^{\N}$ is a sequence of batch sizes and $(X_k^n)_{k \in \lbrace 1, \dots, m_n \rbrace}$ is either an exact Monte Carlo sample from $p(x|y,\theta_n) = p(x,y|\theta_n)/p(y|\theta_n)$, or a sample generated by using a Markov Chain targeting this distribution.

Given a sequence $(\theta_n)_{n = 1}^N$ generated by using \eqref{eq:theta_it}, an approximate solution of \eqref{eq:def_theta_star} can then be obtained by calculating, for example, the average of the iterates, i.e.,
\begin{equation}
  \hat{\theta}_N = \left. \defEns{\sum_{n=1}^N \delta_n \theta_n}\middle/ \defEns{\sum_{n=1}^N \delta_n} \right. \eqsp . \label{eq:theta_avg}
\end{equation}
This estimate converges \as \ to a solution of \eqref{eq:def_theta_star} as $N \rightarrow \infty$ provided that some conditions on $p(y|\theta)$, $g$, $p(x|y,\theta)$, $(\delta_n)_{n \in \N}$, and $\Delta_{\theta_n}$ are fulfilled. Indeed, following three decades of active research efforts in computational statistics and applied probability,  we now have a good understanding of how to construct efficient SA schemes, and the conditions under which these schemes converge (see for example  \cite{benveniste:metivier:priouret:1990,fort:moulines:2003,duchi:hazan:singer:2011,andrieu:moulines:2006,nemirovski:juditsky:lan:shapiro:2008,atchade2017perturbed}).

SA schemes are successfully applied to maximum marginal likelihood estimation problems where the latent variable $x$ has a low or moderately low dimension. However, they are seldomly used them when $x$ is high-dimensional because this usually requires using high-dimensional MCMC samplers that, unless carefully calibrated, exhibit poor convergence properties. Unfortunately, calibrating the samplers within a SA scheme is challenging because the target density $p(x|y,\theta_n)$ changes at each iteration. As a result, it is, for example, difficult to use Metropolis-Hastings algorithms that need to achieve a prescribed acceptance probability range. Additionally, the conditions for convergence of MCMC SA schemes are often difficult to verify for high-dimensional samplers. For these reasons, practitioners rarely use SA schemes in high-dimensional settings.

As mentioned previously, we propose to address these difficulties by using modern inexact Langevin MCMC samplers to drive \eqref{eq:delta_n}. These samplers have received a lot of attention in the late because they can exhibit excellent large-scale convergence properties and significantly outperform their Metropolised counterparts (see \cite{durmus2016efficient} for an extensive comparison in the context of Bayesian imaging models). Stimulated by developments in high-dimensional statistics and machine learning, we now have detailed theory for these algorithms, including explicit and easily verifiable geometric ergodicity conditions \citep{durmus2017unadjusted, dalalyan2017theoretical,eberle2018quantitative,debortoli2018back}. This will allow us to design a stochastic optimisation scheme with favourable convergence properties that can be quantified explicitly and easily checked.

\subsection{Langevin Markov chain Monte Carlo methods}
Langevin MCMC schemes to sample from $p(x|y,\theta)$ are based on stochastic continuous dynamics $(\boldsymbol{X}_t^{\theta})_{t \geq 0}$ for which the target distribution $p(x|y,\theta)$ is invariant. Two fundamental examples are the Langevin dynamics solution of the following Stochastic Differential Equation (SDE)
\begin{equation}
  \label{eq:langevin}
  \rmd \boldsymbol{X}_t^{\theta}  = -\nabla_x \log p(\boldsymbol{X}_t^{\theta}|y,\theta)\rmd t + \sqrt{2} \rmd \boldsymbol{B}_t \eqsp ,
\end{equation}
or the kinetic Langevin dynamics solution of
\begin{equation}
  \label{eq:langevin_kin}
    \rmd \boldsymbol{X}_{t}^{\theta} = \boldsymbol{V}_t^{\theta} \eqsp , \qquad
    \rmd \boldsymbol{V}_t^{\theta} = -\nabla_x \log p(\boldsymbol{X}_t^{\theta}|y,\theta) \rmd t - \boldsymbol{V}_t^{\theta}\rmd t + \sqrt{2} \rmd \boldsymbol{B}_t \eqsp ,
\end{equation}
where $(\boldsymbol{B}_t)_{t \geq 0}$ is a standard $d$-dimensional Brownian motion. Under mild assumptions on $p(x|y,\theta)$, these two SDEs admit strong solutions for which $p(x|y,\theta)$ and $p(x,v|y,\theta) = p(x|y,\theta)\exp(-\norm[2]{v}/2)/(2\uppi)^{d/2}$ are the invariant probability measures. In addition, there are detailed explicit convergence results for $(\boldsymbol{X}_t^{\theta})_{t \geq 0}$ to $p(x|y,\theta)$, and for $(\boldsymbol{X}_t^{\theta},\boldsymbol{V}_t^{\theta})_{t \geq 0}$ to $p(x,v|y,\theta)$, under different metrics \citep{eberle2016reflection,eberle:guillin:zimmer:2017b}.

However, sampling path solutions for these continuous-time dynamics is  not feasible in general. Therefore discretizations have to be used instead. In this paper, we mainly focus on the Euler-Maruyama discrete-time approximation of \eqref{eq:langevin}, known as the Unadjusted Langevin Algorithm (ULA) \citep{roberts:tweedie:1996}, given by
\begin{equation}
  \label{eq:euler_maruyama_langevin}
  X_{k+1} = X_k - \gamma \nabla_{x}\log p(X_k|y,\theta) + \sqrt{2\gamma} Z_{k+1} \eqsp ,
\end{equation}
where $\gamma> 0$ is the discretization time step and $(Z_k)_{k \in \N^*}$ is a \iid~sequence of
$d$-dimensional zero-mean Gaussian random variables with covariance matrix identity. We will use this Markov kernel to drive our SA schemes. 

Observe that \eqref{eq:euler_maruyama_langevin} does not exactly target $p(x|y,\theta)$ because of the bias introduced by the discrete-time approximation. Computational statistical methods have traditionally addressed this issue by complementing \eqref{eq:euler_maruyama_langevin} with a Metropolis-Hastings correction step to asymptotically remove the bias \citep{roberts:tweedie:1996}. This correction usually deteriorates the convergence properties of the chain and may lead to poor non-asymptotic estimation results, particularly in very high-dimensional settings (see for example \cite{durmus2016efficient}). However, until recently it was considered that using \eqref{eq:euler_maruyama_langevin} without a correction step was too risky. Fortunately, recent works have established detailed theoretical guarantees for \eqref{eq:euler_maruyama_langevin} that do not require using any correction \citep{dalalyan2017theoretical,durmus2017unadjusted}. A main contribution of this work is to extend these guarantees to SA schemes that are driven by these highly efficient but inexact samplers.

\subsection{The SOUL algorithm}
\label{sec:soul-algorithm}
We are now ready to present the proposed Stochastic Optimization via Unadjusted Langevin (SOUL) methodology. Let $(\delta_n)_{n \in \nsets} \in (\rset_+^*)^{\nsets}$ and $(m_n)_{n \in \N} \in \left(\nsets\right)^{\N}$ be the sequences of stepsizes and batch sizes defining the SA scheme \eqref{eq:theta_it}-\eqref{eq:delta_n}. For any $\theta \in \Theta$ and $\gamma >0$, denote  by $\Rker_{\gamma, \theta}$ the Langevin Markov kernel  \eqref{eq:euler_maruyama_langevin} to approximately sample from $p(x|y,\theta)$, and by $(\gamma_n)_{n \in \N} \in (\rset_+^*)^{\nset}$ be the sequence of discrete time steps used.

Formally, starting from some $X_0^0 \in \rset^d$ and $\theta_0 \in \Theta$, for $n \in \nset$ and $k \in \lbrace 0, \dots, m_n -1 \rbrace$, we recursively define $(\lbrace X_k^n: k \in \lbrace 0, \dots, m_n \rbrace \rbrace, \theta_n)_{n \in \N}$ on a probability space $(\Omega,\mcf,\mathbb{P})$, where $(X_{k}^n)_{ k \in \{0,\ldots,m_n\}}$ is a Markov chain with Markov kernel $\Rker_{\gamma_n, \theta_n}$, $X_0^n = X_{m_{n-1}}^{n-1}$ given $\mathcal{F}_{n-1}$, and 
	$$
	 \theta_{n+1} = \Pi_{\Theta}\parentheseDeux{\theta_n - \frac{\delta_{n+1}}{m_{n}} \sum_{k=1}^{m_{n}} \Delta_{\theta_n}(X_k^n)} \eqsp,
	 $$
where we recall that $\Pi_{\Theta}$ is the projection onto $\Theta$, and   for all $n \in \nset$
\begin{equation}
 \label{eq:def_F_n}
   \mathcal{F}_n = \sigma \left( \theta_0, \{(X_k^{\ell})_{k \in \{0,\ldots,m_\ell\}} \, : \, \ell \in \iint{0}{n}\} \right) \eqsp ,  \qquad \mcf_{-1} = \sigma(\theta_0)
 \end{equation}
 Note that such a construction is always possible by Kolmogorov extension theorem \cite[Theorem 5.16]{kallenberg2006foundations}, hence for any $n \in \nset$, $\theta_{n+1}$ is $\mcf_n$-measurable. Then, as mentioned previously, we compute a sequence of approximate solutions of \eqref{eq:def_theta_star} by calculating, for example,
\begin{equation}
 \hat{\theta}_N = \left. \defEns{\sum_{n=1}^N \delta_n \theta_n}\middle/ \defEns{\sum_{n=1}^N \delta_n} \right. \eqsp . \label{eq:theta_avg}
\end{equation}

The pseudocode associated with the proposed SOUL method is presented in \Cref{alg:souk} below. Observe that, for additional efficiency, instead of generating independent Markov chains at each SA iteration, we warm-start the chains by setting $X_0^n = X_{m_{n-1}}^{n-1}$, for any $n \in \{1,\ldots, N\}$.

\begin{algorithm}
  \caption{The Stochastic Optimization via Unadjusted Langevin (SOUL) method
    \label{alg:souk}}
  \begin{algorithmic}[1]
    \Inputs{$\theta_0 \in \Theta$, $X_0^0 \in \rset^d$, $(\gamma_n)_{n \in \nset}$, $(\delta_n)_{n \in \nset}$, $(m_n)_{n \in \nset}$, $N$}
     \For{$n \in \{1,\ldots, N-1\}$ }
      \If{$n \geq 1$}
     \Let{$X_0^n = X_{m_{n-1}}^{n-1}$}
     \EndIf
     \For{$k \in \{0,\ldots,m_{n}-1\}$}
     \Let{$Z_{k+1}^n \sim \loiGauss(0,\Idd_d)$}
     \Let{$X_{k+1}^n = X_{k}^n + \gamma_n \nabla_{x} \log p(X^n_{k}|y,\theta_n) + \sqrt{2\gamma_n}Z_{k+1}^n$}
     \EndFor
     \Let{$\Delta_{\theta_n} = \frac{1}{m_n} \sum_{k=1}^{m_n} \nabla_{\theta} \log p(X_k^n,y |\theta_n)$}
     \Let{$\theta_{n+1} = \Pi_{\Theta}[\theta_n + \delta_{n+1}(\Delta_{\theta_n} - \nabla g(\theta_n))]$}
     \EndFor
     \Outputs{$\hat{\theta}_N = \left. \defEns{\sum_{n=1}^N \delta_n \theta_n}\middle/ \defEns{\sum_{n=1}^N \delta_n} \right.$}
  \end{algorithmic}
\end{algorithm}

To conclude, \Cref{sec:numerical-results} below demonstrates the proposed methodology with three numerical experiments related to high-dimensional logistic regression and statistical audio analysis with sparsity promoting priors. A detailed theoretical analysis of the proposed SOUL method is reported in \Cref{sec:convergence-results}. More precisely, we establish that if the cost function $f(\theta) = g(\theta) - \log p(y|\theta)$ defining \eqref{eq:def_theta_star} is convex, and if $(\gamma_n)_{n \in \N}$ and $(\delta_n)_{n \in \N}$ go to $0$ sufficiently fast, then $\expeLigne{f(\hat{\theta}_N)}$ converges to $\min_{\Theta} f$ and quantify the rate of convergence. Moreover, in the case where $(\gamma_n)_{n \in \N}$ is held fixed, \ie \ for all $n \in \N$, $\gamma_n = \gamma$, we show convergence to a neighbourhood of the solution, in the sense that there exist explicit $C, \alpha >0$ such that $\limsup_{N \to +\infty} \expeLigne{f(\hat{\theta}_N)} - \min_{\Theta} f \leq C\gamma^{\alpha}$. Finally, we also study the important case where $f$ is not convex. In that case, we use the results of \cite{kushner2003stochastic} to establish that $(\theta_n)_{n \in \N}$ converges \as \ to a stationary point of the projected ordinary differential equation associated with $\nabla f$ and $\Theta$. We postpone this result to \Cref{sec:non-convex-objective} in the supplementary document because it is highly technical.

\section{Numerical results}
\label{sec:numerical-results}
We now demonstrate the proposed methodology with three experiments that we have chosen to illustrate a variety of scenarios.  \Cref{ssec:num-res-logit} presents an application to empirical Bayesian logistic regression, where \eqref{eq:def_theta_star} can be analytically shown to be a convex optimisation problem with an unique solution $\theta^{\star}$, and where we benchmark our MLE estimate against the solution obtained by calculating the marginal likelihood $p(y|\theta)$ over a $\theta$-grid by using an harmonic mean estimator. Furthermore,  \Cref{ssec:num-res-cs-audio} presents a challenging application related to statistical audio compressive sensing analysis, where we use SOUL to estimate a regularisation parameter that controls the degree of sparsity enforced, and where a main difficulty is the high-dimensionality of the latent space ($d = 2,900$). Finally,  \Cref{ssec:num-res-logit-rand-effects} presents an application to a high-dimensional empirical Bayesian logistic regression with random effects for which the optimisation problem \eqref{eq:def_theta_star} is not convex. All experiments were carried out on an Intel i9-8950HK@2.90GHz workstation running Matlab R2018a.

\subsection{Bayesian Logistic Regression}
\label{ssec:num-res-logit}
In this first experiment we illustrate the proposed methodology with an empirical Bayesian logistic regression problem \citep{wakefield2013bayesian,polson2013bayesian}. We observe a set of covariates $\{v_i\}^{d_y}_{i=1} \in \rset^d$, and binary responses  $\{y_i\}^{d_y}_{i=1} \in \{0,1\}$, which we assume to be conditionally independent realisations of a logistic regression model: for any $i \in \{1,\ldots, d_y\}$, $y_i$ given $\beta$ and $v_i$ has distribution  $\Ber( s(v_i^{\transpose}\beta))$,
where $\beta \in \rset^d$ is the regression coefficient, $\Ber(\alpha)$ denotes the Bernoulli distribution with parameter $\alpha \in [0, 1]$ and $s(u) = \rme^u /(1 + \rme^u )$ is the cumulative distribution function of the standard logistic distribution. The prior for $\beta$ is set to be $\loiGauss(\theta \Unbf_d,\sigma^2\Idd_d)$, the $d$-dimensional Gaussian distribution with mean $\theta \Unbf_d$ and  covariance matrix $\sigma^2 \Idd_d$, where $\theta$ is the parameter we seek to estimate, $\Unbf_d = (1,\ldots,1) \in \rset^{d}$, $\sigma^2 = 5$ and $\Idd_d$ is the $d$-dimensional identity matrix. 
Following an empirical Bayesian approach, the parameter $\theta$ is  computed by maximum marginal likelihood estimation using  \Cref{alg:souk} with the marginal likelihood $ p(y|\theta)$ given by
\begin{equation}
 p(y|\theta)=  (2\uppi\sigma^2)^{-d/2}\int_{\rset^d} \defEns{\prod_{i=1}^{{d_y}} s(v_i^{\transpose} \beta)^{y_{i}}(1-s(v_i^{\transpose} \beta))^{1-y_{i}}}\rme^{-\frac{\norm[2]{\beta - \theta \Unbf_d}}{2 \sigma^2}}\rmd \beta \eqsp.
\label{EQ: bayes-logit-marg-likelihood}
\end{equation}
\Cref{lemma:prekoppa} in \Cref{sec:posterior} of the supplementary document shows  that \eqref{EQ: bayes-logit-marg-likelihood} is log-concave with respect to $\theta$. We use the proposed SOUL methodology to estimate $\theta^{\star}$ for the Wisconsin Diagnostic Breast Cancer dataset\footnote{\small Available online: \url{https://archive.ics.uci.edu/ml/datasets/Breast+Cancer +Wisconsin+(Diagnostic)}}, for which $d_y = 683$ and $d=10$, and where we suitably normalise the covariates. In order to assess the quality of our estimation results, we also calculate $p(y|\theta)$ over a grid of values for $\theta$ by using a truncated harmonic mean estimator.

To implement \Cref{alg:souk} we derive the log-likelihood function 
\begin{equation}
\log p(y|\beta,\theta)=\sum_{i=1}^{{d_y}}\left\{ y_{i}v_i^{\transpose} \beta-\log(1+\rme^{(v_i^{\transpose} \beta)})\right\} \eqsp,
\label{EQ: bayes-logit-likelihood}
\end{equation}
and obtain the following expressions for the gradients used in the MCMC steps \eqref{eq:euler_maruyama_langevin} and SA steps \eqref{eq:theta_it} respectively
\begin{align}
  \label{EQ: bayes-logit-grad-beta-ULA}
\nabla_{\beta}\log p(\beta|y,\theta)&= 
                                      \sum_{i=1}^{{d_y}}\left\{  y_i v_i - s(v_i^{\transpose} \beta) v_i \right\} - \frac{(\beta-\theta \mathbf{1}_d)}{\sigma^2} \eqsp,\\
  \label{EQ: bayes-logit-gradTheta}
   \nabla_{\theta} p(\beta,y |\theta) &= \ps{\Unbf_d}{\beta-\theta \Unbf_d}/\sigma^2 \eqsp.                                       
\end{align}
For the MCMC steps, we use a fixed stepsize $\gamma_n = 8.34 \times 10^{-5}$,  and batch size $m_n=1$, for any $n \in \nset$. On the other hand, we consider for the SA steps, the sequence of stepsizes $\delta_n=60/ n^{0.8}$, $\Theta=\ccint{-100,100}$ and  $\theta_0=0$. Finally,  we first run $100$ burn-in iterations with fixed $\theta_n = \theta_0$ to warm-up the Markov chain, followed by $50$ iterations of  \Cref{alg:souk} to warm-up the iterates. This procedure is then followed by $N=10^6$ iterations of  \Cref{alg:souk} to compute $\hat{\theta}_N$.

\begin{figure}[h!]
		\begin{minipage}[l1]{0.57 \linewidth}
			\centering
			\includegraphics[width=   
			\textwidth]{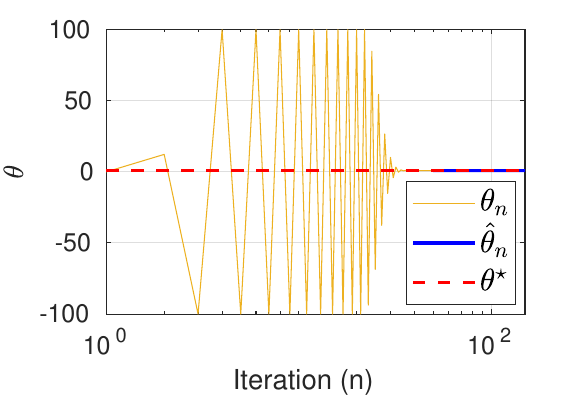}
			\centerline{(a)}
		\end{minipage}		
	\begin{minipage}[l1]{.42\linewidth}
		\centering
		\includegraphics[clip, trim=0.06cm 0cm 0.4cm 0cm, width=   
		\textwidth]{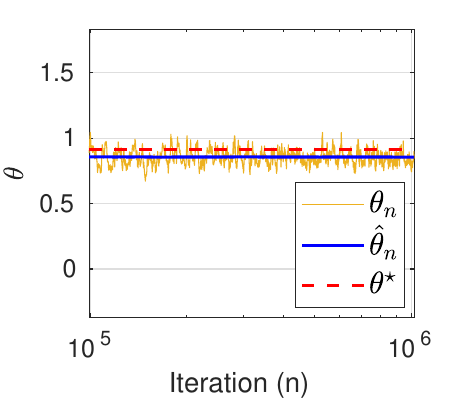}
		\centerline{(b)}
	\end{minipage}	
	\caption{\small Bayesian logistic regression - Evolution of the iterates $ \hat{\theta}_n $ and $\theta_{n}$ for the proposed method during (a) burn-in phase and (b) convergence phase. An estimate of $\theta^{\star}$, the true maximiser of $p(y|\theta)$, is plotted as a reference.} \label{fig:num-res-bayes-logit-theta} 
\end{figure}
 \Cref{fig:num-res-bayes-logit-theta}(a) shows the evolution of the iterates $\theta_n$ during the first $100$ iterations. Observe that the sequence initially oscillates, and then stabilises close to $\theta^{\star}$ after approximately $50$ iterations. \Cref{fig:num-res-bayes-logit-theta}(b) presents the iterates $\theta_n$ for $n = 10^5,\ldots,10^6$. For completeness,  \Cref{fig:num-res-bayes-logit-hist} shows the histograms corresponding to the marginal posteriors $p(\beta_j|y,v,\hat{\theta}_N)$, for $j = 1,\ldots, 10$, obtained as a by-product of \Cref{alg:souk}.
\begin{figure}[h!]
	\centering
	\includegraphics[clip, trim=3.4cm .5cm 2.6cm 0cm, width=   
	\textwidth]{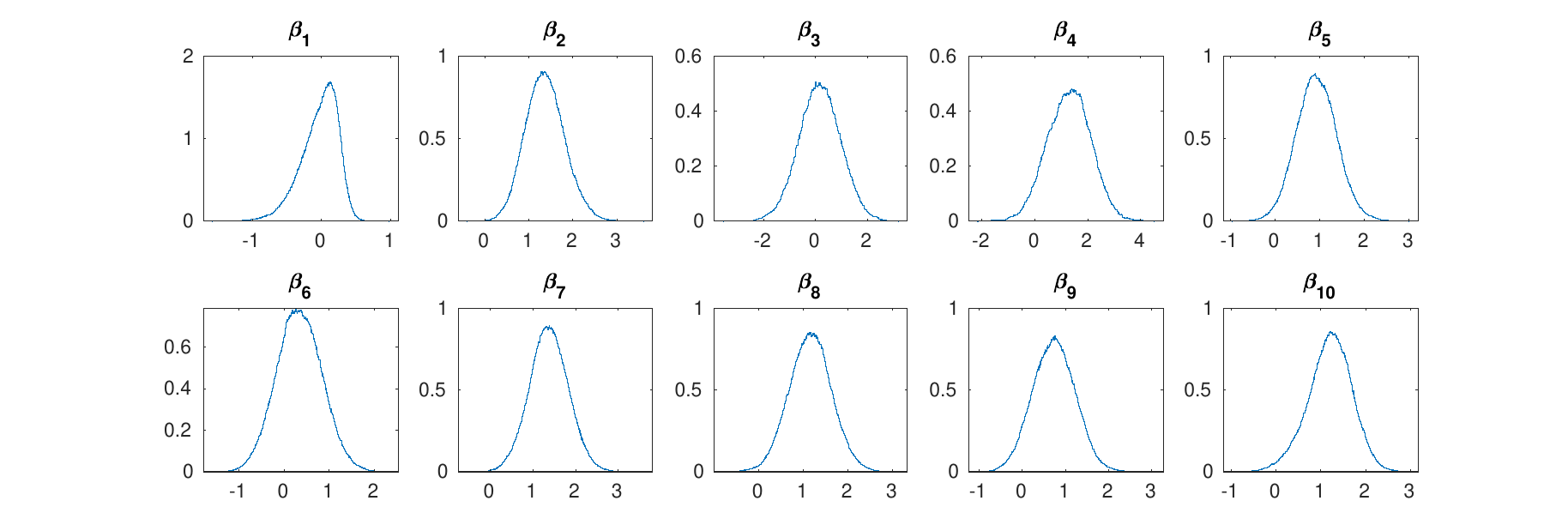}
	\caption{\small Bayesian logistic regression - Normalised histograms of each component of $ \beta $ obtained with $2 \times 10^6$ Monte Carlo samples.} \label{fig:num-res-bayes-logit-hist}
\end{figure}
In order to verify that the obtained estimate $\hat{\theta}_N$ is close to the true MLE $\theta^{\star}$ we use a truncated harmonic mean estimator (THME) \citep{robert2009computational} to calculate the marginal likelihood $p(y|\theta)$ for a range of values of $\theta$. Although obtaining the THME is usually computationally expensive, it is viable in this particular experiment as $\beta$ is low-dimensional. More precisely, given $n$ samples $(\beta_i)_{i \in \{1,\ldots,n\}}$ from $p(\beta|y,\theta)$, we obtain an approximation of $p(y|\theta)$ by computing
\begin{equation}
\hat{p}(y|\theta)=\left. n\vol(\msa)\middle/ \left( \sum_{k=1}^{n}\frac{\1_{\msa}(\beta_k)}{p(\beta_k,y|\theta)}\right) \right. \eqsp,
\label{EQ: bayes-logit-HME}
\end{equation}
where $\msa$ is a $d$-dimensional ball centered at the posterior mean $\bar{\beta}=n^{-1}\sum_{k=1}^{n}\beta_k$, and with radius set such that ${n^{-1} \sum_{i=1}^{n} \1_{\msa}(\beta_i) \approx 0.4}$. Using $n=6 \times 10^5$ samples, we obtain the approximation shown in  \Cref{fig:num-res-bayes-logit-hme-likeli}(a), where in addition to the estimated points we also display a quadratic fit (corresponding to a Gaussian fit in linear scale), which we use to obtain an estimate of $\theta^{\star}$ (the obtained log-likelihood values are small because the dataset is large ($d_y = 683$)).

To empirically study the estimation error involved, we replicate the experiment $10^3$ times.  \Cref{fig:num-res-bayes-logit-hme-likeli} shows the obtained histogram of $\{ \hat{\theta}_{N,i} \}_{i=1}^{1000} $, where we observe that all these estimators are very close to the true maximiser $\theta^{\star}$. Besides, note that the distribution of the estimation error is close to a Gaussian distribution, as expected for a maximum likelihood estimator. Also, there is a small estimation bias of the order of $3\%$, which can be attributed to the discretization error of SDE \eqref{eq:langevin}, and potentially to a small error in the estimation of $\theta^{\star}$.
\begin{figure}[h!]
		\centering
		\centerline{\includegraphics[width=\textwidth]{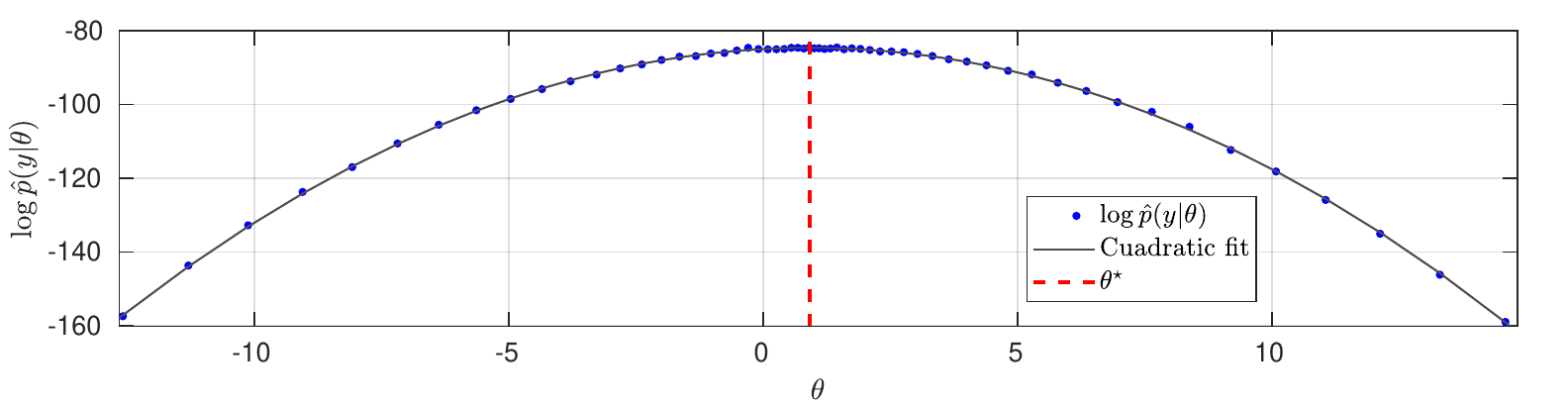}}	
		\centerline{(a)}
		\centerline{	\includegraphics[width=0.6\textwidth]{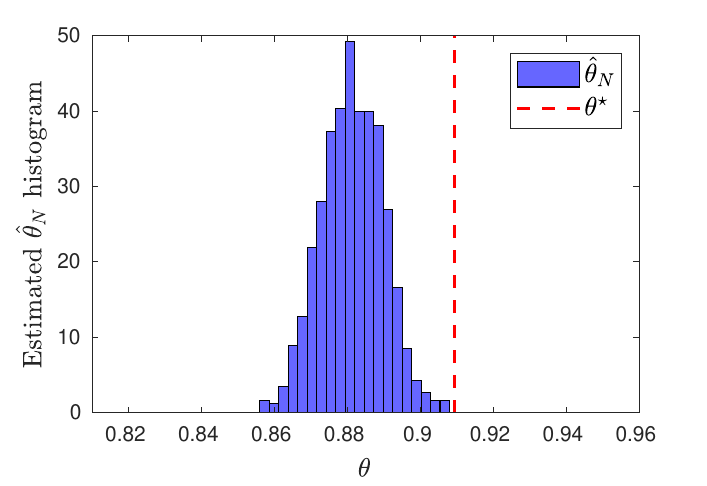}}
		\centerline{(b)}
	\caption{\small Bayesian logistic regression - (a) Estimated points of the marginal log-likelihood $\log \hat{p}(y|\theta)$ with quadratic fit (corresponding to a Gaussian fit in linear scale). (b) Normalised histogram of $\hat{\theta}_N$  for 1000 repetitions of the experiment. An estimate of $\theta^{\star}$, the maximiser of $\hat{p}(y|\theta)$, is plotted as a reference.} \label{fig:num-res-bayes-logit-hme-likeli}
\end{figure}
We conclude this experiment by using SOUL to perform a predictive empirical Bayesian analysis on the binary responses. We split the original dataset into an $80\%$ training set $(y^{\train},v^{\train})$ of size $d_{\train}=546$, and a $20\%$ test set $(y^{\test}, v^{\test})$ of size $d_{\test}=137$, and use SOUL to draw samples from the predictive distribution $p(y^{\test} | y^{\train},v^{\train},v^{\test}, \hat{\theta}_N)$. More precisely, we use SOUL to simultaneously calculate $\hat{\theta}_N$ and simulate from $p(\beta| y^{\train},v^{\train}, \hat{\theta}_N)$, followed by simulation from $p(y^{\test}|\beta, y^{\train}, v^{\train}, v^{\test})$.
We then estimate the maximum-a-posteriori predictive response $\hat{y}^{\test}$, and measure prediction accuracy against the test dataset by computing the error 
 \begin{equation}
 \epsilon= \|y^{\test}-\hat{y}^{\test}\|_1 /d_{\test} = \sum_{i=1}^{d_{\test}}   \abs{y^{\test}_i-\hat{y}^{\test}_i}/d_{\test}\eqsp,
 \label{EQ: bayes-logit-error}
 \end{equation}
 and obtain $\epsilon = 2.2\%$. For comparison, Figure \ref{fig:num-res-bayes-logit-error} below reports the error $\epsilon$ as a function of $\theta$ (the discontinuities arise because of the highly non-linear nature of the model). Observe that the estimated $\hat{\theta}_N$ produces a model that has a very good performance in this regard. 
\begin{figure}[h!]
	\centering
	\includegraphics[width=0.5\textwidth]{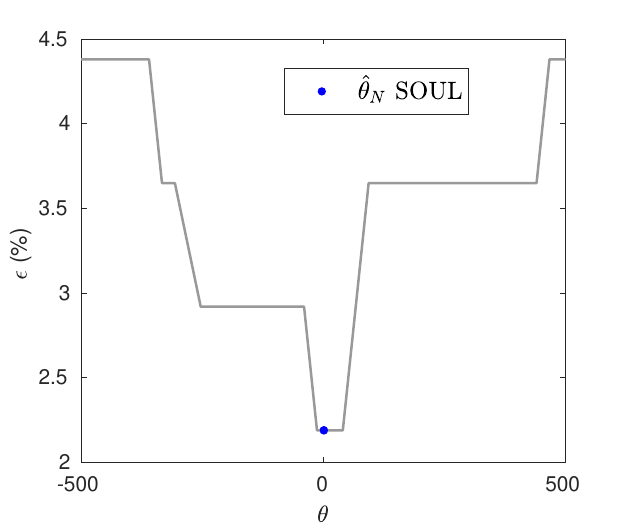}
	\caption{\small Bayesian logistic regression - Percentage of mislabelled binary observations in terms of $ \theta$. In blue we show the value of $\hat{\theta}_N$ obtained with Algo. \ref{alg:souk}.} \label{fig:num-res-bayes-logit-error}
\end{figure}

\subsection{Statistical audio compression}
\label{ssec:num-res-cs-audio}
Compressive sensing techniques exploit sparsity properties in the data to estimate signals from fewer samples than required by the Nyquist–Shannon sampling theorem \citep{candes2006compressive, candes2008introduction}. Many real-world data admit a sparse representation on some basis or dictionary.
Formally, consider an $\ell$-dimensional time-discrete signal $  z \in \mathbb{R}^{\ell} $  that is sparse in some dictionary $ \Psimatrix \in \mathbb{R}^{\ell \times d} $, i.e, there exists a latent vector $ x \in \mathbb{R}^d $ such that $  z = \Psimatrix x$ and $ \|x\|_0  = \sum_{i=1}^d \1_{\rset^*}(x_i) \ll \ell $. This prior assumption can be modelled by using a smoothed-Laplace distribution \citep{lingala2012blind}
\begin{equation}
p(x|\theta) \propto \exp\parenthese{-\theta \sum_{i=1}^{d} h_\lambda(x_i)} \eqsp ,
\label{EQ: num-res-audio-prior}
\end{equation}
where $h_\lambda$ is the Huber function given for any $u \in \rset$ by 
\begin{equation}
h_\lambda(u)=\begin{cases}
u^{2}/2 & \text{ if $|u|\leq\lambda$ } \eqsp,\\
\lambda(|u|-{\lambda/2}) & \text{ otherwise} \eqsp.
\end{cases} 
\label{EQ: huber-function}
\end{equation}
Acquiring $  z $ directly would call for measuring $ \ell $ univariate components. Instead, a carefully designed measurement matrix $\measMat \in \mathbb{R}^{p \times \ell}$, with $ p \ll \ell $, is used to directly observe a “compressed” signal $\measMat  z $, which only
requires taking $ p $ measurements. In addition, measurements are typically noisy which results in an observation $y \in \rset^p$  modeled as $ \y = \measMat  z + w$ where we assume that the noise $w$ has distribution $ \loiGauss(0,\sigma^2\Idd_p)$,  and therefore the likelihood function is given by
\begin{equation}
p(\y| x) \propto \exp\parenthese{-\norm{\y- \measMat\Psimatrix x}^2_2/(2\sigma^{2})} \eqsp ,
\label{EQ: num-res-audio-likelihood}
\end{equation}
leading to the posterior distribution
\begin{equation}
p( x|\y) \propto \textrm{exp}\parenthese{-\norm{\y- \measMat\Psimatrix x}^2_2/(2\sigma^{2}) -\theta \sum_{i=1}^{d} h_\lambda(x_i)} \eqsp .
\label{EQ: num-res-audio-posterior}
\end{equation}
To recover $ z$ from $\y$, we then compute the maximum-a-posteriori estimate
\begin{equation}
\hat{ x}_{\MAP} \in \underset{ x \in \mathbb{R}^d}{\mathrm{argmin}}~\defEns{ \norm{\y- \measMat\Psimatrix x}^2_2/2\sigma^{2}  + \theta ~ \sum_{i=1}^{d} h_\lambda(x_i)} \eqsp ,
\label{EQ: mapEstim}
\end{equation}
and set $\hat{ z}_{\MAP}=\Psimatrix \hat{ x}_{\MAP}$.

Following decades of active research, there are now many convex optimisation algorithms that can be used to efficiently solve \eqref{EQ: mapEstim}, even when $d$ is very large \citep{chambolle2016introduction,monga2017handbookconvexopt}. However, the selection of the value of $\theta$ in \eqref{EQ: mapEstim} remains a difficult open problem. This parameter controls the degree of sparsity of $x$ and has a strong impact on estimation performance.  

A common heuristic within the compressive sensing community is to set $\theta_{\mathrm{cs}}=0.1 \times\norm{(\measMat\Psimatrix)^\intercal y}_{\infty}/\sigma^2$, where for any $z \in \rset^\ell$, $\norm{z}_{\infty} = \max_{i\in\{1,\ldots,\ell\}} \abs{z_i}$, as suggested in \cite{kim2007method} and \cite{figueiredo2007gradient}; however, better results can arguably be obtained by adopting a statistical approach to estimate $\theta$.

The Bayesian framework offers several strategies for estimating $\theta$ from the observation $\y$. In this experiment we adopt an  empirical Bayesian approach and use SOUL to compute the MLE $\theta^{\star}$, which is challenging given the high-dimensionality of the latent space. 

To illustrate this approach, we consider the audio  experiment proposed in \cite{balzano2010compressed} for the \textit{``Mary had a little lamb''} song. The MIDI-generated audio file  $ z$ has $ \ell = 319,725 $ samples, but we only have access to a noisy observation vector $y$ with $ p = 456 $ random time points of the audio signal, corrupted by additive white Gaussian noise with $\sigma =0.015$.  
The latent signal $x$ has dimension $ d = 2,900 $
and is related to $ z$ by a dictionary matrix $\Psimatrix$ whose row vectors correspond to different piano notes lasting a quarter-second long \footnote{Each quarter-second sound can have one of 100 possible frequencies and be in 29 different positions in time. }. The parameter $\lambda$ for the prior \eqref{EQ: num-res-audio-prior} is set to $\lambda=4\times 10^{-5}$.  We used the heuristic $\theta_{\mathrm{cs}}$ as the initial value for $\theta$ in our algorithm. To solve the optimisation problem \eqref{EQ: mapEstim} we use the Gradient Projection for Sparse Reconstruction (GPSR) algorithm proposed in \cite{figueiredo2007gradient}. We use this solver because it is the one used in the online MATLAB demonstration of \cite{balzano2010compressed}, however, more modern algorithms could be used as well. We implemented \Cref{alg:souk} using a fixed  stepsize $\gamma_n=6.9 \times 10^{-6}$, a fixed batch size $m_n=1$, $\delta_n=20 \, n^{-0.8}/d=0.0069\, n^{-0.8}$ and 100 burn-in iterations.

The algorithm converged in approximately 500 iterations, which were computed in only 325 milliseconds.  \Cref{fig:num-res-audio} (left), shows the first 250 iterations of the sequence $\theta_n$ and of the weighted average $\hat{\theta}_{n}$. Again, observe that the iterates oscillate for a few iterations and then quickly stabilise. Finally, to assess the quality of the estimate $\hat{\theta}_N$, \Cref{fig:num-res-audio} (right) presents the reconstruction mean squared error as a function of $\theta$. The error is measured with respect to the reconstructed signal and is given by $\mathrm{MSE} (\hat{ x}_{\MAP}) = \|z^\star-\Psimatrix \hat{x}_{\MAP}\|_2^2/\ell$, where $z^\star$ is the true audio signal. Observe that the estimated value $\hat{\theta}_N$ is very close to the value that minimises the estimation error, and significantly outperforms the heuristic value $\theta_{cs}$ commonly used by practitioners.

\begin{figure}[h!]
	\begin{minipage}[l1]{.5\linewidth}
		\centering
		\centerline{\includegraphics[width=\textwidth]{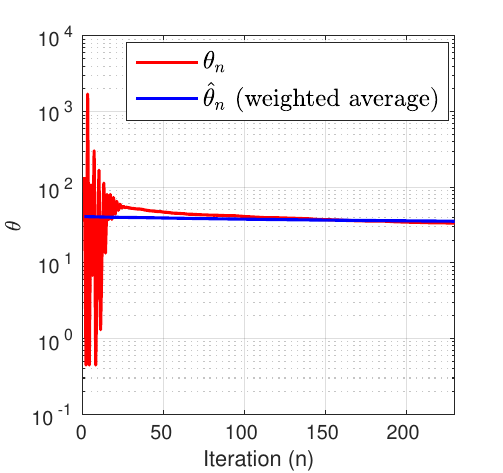}}
	\end{minipage}
	\begin{minipage}[l2]{.48\linewidth}
		\centering
		\centerline{\includegraphics[width=0.9\textwidth]{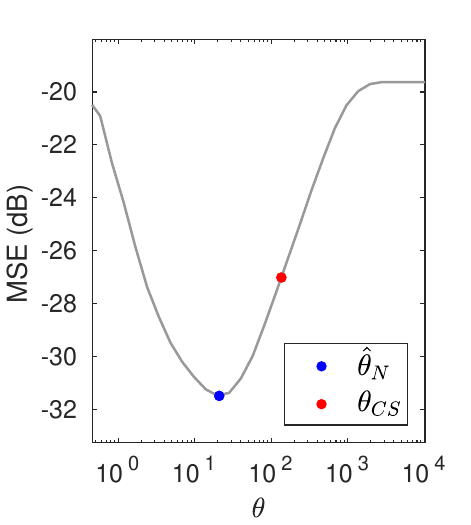}}
		
	\end{minipage}
	\caption{\small Statistical audio compression - Evolution of the the iterate $\theta_n$ and $\hat{\theta}_n$ with $\sigma=0.015$ in log scale (left). Reconstruction mean squared error (MSE) in dB as a function of the $\theta$ (right).} \label{fig:num-res-audio}
\end{figure}
        

\subsection{Sparse Bayesian logistic regression with random effects}
\label{ssec:num-res-logit-rand-effects}

Following on from the Bayesian logistic regression in Section \ref{ssec:num-res-logit}, where $p(y|\theta)$ is log-concave and hence $\theta^{\star}$ unique, we now consider a significantly more challenging sparse Bayesian logistic regression with random effects problem. In this experiment $p(y|\theta)$ is no longer log-concave, so SOUL can potentially get trapped in local maximisers. Furthermore, the dimension of $\theta$ in this experiment is very large ($d_\theta = 1001$), making the MLE problem even more challenging. This experiment was previously considered by \cite{atchade2017perturbed} and we replicate their setup. 

Let $\{y_i\}^{d_y}_{i=1} \in \{0,1\}$ be a vector of binary responses which can be modelled as ${d_y}$ conditionally independent realisations of a random effect logistic regression model, 
\begin{equation}
y_i|x \sim \textrm{Ber}\left( s(v_i^{\transpose} \beta+\upsigma z_i^{\transpose} x)\right) \eqsp,  \quad i \in \{1,\ldots,d_y\} \eqsp ,
\label{EQ: rand-ef-model}
\end{equation}
where $ v_i\in \rset^p $ are the covariates, $\beta \in \rset^p$ is the regression vector, $z_i \in \rset^d$ are (known) loading vectors, $x$ are random effects and $\upsigma >0$. In addition, recall that $\Ber(\alpha)$ denotes the Bernoulli distribution with parameter $\alpha \in \ccint{0, 1}$ and $s(u) = \rme^u /(1 + \rme^u )$ is the cumulative distribution function of the standard logistic distribution. The goal is to estimate the unknown parameters $\theta=(\beta, \upsigma) \in \rset^p \times \ooint{0,\plusinfty}$ directly from $\{y_i\}^{d_y}_{i=1}$, without knowing the value of $x$, which we assume to follow a standard Gaussian distribution, \ie~$p(x)=\exp \{-\norm{x}^2_2/2\}/(2\pi)^{d/2}$. We estimate $\theta$ by MLE using \Cref{alg:souk} to maximize \eqref{eq:def_theta_star}, with marginal likelihood given by
\begin{equation}
p(y|\theta)= \int_{\rset^d} \prod_{i=1}^{{d_y}} s(v_i^{\transpose} \beta+\upsigma z_i^{\transpose} x)^{y_{i}}(1-s(v_i^{\transpose} \beta+\upsigma z_i^{\transpose} x))^{1-y_{i}}p(x)\rmd x \eqsp ,
\label{EQ: rand-ef-marg-likelihood}
\end{equation}
and we use the penalty function
\begin{equation}
g(\theta)= \sum_{j=1}^{d} h_\lambda(\beta_j) \eqsp ,
\label{EQ: rand-ef-g-huber-loss}
\end{equation}
where $h_\lambda$ is the Huber function defined in \eqref{EQ: huber-function}.

We follow the procedure described in \cite{atchade2017perturbed} to generate the observations $\{y_i\}^{{d_y}}_{i=1}$, with ${d_y} = 500 $, $p = 1000$ and $d = 5$\footnote{\scriptsize We renamed some symbols for notation consistency. What we denote by $v_i$, $x$, ${d_y}$ and $d$, is denoted in \cite{atchade2017perturbed} by $x_i$, $\mathbf{U}$, $N$ and $q$ respectively.}. The vector of regressors $\beta_{\true}$  is generated from the uniform distribution on $ [1, 5] $ and $ 98\% $ of its coefficients are randomly set to zero. The variance $\upsigma_{\true}$ of the random effect is set to 0.1, and the projection interval for the estimated $\upsigma$ is $[10^{-5},+\infty)$. Finally, the parameter $\lambda$ in
 \eqref{EQ: rand-ef-g-huber-loss} is set to $\lambda=30$. We emphasize at this point that $\theta$ is high-dimensional in this experiment ($d_\Theta=1001$), making the estimation problem particularly challenging.

The conditional log-likelihood function for this model is
\begin{equation}
\log p(y|x,\theta)=\sum_{i=1}^{{d_y}}\left\{ y_{i}(v_i^{\transpose} \beta+\upsigma z_i^{\transpose} x)-\log(1+\rme^{v_i^{\transpose} \beta+\upsigma z_i^{\transpose} x})\right\} \eqsp .
\label{EQ: rand-ef-likelihood}
\end{equation}
To implement \Cref{alg:souk} we use the gradients 
\begin{align}
\nabla_{x}\log p(x|y,\theta)&= \sum_{i=1}^{{d_y}}\left\{ \upsigma z_{i}(y_{i}-s(v_{i}^{\transpose}\beta+\upsigma z_{i}^{\transpose}x))\right\} -x \eqsp ,
\label{EQ: rand-ef-grad-x-ULA}\\
  \nabla_{\theta} \log p(x ,y |\theta)&=\sum_{i=1}^{{d_y}}\left\{ (y_i -s(v_{i}^{\transpose}\beta+\upsigma z^{\transpose}_{i}x))   \begin{bmatrix}
      v_i \\
z_i^{\transpose} x
\end{bmatrix} \right\} \eqsp .
\label{EQ: rand-ef-grad-gradTheta}
\end{align}
Finally the gradient of the penalty function is given by
\begin{equation}
\frac{\partial}{\partial{\beta_i}} g(\theta)=\begin{cases}
{\beta_i} & |\beta_i|\leq\lambda\\
\lambda \, \sign(\beta_i), & |\beta_i|>\lambda
\end{cases} \eqsp , \qquad \frac{\partial}{\partial{\upsigma}} g(\theta) = 0 \eqsp ,
\label{EQ: rand-ef-grad-proxThetaG}
\end{equation}
where $\sign$ denotes the sign function, \ie~for any $s \in \Rset$, $\sign(s) = |s|/s$ if $s \neq 0$, and $\sign(s) = 0$ otherwise.

We use $\gamma_n=0.01$, $\delta_n=n^{-0.95}/d=0.2 \times n^{-0.95}$, a fixed batch size $m_n=1$, $\beta_0=\Unbf_p$ and $\upsigma_0=1$ as initial values. Moreover, we perform $10^4$ burn-in iterations with a fixed value of $\theta_0=(\beta_0,\upsigma_0)$ to warm-up the Markov chain, and further $600$ iterations of \Cref{alg:souk} to warm-start the iterates. Following on from this, we run $N=5\times10^4$ iterations of \Cref{alg:souk} to compute $\hat{\theta}_N$.  Computing this estimates required $25$ seconds in total.

Figure \ref{fig:num-res-logit-atchade-iterates} shows the evolution of the iterates throughout iterations, where we used $\|\hat{\beta}_n\|_0$ as a summary statistic to track the number of active components. Because the Huber penalty \eqref{EQ: huber-function} does not enforce exact sparsity on $\beta$, to estimate the number of active components we only consider values that are larger than a threshold $\tau$ (we used $\tau=0.005$).

\begin{figure}[h!]
	\begin{minipage}[l1]{.48\linewidth}
		\centering
		\centerline{\includegraphics[width=\textwidth]{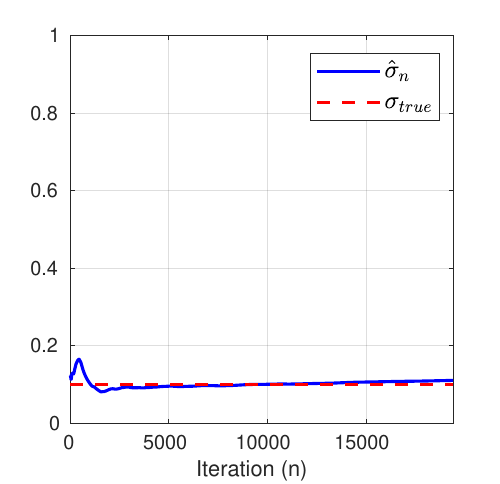}}
	\end{minipage}
	\begin{minipage}[l2]{.48\linewidth}
		\centering
		\centerline{\includegraphics[width=\textwidth]{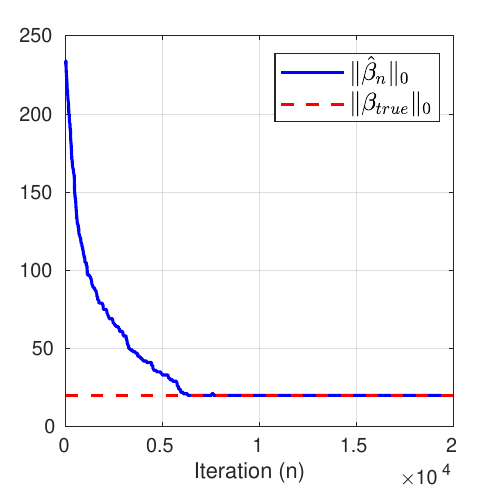}}
		
	\end{minipage}
	\caption{\small Sparse Bayesian logistic regression with random effects - Evolution of the  $\normLigne{\hat{\beta}_n}_0$ and of the iterate $ \hat{\upsigma}_n$ for the proposed method. The true values are plotted in red as a reference.} \label{fig:num-res-logit-atchade-iterates}
\end{figure}

From Figure \ref{fig:num-res-logit-atchade-iterates} we observe that  $\hat{\upsigma}_n$ converges to a value that is very close to $\upsigma_{\true}$, and that the number of active components is also accurately estimated. Moreover, Figure \ref{fig:num-res-logit-atchade-betaSupport} shows that most active components were correctly identified.  We also observe that $\hat{\beta}_n$ stabilizes after approximately $6300$ iterations, which correspond to $6300$ Monte Carlo samples as $m_n$=1. This is in close agreement with the results presented in  \cite[Figure 5]{atchade2017perturbed}, where they observe stabilization after a similar number of iterations of their highly specialised Polya-Gamma sampler. 
 
\begin{figure}[h!]
	\centering
	\includegraphics[clip, trim=2cm 0.1cm 1.5cm 0.1cm, width=   
	\textwidth]{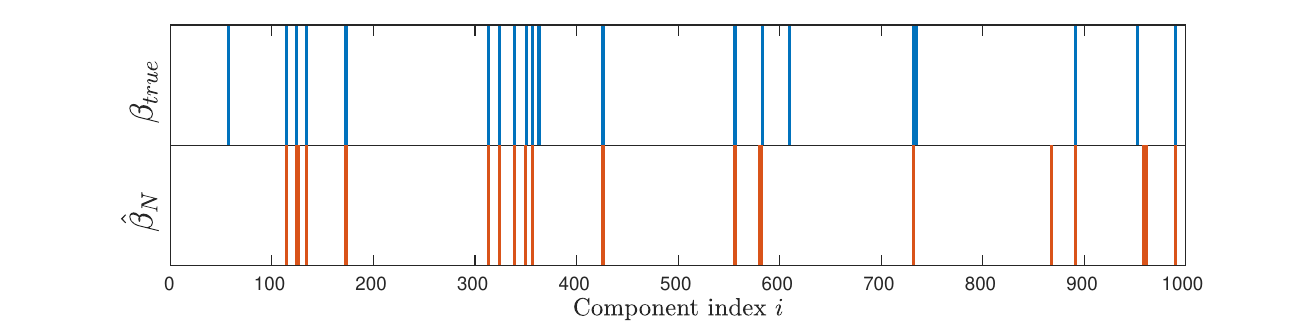}
	\caption{\small Sparse Bayesian logistic regression with random effects - Support of the estimated $\hat{\beta}_N$ compared with the support of $\beta_{true}$.} \label{fig:num-res-logit-atchade-betaSupport}
\end{figure}

It is worth emphasising at this point that \cite{atchade2017perturbed} considers the non-smooth penalty $g(\theta)=\lambda \|\beta\|_1$ instead of \eqref{EQ: rand-ef-g-huber-loss}. Consequently, instead of using the gradient of $g$, they resort to the so-called proximal operator of $g$ \cite{chambolle2016introduction}.
The generalisation of the SOUL methodology proposed in this paper to models that have non-differentiable terms is addressed in \citet{vidal2018maximum, vidal:debortoli:pereyra:durmus:2019}.


\section{Theoretical convergence analysis for SOUL, and generalisation to other inexact MCMC kernels (SOUK)}
\label{sec:convergence-results}
In this section we state our main theoretical results for SOUL. For completeness, we first present the results in a general stochastic optimisation setting and by considering a generic inexact MCMC sampler, and then show that our results apply to the specific MLE optimisation problem \eqref{eq:def_theta_star}, and to the specific Langevin algorithm \eqref{eq:euler_maruyama_langevin} used in SOUL.

\subsection{Notations and convention}
Denote by $\mathcal{B}(\rset^d)$ the Borel $\sigma$-field of
$\rset^d$, $\functionspace[]{\rset^d}$ the set of all Borel measurable
functions on $\rset^d$ and for $f \in \functionspace[]{\rset^d}$,
$\Vnorm[\infty]{f}= \sup_{x \in \rset^d} \abs{f(x)}$.  For $\mu$ a probability measure
on $(\rset^d, \mathcal{B}(\rset^d))$ and $f \in
\functionspace[]{\rset^d}$ a $\mu$-integrable function, denote by
$\mu(f)$ the integral of $f$ \wrt~$\mu$. For $f \in \functionspace[]{\rset^d}$, the $V$-norm of $f$ is given by $\Vnorm[V]{f}= \sup_{x \in \rset^d} |f(x)|/V(x)$. Let $\xi$ be a finite signed measure on $(\rset^d,\mcbb(\rset^d))$. The $V$-total variation distance of $\xi$ is defined as
\begin{equation}
\Vnorm[V]{\xi} = \sup_{f \in \functionspace[]{\rset^d}, \Vnorm[V]{f} \leq 1}  \abs{\int_{\rset^d } f(x) \rmd \xi (x)} \eqsp.
\end{equation}
If $V \equiv 1$, then $\Vnorm[V]{\cdot}$ is the total variation  denoted by $\tvnorm{\cdot}$. Let $\mu$ be a finite signed measure, then by the Hahn-Jordan theorem \cite[Theorem D.1.3]{douc:moulines:priouret:soulier:2018}, there exists a pair of finite singular measures $\mu^{+}, \mu^{-}$ such that $\mu = \mu^{+} - \mu^{-}$. The total variation measure $\abs{\mu}$ is given by $\abs{\mu} = \mu^+ + \mu^-$.

Let $\msu$ be an open set of $\rset^d$. We denote by
$\rmc^{k}(\msu, \rset^p)$
the set of $\rset^p$-valued $k$-differentiable functions, respectively
the set of compactly supported $\rset^p$-valued $k$-differentiable
functions. $\rmc^k(\msu)$  stands $\rmc^k(\msu,\rset)$.  Let $f : \msu \to \rset$, we denote by $\nabla f$,
the gradient of $f$ if it exists. $f$ is said to me $m$-convex with $m\geq 0$ if
for all $x,y \in \rset^d$ and $t \in \ccint{0,1}$,
\begin{equation}
f(t x + (1-t) y) \leq t f(x)  + (1-t) f(y) -(m/2)  \norm[2]{x-y}  \eqsp.
\end{equation}
We recall that if $f : \msu \to \rset$ is twice differentiable at point $a \in \rset^d$, its Laplacian is given by $\Delta f(a) = \sum_{i=1}^d \fraca{\partial^2 f}{\partial x_i^2}(a)$.
For any $\msa \subset \rset^d$, we denote by $\partial \msa$ the boundary of $\msa$. 
Let $(\Omega,\mcf,\PP)$ be a probability space. 
Denote by $\mu
\ll \nu$ if $\mu$ is absolutely continuous \wrt~$\nu$ and $\rmd \mu /
\rmd \nu$ an associated density. Let $\mu,\nu$ be two probability
measures on $(\rset^d, \mcbb(\rset^d))$. Define the Kullback-Leibler
divergence of $\mu$ from $\nu$ by
\begin{equation}
  \KL{\mu}{\nu} =
  \begin{cases}
    \int_{\rset^d} \frac{\rmd \mu}{\rmd \nu}(x) \log \parenthese{\frac{\rmd \mu}{\rmd \nu} (x)} \rmd \nu (x) \eqsp, & \text{if } \mu \ll \nu \\
\plusinfty & \text{ otherwise} \eqsp.
  \end{cases}
\end{equation}
The complement of a set $\msa \subset \rset^d$, is denoted by $\msa^{\complementary}$. We take the convention that $\prod_{k=p}^n = 1$ and $\sum_{k=p}^n =0$ for $n,p \in \nset$, $n< p$. All densities are w.r.t. the Lebesgue measure unless stated otherwise.


\subsection{Stochastic Optimization with inexact MCMC methods} \label{sec:stoch-optim-with}
We consider the problem of minimizing a function $f : \Theta \to \rset$ with $\Theta \subset \rset^{d_{\Theta}}$ under the following assumptions.

\begin{assumption}
  \label{assum:theta_compact}
  $\Theta$ is a convex compact set and $\Theta \subset \boulefermee{0}{\Rtheta}$ with $\Rtheta>0$.
\end{assumption}

\begin{assumption}
  \label{assum:f_diff}
  There exist an open set $\msu \subset \rset^{d_{\Theta}}$ and $\L \geq 0$ such that $\Theta \subset \msu$, $f \in \rmc^1(\msu,\R)$ and satisfies for any $\theta_1, \theta_2 \in \Theta$
  \begin{equation}
    \| \nabla f (\theta_1) - \nabla f (\theta_2) \| \leq \L \| \theta_1 - \theta_2 \| \eqsp .
  \end{equation}
  \end{assumption}

\begin{assumption}
  \label{assum:grad_expec}
  For any $\theta \in \Theta$, there exist $H_{\theta}: \ \rset^d \ \rightarrow \ \rset^{d_{\Theta}}$ and a probability distribution $\pi_{\theta}$ on $(\rset^d, \mcbb(\rset^d))$ satisfying that  
  $\pi_{\theta}(H_{\theta})< \plusinfty$ and for any $\theta \in \Theta$
  \begin{equation}
    \nabla f(\theta) = \int_{\rset^d} H_{\theta}(x) \rmd \pi_{\theta}(x) \eqsp .
    \end{equation}
  In addition, $(\theta, x) \mapsto H_{\theta}(x)$ is measurable.
\end{assumption}

\label{sec:stoch-optim-with-1}

Note that for the maximum marginal likelihood estimation problem \eqref{eq:def_theta_star}, $f$ corresponds to $\theta \mapsto -\log(p(y|\theta)) + g(\theta)$, for any $\theta \in \Theta$, $H_{\theta} : x \mapsto  \nabla_{\theta} \log(p(x,y | \theta))$ and $\pi_{\theta}$ is the probability distribution with density with respect to the Lebesgue measure $x \mapsto  p(x|y, \theta)$.

To minimize the objective function $f$ we suggest the use of a SA strategy which extends the one presented in \Cref{sec:sto_optim_Langevin}. 
More precisely, motivated by the methodology described in \Cref{sec:sto_optim_Langevin}, we propose a SA scheme which relies on biased estimates of $\nabla f(\theta)$ through a family of Markov kernels $\lbrace \Kker_{\gamma, \theta}, \gamma \in (0, \bgamma] \ \text{and} \ \theta \in \Theta \rbrace$, for $\bgamma >0$, such that for any $\theta \in \Theta$ and $\gamma \in (0,\bgamma]$, $\Kker_{\gamma, \theta}$ admits an invariant probability distribution $\pi_{\gamma,\theta}$ on $(\rset^d, \mathcal{B}(\rset^d))$. In the SOUL method, the Markov kernel $\Kker_{\gamma, \theta}$ stands for $\Rker_{\gamma, \theta}$ for any $\gamma \in (0, \bgamma]$  and $\theta \in \Theta$, where $\Rker_{\gamma, \theta}$ is the Markov kernel associated with \eqref{eq:euler_maruyama_langevin}. We assume in addition that the bias associated to the use of this family of Markov kernels can be controlled \wrt \ to $\gamma$ uniformly in $\theta$, \ie \ for example there exists $C>0$ such that for all $\gamma \in (0, \bgamma]$  and $\theta \in \Theta$, $\tvnorm{\pi_{\gamma, \theta} - \pi_{\theta}} \leq C \gamma^{\alpha}$ with $\alpha > 0$. 

Let now $(\delta_n)_{n \in \nset} \in (\rset_+^*)^{\nset}$ and $(m_n)_{n \in \N} \in \left(\nsets\right)^{\N}$ be sequences of stepsizes and batch sizes which will be used to define the sequence relatively to the variable $\theta$ similarly to \eqref{eq:theta_it} and \eqref{eq:delta_n}. Let $(\gamma_n)_{n \in \N} \in (\rset_+^*)^{\nset}$ be a sequence of stepsizes which will be used to get approximate samples from $\pi_{\theta_n}$, similarly to \eqref{eq:euler_maruyama_langevin}. Starting from $X_0^0 \in \rset^d$ and $\theta_0 \in \Theta$, we define on a probability space $(\Omega,\mcf,\mathbb{P})$, $(\lbrace X_k^n: k \in \lbrace 0, \dots, m_n \rbrace \rbrace, \theta_n)_{n \in \N}$ by the following recursion for $n \in \nset$ and $k \in \lbrace 0, \dots, m_n -1 \rbrace$
\begin{equation}
    \label{eq:algo_SOUL}
\begin{aligned}
   (X_{k}^n)_{ k \in \{0,\ldots,m_n\}} & \text{ is a MC with kernel } \Kker_{\gamma_n, \theta_n}    
  \text{ and } X_0^n = X_{m_{n-1}}^{n-1} \text{ given } \mathcal{F}_{n-1} \eqsp , \\
  \theta_{n+1} &= \Pi_{\Theta}\parentheseDeux{\theta_n - \frac{\delta_{n+1}}{m_{n}} \sum_{k=1}^{m_{n}} H_{\theta_n}(X_k^n)} \eqsp ,
\end{aligned}
\end{equation}
where $\Pi_{\Theta}$ is the projection onto $\Theta$ and  $\mathcal{F}_n$ is defined as follows for all $n \in \nset$
\begin{equation}
  \label{eq:def_F_n}
    \mathcal{F}_n = \sigma \left( \theta_0, \{(X_k^{\ell})_{k \in \{0,\ldots,m_\ell\}} \, : \, \ell \in \iint{0}{n}\} \right) \eqsp ,  \qquad \mcf_{-1} = \sigma(\theta_0, X_0^0)
  \end{equation}
  where $\{(X_k^{\ell})_{k \in \{0,\ldots,m_\ell\}} \, : \, \ell \in \iint{0}{n}\}$ is given by \eqref{eq:algo_SOUL}.
  Note that such a construction is always possible by the Kolmogorov extension theorem \cite[Theorem 5.16]{kallenberg2006foundations}, and  by \eqref{eq:algo_SOUL}, for any $n \in \nset$, $\theta_{n+1}$ is $\mcf_n$-measurable. Then the sequence of approximate minimizers of $f$ is given by $(\hat{\theta}_N)_{N \in \nset}$, \eqref{eq:theta_avg}. 

Under different sets of conditions on $f, H, (\delta_n)_{n \in \N}, (\gamma_n)_{n \in \N}$ and $(m_n)_{n \in \N}$ we obtain that
$(\theta_n)_{n \in \nset}$ converges \as \ to an element of $\argmin_{\Theta} f$. In particular in this section we consider the case where $f$ is assumed to be convex.
We establish that if $(\gamma_n)_{n \in \N}$ and $(\delta_n)_{n \in \N}$ go to $0$ sufficiently fast, $\expeLigne{f(\hat{\theta}_N)} - \min_{\Theta} f$ goes to $0$ with a quantitative rate of convergence. In the case where $(\gamma_n)_{n \in \N}$ is held fixed, \ie \ for all $n \in \N$, $\gamma_n = \gamma$, we show that while $\expeLigne{f(\hat{\theta}_N)}$ does not converge to $0$, there exists $C, \alpha >0$ such that $\limsup_{N \to +\infty} \expeLigne{f(\hat{\theta}_N)} - \min_{\Theta} f \leq C\gamma^{\alpha}$. In the case where $f$ is non-convex, we apply some results from stochastic approximation \cite{kushner2003stochastic}  which establish that the sequence $(\theta_n)_{n \in \N}$ converges \as \ to a stationary point of the projected ordinary differential equation associated with $\nabla f$ and $\Theta$. We postpone this result to \Cref{sec:non-convex-objective}, since it involves a theoretical background which we think is out of the scope of the main document.

 \subsection{Main results}
\label{sec:main-results}

We impose a stability condition on the stochastic process $\{(X_k^n)_{k \in \{0,\ldots,m_n\}} \, : \, n \in \nset\}$ defined by \eqref{eq:algo_SOUL} and that for any $\gamma \in \ocint{0, \bgamma}$ and $ \theta \in \Theta$ the iterates of $\Kker_{\gamma, \theta}$ are close enough to $\pi_{\theta}$ after a sufficiently large number of iterations.
\begin{assumptionH}
  \label{assum:condition_majo_V}
    There exists a measurable function $V : \rset^d \to \coint{1,\plusinfty}$ satisfying  the following conditions.
  \begin{enumerate}[label=(\roman*)]
  \item  \label{assum:condition_majo_V_i}  There exists    $A_1 \geq 1$ such that for any $n,p \in \N$, $k \in \{0, \dots, m_n \}$
       \begin{equation}
     \label{eq:condition_majo_V}
     \CPE{\Kker_{\gamma_n, \theta_n}^p V(X_{k}^n)}{X_0^0} \leq A_1 V(X_0^0) \eqsp , \qquad \expe{V(X_0^0)} < + \infty \eqsp , 
   \end{equation}
  where $\{(X_k^{\ell})_{k \in \{0,\ldots,m_\ell\}} \, : \, \ell \in \iint{0}{n}\}$  is given by \eqref{eq:algo_SOUL}.
 \item   \label{assum:condition_majo_V_ii}    There exist $A_2, A_3\geq 1$, $\rho \in \coint{0,1}$ such that for any $\gamma \in\ocint{0,\bgamma}$, $\theta \in \Theta$, $x \in \rset^d$ and  $n \in \N$, $\Kker_{\gamma,\theta}$ has a stationary distribution $\pi_{\gamma,\theta}$ and 
   \begin{equation}
     \label{eq:condition_V}
     \Vnorm{\updelta_x \Kker_{\gamma, \theta}^n - \pi_{\gamma, \theta}} \leq A_2 \rho^{n \gamma} V(x) \eqsp , \qquad \pi_{\gamma, \theta}(V) \leq A_3 \eqsp .
   \end{equation}  
 \item \label{assum:condition_majo_V_iii} There exists $\Psibf: \ \rset_+^{\star} \to \rset_+$ such that for any $\gamma \in \ocint{0, \bgamma}$ and $\theta \in \Theta$
   \begin{equation}
     \Vnorm[V^{1/2}]{\pi_{\gamma, \theta} - \pi_{\theta}} \leq \Psibf(\gamma) \eqsp .
   \end{equation}  
  \end{enumerate}
 \end{assumptionH}

 \Cref{assum:condition_majo_V}-\ref{assum:condition_majo_V_ii} is an ergodicity condition in $V$-norm for the Markov kernel $\Kker_{\gamma, \theta}$ uniform in $\theta \in \Theta$. There exists an extensive literature on the conditions under which a Markov kernel is ergodic \cite{meyn1993criteria_i,douc:moulines:priouret:soulier:2018}. \Cref{assum:condition_majo_V}-\ref{assum:condition_majo_V_iii} ensures that the distance between the invariant measure  $\pi_{\gamma, \theta}$ of the Markov kernel $\Kker_{\gamma, \theta}$ and $\pi_{\theta}$ can be controlled uniformly in $\theta$. We show that this condition holds in the case of the Langevin Monte Carlo algorithm in \Cref{propo:discrete_vs_continuous}. 

 We now state our mains results.

\begin{theorem}[Increasing batch size 1]
  \label{thm:salem_cv}
  Assume \tup{\Cref{assum:theta_compact}}, \tup{\Cref{assum:f_diff}}, \tup{\Cref{assum:grad_expec}} hold and $f$ is convex.
   Let $(\gamma_n)_{n \in \nset}$,
   $(\delta_n)_{n \in \nset}$ be sequences of non-increasing positive real numbers and $(m_n)_{n \in \nset}$ be sequences of positive integers satisfying 
   $\sup_{n \in \nset} \delta_n < 1/\L$, $\sup_{n \in \nset} \gamma_n < \bgamma$ and 
     \begin{equation}
    \label{eq:condition_cv}
    \sum_{n=0}^{+\infty} \delta_{n+1} = +\infty \eqsp , \qquad \sum_{n=0}^{+\infty} \delta_{n+1} \Psibf(\gamma_n) < +\infty \eqsp , \qquad \sum_{n=0}^{+\infty} \delta_{n+1} / (m_n \gamma_n) < +\infty \eqsp .
  \end{equation}
  Let $\{(X_k^n)_{k \in \{0,\ldots,m_n\}} \, : \, n \in \nset\}$ and $(\theta_n)_{n \in \nset}$ be given by \eqref{eq:algo_SOUL}.
  Assume in addition that \tup{\Cref{assum:condition_majo_V}} is satisfied and that for any $\theta \in \Theta$ and $x \in \rset^d$, $\norm{H_{\theta}(x)} \leq V^{1/2}(x)$. 
  Then, the following statements hold:
  \begin{enumerate}[label=(\alph*)]
  \item   $(\theta_n)_{n \in \N}$ 
    converges \as~to some $\theta^{\star} \in \argmin_{\Theta} f$ ;
  \item \label{thm:salem_cv_ii} furthermore, \as~there exists $C\geq 0$ such that for any $n \in \nsets$ 
  \begin{equation}
    \defEns{\left. \sum_{k=1}^n \delta_k f(\theta_k) \middle/ \sum_{k=1}^n \delta_k \right. } - \min_{\Theta} f \leq \left. C \middle/\left( \sum_{k=1}^n \delta_k \right) \right.  \eqsp.
  \end{equation}
  \end{enumerate}
\end{theorem}

\begin{proof}
  The proof is postponed to \Cref{thm:salem_cv_proof}.
\end{proof}

Note that in \eqref{eq:algo_SOUL}, $X_0^n = X_{m_{n-1}}^{n-1}$ for $n \in \nsets$. This procedure is referred to as  warm-start in the sequel. An inspection of the proof of \Cref{thm:salem_cv} shows that $X_0^n$ could be any random variable independent from $\mathcal{F}_{n-1}$ for any $n \in \N$ with $\sup_{n \in \nsets} \expe{V(X_0^n)} < +\infty$. It is not an option in the fixed batch size setting of  \Cref{thm:salem_cv_fix}, where the warm-start procedure is crucial for the convergence to occur.

We extend this theorem to non convex objective function see \Cref{thm:salem_cv_noncvx} in \Cref{sec:non-convex-objective}. Under the conditions of \Cref{thm:salem_cv} with the additional assumption that $\partial \Theta$ is a smooth manifold we obtain that $(\theta_n)_{n \in \nset}$ converges \as \ to some point $\theta^{*}$ such that $\nabla f(\theta^{*}) + \bfn =0$ with $\bfn = 0$ if $\theta^{*} \in \interior (\Theta)$ and $\bfn \in \mathrm{T}(\theta^{*}, \partial \Theta)^{\perp}$ if $\theta^{*} \in \partial \Theta$, where $\mathrm{T}(\theta, \partial \Theta)$ is the tangent space of $\partial \Theta$ at point $\theta \in \partial \Theta$, see \cite[Chapter 2]{aubin:2000}.

  In the case where $\Kker_{\gamma, \theta} = \Rker_{\gamma, \theta}$ is the Markov kernel associated with the Langevin update \eqref{eq:euler_maruyama_langevin}, under appropriate conditions \Cref{propo:discrete_vs_continuous} shows that for any $\gamma \in\ocint{0,\bgamma}$ with $\bgamma >0$, $\Psibf(\gamma) = \bigO(\gamma^{1/2})$. In that case, assume then that there exist $a,b,c >0$ such that for any $n \in \nsets$, $\delta_n = n^{-a}$, $\gamma_n = n^{-b}$ and $m_n = \ceil{n^c}$ then \eqref{eq:condition_cv} is equivalent to 
  \begin{equation}
    a < 1 \eqsp , \qquad a + b/2 > 1 \eqsp , \qquad a - b + c > 1 \eqsp . \label{eq:cv}
  \end{equation}
  Suppose $a \in \coint{0,1}$ is given, then the previous equation reads
  \begin{equation} b = 2(1-a) + \varsigma_1 \eqsp, \qquad c=3(1-a) + \varsigma_2 \eqsp, \qquad \varsigma_2 > \varsigma_1 >0 \eqsp . \label{eq:rates_n} \end{equation}
 This illustrates a trade-off between the intrinsic inaccuracy of our algorithm through the family of Markov kernels \eqref{eq:algo_SOUL} which do not exactly target $\pi_{\theta}$ and the minimization aim of our scheme.
  Note also that $(\delta_n)_{n \in \N}$ is allowed to be constant. This case yields $\gamma_n = n^{-2-\varsigma_1}$ and $m_n = \ceil{n^{3+\varsigma_2}}$ with $\varsigma_2 > \varsigma_1 >0$. 

  In our next result we derive an non-asymptotic upper-bound of $(\expeLigne{f(\bartheta_n) - \min_{\Theta} f})_{n \in \nset}$.
\begin{theorem}[Increasing batch size 2]
  \label{thm:salem_cv_control}
    Assume \tup{\Cref{assum:theta_compact}}, \tup{\Cref{assum:f_diff}}, \tup{\Cref{assum:grad_expec}} hold and $f$ is convex. Let  $(\gamma_n)_{n \in \nset}$,
   $(\delta_n)_{n \in \nset}$ be sequences of non-increasing positive real numbers and $(m_n)_{n \in \nset}$ be a sequence of positive integers satisfying 
   $\sup_{n \in \nset} \delta_n < 1/\L$, $\sup_{n \in \nset} \gamma_n < \bgamma$. 
  Let $\{(X_k^n)_{k \in \{0,\ldots,m_n\}} \, : \, n \in \nset\}$ be given by \eqref{eq:algo_SOUL}.
 Assume in addition that  \tup{\Cref{assum:condition_majo_V}} is satisfied and that for any $\theta \in \Theta$ and $x \in \rset^d$, $\norm{H_{\theta}(x)}\leq V^{1/2}(x)$. 
   Then, there exists $(E_n)_{n \in \N}$ such that for any $n \in \nsets$
    \begin{equation}
    \expe{  \defEns{\left. \sum_{k=1}^n \delta_k f(\theta_k) \middle/ \sum_{k=1}^n \delta_k \right. } - \min_{\Theta} f  }\leq  \left. E_n \middle/  \left( \sum_{k=1}^n \delta_k \right) \right. \eqsp ,
  \end{equation}
  with for any $n \in \nsets$,
  \begin{multline}
E_n =  2\Rtheta^2 + 2B_1 \Rtheta\expe{V^{1/2}(X_0^0)} \sum_{k=0}^{n-1} \delta_{k+1}/ (m_k \gamma_k) \\ + 2 \Rtheta \sum_{k=0}^{n-1} \delta_{k+1} \Psibf(\gamma_k) +  4B_1^2 \expe{V(X_0^0)} \sum_{k=0}^{n-1} \delta_{k+1}^2/ (m_k \gamma_k)^2   \\ +   4  \sum_{k=0}^{n-1} \delta_{k+1}^2 \Psibf(\gamma_k)^2  + B_2 \sum_{k=0}^{n-1} \delta_{k+1}^2 / (m_k \gamma_k)^2  \eqsp , \numberthis     \label{eq:horrible_bound}
\end{multline}
where $B_1$ and $B_2$ are given in \Cref{lemma:error_bound} and \Cref{lemma:error_variance} respectively.
\end{theorem}

\begin{proof}
  The proof is postponed to \Cref{thm:salem_cv_control_proof}.
\end{proof}

 We recall that in the case where $\Kker_{\gamma, \theta} = \Rker_{\gamma, \theta}$ is the Markov kernel associated with the Langevin update \eqref{eq:euler_maruyama_langevin}, under appropriate conditions \Cref{propo:discrete_vs_continuous} shows that for any $\gamma \in\ocint{0,\bgamma}$ with $\bgamma >0$, $\Psibf(\gamma) = \bigO(\gamma^{1/2})$. 
In that case, if there exist $a, b,c \geq 0$ such that for any $n \in \nsets$, $\delta_n = n^{-a}$, $\gamma_n = n^{-b}$, $m_n = n^c$ and \eqref{eq:cv} holds, 
the accuracy, respectively the complexity, of the algorithm are of orders  $\left(\sum_{k=1}^n \delta_k\right)^{-1} = \bigO(n^{a-1})$, respectively $\sum_{k=0}^n m_k = \bigO(n^{3(1-a) + \varsigma_2 + 1})$ for $\varsigma_2 >0$. Thus, for a fix target precision $\vareps >0$, it requires that $\varespilon = \bigO(n^{a-1})$ and  the complexity reads $ \bigO(\vareps^{-3} \parenthese{\log(1/\vareps)/(1-a)}^{1+ \varsigma_2})$. On the other hand, if we fix the complexity budget to $N$ the accuracy is of order $\bigO(N^{-(3 + (1 + \varsigma_2)/(1-a))^{-1}})$. These two considerations suggest to set $a$ close to $0$. In the special case where $a = 0$, we obtain that the accuracy is of order $\bigO(n^{-1})$, which is similar to the order identified in the deterministic gradient descent for convex functionals.

A case of interest is the fix stepsize setting, \ie \ for all $n \in \N, \ \gamma_n = \gamma >0$. Assume that $ (\delta_n)_{ n \in \N}$ is non-increasing $\lim_{n \to +\infty} \delta_n = 0$ and $\lim_{n \to +\infty} m_n = +\infty$. In addition, assume that $\sum_{n \in \nsets} \delta_n = +\infty$ then, by  \cite[Problem 80, Part I]{polya1998problem}, it holds that
\begin{equation}
  \begin{cases}
    &\lim_{n \to +\infty} \parentheseDeux{\left . \left(\sum_{k=1}^n \delta_{k} / m_k\right) \middle/ \left(\sum_{k=1}^n \delta_k\right) \right.} = \lim_{n \to +\infty} 1/m_n = 0  \eqsp ; \\ &
    \lim_{n \to +\infty} \parentheseDeux{\left. \left(\sum_{k=1}^n \delta_{k}^2\right) \middle / \left(\sum_{k=1}^n \delta_k\right) \right.}= \lim_{n \to +\infty} \delta_n = 0 \eqsp  .
    \end{cases}
\end{equation}
Therefore, we obtain that
 \begin{equation}
 \limsup_{n \to \plusinfty} \expe{ \defEns{\left . \sum_{k=1}^n \delta_k f(\theta_k) \middle/ \sum_{k=1}^n \delta_k \right.} - \min f} \leq  2 \Rtheta  \Psibf(\gamma)  \eqsp .
\end{equation}
Similarly, if the stepsize is fixed and the number of Markov chain iterates is fixed, \ie \ for all $n \in \nset$, $\gamma_n=\gamma$ and $m_n = m$ with $\gamma >0$ and $m \in \nsets$, we obtain that 
\begin{equation}
  \label{eq:borne_fix}
  \limsup_{n \to \plusinfty} \expe{ \defEns{\left . \sum_{k=1}^n \delta_k f(\theta_k) \middle/ \sum_{k=1}^n \delta_k \right.} - \min f} \leq  \Xibf_1(\gamma)  \eqsp ,
\end{equation}
with
  \begin{equation}
    \label{eq:U1}
    \Xibf_1(\gamma) = 2B_1\Rtheta \expe{V^{1/2}(X_0^0)} /\gamma + 2 \Rtheta  \Psibf(\gamma)  \eqsp .
  \end{equation}
However if $(m_n)_{n \in \N}$ is constant the convergence cannot be obtained using \Cref{thm:salem_cv}.
Strengthening the conditions of \Cref{thm:salem_cv} and making use of the warm-start property of the algorithm
  we can derive the convergence in that case. 


We now are interested in the case where the batch size is fixed, \ie \ $m_n = m_0$ for all $n \in \nset$. For ease of exposition we only consider $m_0 = 1$ and let $\tX_{n+1} = X_1^n$ for any $n \in \nset$. However the general case can be adapted from the proof of the result stated below. 
More precisely we consider the setting where the recursion \eqref{eq:algo_SOUL} can be written for any $n \in \nset$ as 
\begin{equation}
   \label{eq:algo_x_fix}  
\begin{aligned}
  & \X_{n+1} \text{ has distribution } \Kker_{\gamma_n, \ttheta_n}(\X_n, \cdot) \text{ conditionally to } \F_{n} \eqsp , \\
  &\ttheta_{n+1} = \Pi_{\Theta}\parentheseDeux{\ttheta_n - \delta_{n+1} H_{\ttheta_n}(\X_{n+1})} \eqsp , 
\end{aligned}
\end{equation}
with $\theta_0 \in \Theta$, $\tX_0 \in \rset^d$ and where  $\F_n$ is given by
\begin{equation}
  \label{eq:def_tmcf}
  \F_n = \sigma \left( \ttheta_0, (\X_{\ell})_{\ell \in \{0,\ldots,n\}} \right) \eqsp.
\end{equation}
We consider the following assumption on the family $\ensemble{H_{\theta}}{\theta \in \Theta}$.
\begin{assumption}
  \label{assum:H_lip}
 There exists $L_H \geq 0$ such that for any $x \in \rset^d$ and $\theta_1, \theta_2 \in \Theta$,
  \begin{equation}
    \| H_{\theta_1}(x) - H_{\theta_2}(x) \| \leq L_H \| \theta_1 - \theta_2\| V^{1/2}(x) \eqsp .
  \end{equation}
\end{assumption}
We consider a similar property as \Cref{assum:H_lip} on the family of Markov kernels $\defEns{ \Kker_{\gamma, \theta}, \gamma \in\ocint{0,\bgamma}, \theta \in \Theta }$, which weakens the assumption \cite[H6]{atchade2017perturbed}.
\begin{assumptionH}
  \label{assum:condition_kernel_fix}
  There exist a measurable function $V : \rset^d \to \coint{1,\plusinfty}$, 
    $\Lambdabf_1: \left( \rset_+^* \right)^2 \to \rset_+$ and $\Lambdabf_2: \left( \rset_+^* \right)^2 \to \rset_+$ such that for any $\gamma_1,\gamma_2 \in\ocint{0,\bgamma}$ with $\gamma_2 < \gamma_1$, $\theta_1,\theta_2 \in \Theta$, $x \in \rset^d$
  and $a \in \ccint{1/4, 1/2}$
        \begin{equation}     
     \Vnorm[V^{a}]{\updelta_x \Kker_{\gamma_1, \theta_1} - \updelta_x \Kker_{\gamma_2, \theta_2} } \leq \parentheseDeux{ \Lambdabf_1(\gamma_1, \gamma_2)  + \Lambdabf_2(\gamma_1, \gamma_2)\| \theta_1 - \theta_2 \| } V^{2a}(x) \eqsp .
   \end{equation}
\end{assumptionH}



The following theorem ensures convergence properties for $(\theta_n)_{n \in \nset}$ 
similar to the ones of \Cref{thm:salem_cv}.
The proof of this result is based on  a generalization of \cite[Lemma 4.2]{fort2011convergence} for inexact MCMC schemes.
\begin{theorem}[Fixed batch size 1]
  \label{thm:salem_cv_fix}
  Assume \tup{\Cref{assum:theta_compact}}, \tup{\Cref{assum:f_diff}}, \tup{\Cref{assum:grad_expec}}, \tup{\Cref{assum:H_lip}} hold and $f$ is convex.
  Let $\bgamma>0$, $(\gamma_n)_{n \in \nset}$ and
   $(\delta_n)_{n \in \nset}$ be sequences of non-increasing positive real numbers satisfying 
   $\sup_{n \in \nset} \delta_n < 1/\L$, $\sup_{n \in \nset} \gamma_n < \bgamma$,   $ \sup_{n \in \N} |\delta_{n+1} - \delta_n | \delta_n^{-2} < +\infty$,  $\sum_{n=0}^{+\infty} \delta_{n+1} = +\infty$ and
  \begin{equation}
    \label{eq:condition_cv_fix_1}
    \begin{aligned}
      & \sum_{n=0}^{+\infty} \delta_{n+1} \Psibf(\gamma_n) < +\infty \eqsp , \qquad \sum_{n=0}^{+\infty} \delta_{n+1}^2\gamma_n^{-2}  < +\infty\eqsp , \\
      & \sum_{n=0}^{+\infty} \delta_{n+1} \gamma_{n+1}^{-2} \parentheseDeux{ \Lambdabf_1(\gamma_n, \gamma_{n+1}) +  \delta_{n+1} \Lambdabf_2(\gamma_n, \gamma_{n+1})}  < +\infty\eqsp .
    \end{aligned}
  \end{equation}
  Let $(\X_n)_{n \in \N}$ be given by \eqref{eq:algo_x_fix}.
    Assume in addition that  \tup{\Cref{assum:condition_majo_V}} and 
      \tup{\Cref{assum:condition_kernel_fix}} are satisfied and that for any $\theta \in \Theta$ and $x \in \rset^d$, $\norm{H_{\theta}(x)} \leq V^{1/4}(x)$. 
  Then the following statements hold:
  \begin{enumerate}[label=(\alph*)]
  \item   $(\ttheta_n)_{n \in \N}$ 
    converges \as~to some $\theta^{\star} \in \argmin_{\Theta} f$ ;
  \item furthermore, \as~there exists $C\geq0$ such that for any $n \in \nsets$
  \begin{equation}
    \defEns{\left. \sum_{k=1}^n \delta_k f(\ttheta_k) \middle/ \sum_{k=1}^n \delta_k \right. } - \min_{\Theta} f \leq \left. C \middle/\left( \sum_{k=1}^n \delta_k \right) \right.  \eqsp.
  \end{equation}  
\end{enumerate}

\end{theorem}

\begin{proof}
  The proof is postponed to \Cref{thm:salem_cv_fix_proof}.
\end{proof}

In the case where $\Kker_{\gamma, \theta} = \Rker_{\gamma, \theta}$ is the Markov kernel associated with the Langevin update \eqref{eq:euler_maruyama_langevin}, under appropriate conditions \Cref{propo:discrete_vs_continuous} and \Cref{lem:v_kernel_error} show that for any $\gamma_1, \gamma_2 \in\ocint{0,\bgamma}$ with $\bgamma >0$ and $\gamma_1 > \gamma_2$, $\Psibf(\gamma_1) = \mathrm{C}_1 \gamma^{1/2}$, $\Lambdabf_1(\gamma_1, \gamma_2) =  \mathrm{C}_2(\gamma_1/\gamma_2-1)$ and $\Lambdabf_2(\gamma_1, \gamma_2) = \mathrm{C}_3 \gamma_2^{1/2}$, for $\mathrm{C}_1,\mathrm{C}_2,\mathrm{C}_3 \geq 0$. Thus we obtain that the following series should converge
\begin{equation}
  \begin{aligned}
    &\sum_{n=0}^{+\infty} \delta_{n+1}\gamma_n^{1/2} < +\infty \eqsp , \qquad  \sum_{n=0}^{+\infty} \delta_{n+1}^2 / \gamma_{n+1}^2  < +\infty \eqsp , \\
    & \sum_{n=0}^{+\infty} \delta_{n+1} (\gamma_{n} - \gamma_{n+1}) / \gamma_{n+1}^3  < +\infty \eqsp . \end{aligned} \label{eq:cv_fixed_1}
\end{equation}
If there exist $a, b >0$ such that $\delta_n = n^{-a}$ and $\gamma_n = n^{-b}$, then \eqref{eq:cv_fixed_1} is satisfied if $b \in \ooint{2(1-a), a-1/2}$ which is not empty if $a>5/6$.


\begin{theorem}[Fixed batch size 2]
  \label{thm:salem_cv_control_fix}
    Assume \tup{\Cref{assum:theta_compact}}, \tup{\Cref{assum:f_diff}}, \tup{\Cref{assum:grad_expec}}, \tup{\Cref{assum:H_lip}} hold and $f$ is convex. Let  $(\gamma_n)_{n \in \nset}$,
   $(\delta_n)_{n \in \nset}$ be sequences of non-increasing positive real numbers and $(m_n)_{n \in \nset}$ be a sequence of positive integers satisfying 
   $\sup_{n \in \nset} \delta_n < 1/\L$ and $\sup_{n \in \nset} \gamma_n < \bgamma$. 
 Let $(\X_n)_{n \in \N}$ be given by \eqref{eq:algo_x_fix}.
 Assume in addition that  \tup{\Cref{assum:condition_majo_V}} and \tup{\Cref{assum:condition_kernel_fix}} are satisfied and that for any $\theta \in \Theta$ and $x \in \rset^d$, $\norm{H_{\theta}(x)} \leq V^{1/4}(x)$. 
   Then, there exists $(\tE_n)_{n \in \N}$ such that for any $n \in \nsets$
    \begin{equation}
    \expe{  \defEns{\left. \sum_{k=1}^n \delta_k f(\theta_k) \middle/ \sum_{k=1}^n \delta_k \right. } - \min_{\Theta} f  }\leq  \left. \tE_n \middle/  \left( \sum_{k=1}^n \delta_k \right) \right. \eqsp ,
  \end{equation}
  with for any $n \in \nsets$,
  \begin{align}
    \label{eq:horrible_bound_not}
    \tE_n &= 2 \Rtheta +2\Rtheta \sum_{k=0}^n \delta_{k+1} \Psibf(\gamma_k) + C_3 \sum_{k=0}^n \abs{\delta_{k+1} - \delta_k}\gamma_k^{-1} \\ &+2 \Rtheta C_2 \sum_{k=0}^n \delta_{k+1} \gamma_{k+1}^{-1} \parentheseDeux{\gamma_{k+1}^{-1}\defEns{ \Lambdabf_1(\gamma_{k}, \gamma_{k+1})  + \Lambdabf_2(\gamma_{k}, \gamma_{k+1})\delta_{k+1}} + \delta_{k+1}}  \\ &+ C_3 \sum_{k=0}^n \delta_{k+1}^2\gamma_{k+1}^{-1} + C_3 (\delta_{n+1} / \gamma_n - \delta_0 / \gamma_0) + C_1 \sum_{k=0}^n \delta_{k+1}^2 \eqsp .
  \end{align}
where $C_1$, $C_2$ and $C_3$ are given in \Cref{lemma:cv_sq_fix}, \Cref{lemma:cv_norm_c_d} and \Cref{lemma:cv_norm_b} respectively.
\end{theorem}

\begin{proof}
  The proof is postponed to \Cref{thm:salem_cv_control_fix_proof}.
\end{proof}

\Cref{thm:salem_cv_control_fix} improves the conclusions of \Cref{thm:salem_cv_control} in the case where $\gamma_n = \gamma >0$ for any $n \in\nset$. Indeed, in that case, similarly to \eqref{eq:borne_fix}, assuming that $\lim_{n \to +\infty} \delta_n =0$, $\sup_{n \in \nset} \abs{\delta_{n+1} - \delta_n} \delta_n^{-2} < +\infty$, $\Lambdabf_1(t,t) = 0$ for any $t >0$, we obtain that 
for all $n \in \nset$
\begin{equation}  
\limsup_{n \to \plusinfty} \expe{ \defEns{\left . \sum_{k=1}^n \delta_k f(\theta_k) \middle/ \sum_{k=1}^n \delta_k \right.} - \min f} \leq \Xibf_2(\gamma) \eqsp ,
\end{equation}
with $\Xibf_2(\gamma) = 2 \Rtheta \Psi(\gamma) \leq \Xibf_1(\gamma) = 2B_1\Rtheta \expe{V^{1/2}(X_0^0)} /\gamma + 2 \Rtheta  \Psibf(\gamma)$.
In the case where $\sup_{\gamma \in \ocint{0, \bgamma}} \Psibf(\gamma) < +\infty$, $\Xibf_2(\gamma)$ is of order $\bigO(\Psibf(\gamma))$ and $\Xibf_1(\gamma)$ is of order $\bigO(\gamma^{-1})$. Therefore if $\lim_{\gamma \to 0} \Psibf(\gamma) = 0$, even in the fixed batch size setting, the minimum of the objective function $f$ can be approached with arbitrary precision $\vareps >0$ by choosing $\gamma$ small enough. 

\subsection{Application to SOUL}
\label{sec:soul-method}

We now apply our results to the SOUL methodology introduced in 
\Cref{sec:sto_optim_Langevin} where the Markov kernel $\Rker_{\gamma, \theta}$ with $\gamma \in \ocint{0, \bgamma}$ and $\theta \in \Theta$ is given by a Langevin Markov kernel and 
associated with recursion \eqref{eq:euler_maruyama_langevin}.
Setting for any $\theta \in \Theta$, $\pi_{\theta} = p(\cdot |y, \theta)$, we consider the following assumption on the family of probability distributions $(\pi_{\theta})_{\theta \in \Theta}$.
\begin{assumptionL}
  \label{assum:glip}
  For any $\theta \in \Theta$, there exists $U_{\theta}: \rset^d \to \rset$ such that $\pi_{\theta}$ admits a probability density function \wrt \ to the Lebesgue measure proportional to $x \mapsto \exp(-U_{\theta}(x))$.
  In addition $(\theta, x) \mapsto U_{\theta}(x)$ is continuous, $x \mapsto U_{\theta}(x)$ is differentiable for all $\theta \in \Theta$ and there exists $\LUx \geq 0$ such that for any $x,y \in \rset^d$,
  \begin{equation}
    \sup_{\theta \in \Theta}       \norm{\nabla_x U_{\theta}(x) -\nabla_x U_{\theta}(y)} \leq \LUx \norm{x-y} \eqsp ,
  \end{equation}
  and $\ensemble{\| \nabla_x U_{\theta}(0) \|}{\theta \in \Theta}$ is bounded.
\end{assumptionL}

In the case where $\Kker_{\gamma, \theta} = \Rker_{\gamma, \theta}$ for any $\gamma \in \ocint{0, \bgamma}$ and $\theta \in \Theta$, the first line of \eqref{eq:algo_SOUL} can be rewritten for any $n \in \N$ and $k \in \lbrace 0, \dots, m_n -1 \rbrace$
\begin{equation}
  X_{k+1}^n = X_k^n - \gamma_n \nabla_x U_{\theta_n}(X_k^n) + \sqrt{2\gamma_n} Z_{k+1}^n  \eqsp, \text{ with  $X_0^n = X_{m_{n-1}}^{n-1}$ if $n \geq 1$ }\eqsp,  \label{eq:langevin_discrete} 
\end{equation}
given $(\gamma_n)_{n \in \N} \in \ocint{0,\bgamma}^{\nset}$, $(m_n)_{n \in \N} \in \left(\nsets\right)^{\N}$ and $(Z_k^n)_{n \in \N, k \in \lbrace 1, \dots, m_n \rbrace}$ a family of i.i.d  $d$-dimensional zero-mean Gaussian random variables with covariance matrix identity. In the following propositions, we show that the results above hold by deriving sufficient conditions under which \tup{\Cref{assum:condition_majo_V}} and \tup{\Cref{assum:condition_kernel_fix}} are satisfied. 

Under \Cref{assum:glip}, the Langevin diffusion defined by \eqref{eq:langevin}
admits a unique strong solution for any $\theta \in \Theta$. Consider now the following additional tail condition on $U_{\theta}$ which ensures geometric ergodicity of $\Rker_{\gamma, \theta}$ for any  $\theta \in \Theta$ and $\gamma \in\ocint{0,\bgamma}$, with $\bgamma$ which will be specified below.
\begin{assumptionL}
  \label{assum:curvature}
  There exist $\mttdeux >0$ and $\mtttrois,\ccur, \Rdeux \geq 0$ such that  for any $\theta \in \Theta$ and $x \in \rset^d$,
  \[ \langle \nabla_xU_{\theta}(x), x \rangle \geq\mttdeux \| x \| \1_{\boule{0}{\Rdeux}^{\complementary}}(x) + \mtttrois \| \nabla_x U_{\theta} (x) \|^2 - \ccur\eqsp.\]
\end{assumptionL}

 \begin{assumptionL}
  \label{assum:glip_theta}
There exists $L_{U} \geq 0$ such that for any $x\in \rset^d$ and $\theta_1, \theta_2 \in \Theta$
  \begin{equation}
  \norm{\nabla_x U_{\theta_1}(x) -\nabla_x U_{\theta_2}(x)} \leq L_U \| \theta_1 - \theta_2 \| V(x)^{1/2} \eqsp.
  \end{equation}
\end{assumptionL}

The next theorems assert that under \Cref{assum:glip}, \Cref{assum:curvature} and \Cref{assum:glip_theta} the SOUL algorithm introduced in \Cref{sec:sto_optim_Langevin} satisfy \Cref{assum:condition_majo_V} and \Cref{assum:condition_kernel_fix} and therefore \Cref{thm:salem_cv}, \Cref{thm:salem_cv_control}, \Cref{thm:salem_cv_fix} and \Cref{thm:salem_cv_control_fix}
can be applied if in addition \Cref{assum:theta_compact}, \Cref{assum:f_diff}, \Cref{assum:grad_expec} and \Cref{assum:H_lip} hold.

 Under \tup{\Cref{assum:curvature}} define for any $x \in \rset^d$
\begin{equation}\Ve(x) = \exp\parentheseDeux{\mttdeux\sqrt{1 + \norm{x}^2}/4} \eqsp . \end{equation}

 \begin{theorem}
   \label{thm:salem_langevin_cv}
   Assume \tup{\Cref{assum:glip}} and \tup{\Cref{assum:curvature}}. 
Then
,  \tup{\Cref{assum:condition_majo_V}} holds with $V \leftarrow \Ve$, $\bgamma \leftarrow \min(1,2\mtttrois)$ and $\Psibf(\gamma) = D_4 \sqrt{\gamma}$ where $D_4$ is given in \Cref{propo:discrete_vs_continuous}.
   
 \end{theorem}

 \begin{proof}
   The proof is postponed to \Cref{sec:proofs-crefs-meth}.
 \end{proof}

 \begin{theorem}
   \label{thm:salem_langevin_cv_fixed}
   Assume \tup{\Cref{assum:glip}}, \tup{\Cref{assum:curvature}}, \tup{\Cref{assum:glip_theta}} and that for any $\theta \in \Theta$ and $x \in \rset^d$, $\norm{H_{\theta}(x)} \leq \Ve^{1/4}(x)$.
    \tup{\Cref{assum:condition_kernel_fix}} 
    holds with $V \leftarrow \Ve$ and  $\bgamma \leftarrow \min(1,2\mtttrois)$ and for any $\gamma_1,\gamma \in \ocint{0,\bgamma}$, $\gamma_2 < \gamma_1$,
    \begin{equation}
        \Lambdabf_1(\gamma_1, \gamma_2) = D_5 (\gamma_1/\gamma_2-1) \eqsp , \quad
        \Lambdabf_2(\gamma_1, \gamma_2) = D_5 \gamma_2^{1/2} \eqsp ,
    \end{equation}
     where $D_5$ is given in \Cref{lem:v_kernel_error} in \Cref{sec:proof-crefthm:s-1}. 
 \end{theorem}

 \begin{proof}
   The proof is postponed to \Cref{sec:proof-crefthm:s-1}.
 \end{proof}




\bibliographystyle{plainnat}
\bibliography{main.bbl}

\appendix
\section{Posterior convexity}
\label{sec:posterior}

\begin{lemma}
  \label{lemma:prekoppa}
For any $y \in \{ 0, 1 \}^{d_y}$, $\theta \mapsto p(y |\theta)$ given by \eqref{EQ: bayes-logit-marg-likelihood} is log-concave.
\end{lemma}

\begin{proof}
  Let $\theta \in \rset$, then by \eqref{EQ: bayes-logit-marg-likelihood}, for any $y \in \rset$ we have $p(y|\theta) = \int_{\rset^d} p(y, \beta | \theta) \rmd \beta$ with
  \begin{equation}
    p(y, \beta | \theta) = (2\uppi\sigma^2)^{-d/2}\defEns{\prod_{i=1}^{{d_y}} s(x_i^{\transpose} \beta)^{y_{i}}(1-s(x_i^{\transpose} \beta))^{1-y_{i}}} \rme^{-\frac{\norm[2]{\beta - \theta \Unbf_d}}{2 \sigma^2}} \eqsp .
  \end{equation}
  Therefore we have using that for any $t \in \rset$, $1 - s(t) = s(-t)$
  \begin{multline}
    \log p(y, \beta | \theta) = (-d/2)\log(2\uppi\sigma^2) \\ + \defEns{\sum_{i=1}^{{d_y}} y_i \log(s(x_i^{\transpose} \beta))  + (1- y_i) \log(s(-x_i^{\transpose} \beta))} -\frac{\norm[2]{\beta - \theta \Unbf_d}}{2 \sigma^2} \eqsp .
  \end{multline}
  Since $y_i \geq 0$, $1-y_i \geq 0$, $(\beta, \theta) \mapsto \norm[2]{\beta - \theta \Unbf_d}$, $t \mapsto \log(s(t))$ and $t \mapsto \log(s(-t))$ are convex, we obtain that $(\beta, \theta) \mapsto p(y, \beta | \theta)$ is log-concave. Using the Prékopa–Leindler inequality \cite[Theorem 7.1]{gardner2002brunn}  we obtain that $\theta \mapsto p(y | \theta)$ is log-concave which concludes the proof.
\end{proof}

\section{Non-convex objective function}
\label{sec:non-convex-objective}

In this section we turn to the case where $f$ is non-convex. We recall that the normal space of a sub-manifold $\mathcal{M} \subset \R^{d_{\Theta}}$ at point $\theta$ is given by
 \begin{equation}
   \mathrm{N}(\theta,\mathcal{M}) = \left\lbrace \begin{aligned}
       & \mathrm{T}(\theta,\mathcal{M})^{\perp}  \quad \text{if } \theta \in \mathcal{M} \eqsp ; \\
       & \lbrace 0 \rbrace \quad \text{otherwise} \eqsp ,
     \end{aligned}
     \right. \label{eq:normal_plan}
 \end{equation}
where $\mathrm{T}(\theta,\mathcal{M})$ is the tangent space of the sub-manifold $\mathcal{M}$ at point $x$, see \cite{aubin:2000}.

 \begin{theorem}
   \label{thm:salem_cv_noncvx}
     Assume \tup{\Cref{assum:theta_compact}}, \tup{\Cref{assum:f_diff}}, \tup{\Cref{assum:grad_expec}} and that $\Theta$ is a $d_{\Theta}$ dimensional connected differentiable manifold with boundary and continuously differentiable outer normal. 
    Let $\bgamma>0$, $(\gamma_n)_{n \in \nset}$,
    $(\delta_n)_{n \in \nset}$ be sequences of non-increasing positive real numbers and $(m_n)_{n \in \nset}$ be a sequence of positive integers such that
    $\sup_{n \in \nset} \delta_n < 1/\L$, $\sup_{n \in \nset} \gamma_n < \bgamma$ and \eqref{eq:condition_cv} are satisfied.   Let $\{(X_k^n)_{k \in \{0,\ldots,m_n\}} \, : \, n \in \nset\}$ be given by \eqref{eq:algo_SOUL}. Assume in addition that  \tup{\Cref{assum:condition_majo_V}} is satisfied.
Then  $(\theta_n)_{n \in \N}$ defined by \eqref{eq:algo_SOUL} converges \as~to some $\theta^{\star} \in \lbrace \theta \in \Theta \, : \,  \ \nabla f(\theta) + \bfn = 0, \  \bfn  \in \mathrm{N}(\theta, \partial \Theta) \rbrace$.
\end{theorem}

 \begin{proof}
   The proof is an application of \cite[Chapter 5, Theorem 2.3]{kushner2003stochastic} using the decomposition of the error term considered in the proof of \Cref{thm:salem_cv} and \Cref{thm:salem_cv_fix}. Indeed we decompose the error term $\eta_n$ defined by \eqref{eq:error_term} as $\eta_n= \delta M_n + B_n$, where $\delta M_n$ is a martingale increment. Then, we only need to show that the following sums converge
   \begin{equation}
     \sum_{k=0}^n \delta_{k+1}^2 \expe{\| \delta M_k\| ^2} \eqsp , \qquad \sum_{k=0}^n \delta_{k+1} \expe{\| B_k\|} \eqsp ,
   \end{equation}
    which is established in \Cref{lemma:error_bound} and \Cref{lemma:error_variance}.
 \end{proof}


\section{Postponed proofs}
We first derive the following technical lemmas.
 \begin{lemma}
   \label{lem:tech_sum_finie}
   Let $t \in (0,1)$ and $\gamma \in \ocint{0,\bgamma}$ with $\bgamma >0$ then $\sum_{n \in \nset} t^{n \gamma}  \leq t^{-\bgamma}\log^{-1}(1/t)\gamma^{-1}$ and $\sum_{n \in \nset} n t^{n \gamma}  \leq t^{- \bgamma}\log^{-2}(1/t)\gamma^{-2}$.
 \end{lemma}
 \begin{proof}

    Let $t \in (0,1)$ and $\gamma \in \ocint{0,\bgamma}$ with $\bgamma >0$. Using that $e^u -1 \leq u e^u$ for all $u\geq 0$, we have 
    \begin{equation}      
      \sum_{n \in \nset} t^{n \gamma} = - (t^{\gamma} - 1)^{-1}  \leq - \gamma^{-1}\log^{-1}(t)\exp(-\log(t)\gamma) \leq t^{-\bgamma}\log^{-1}(1/t)\gamma^{-1} \eqsp,
    \end{equation}
    and
    \begin{equation}      
      \sum_{n \in \nset} n  t^{n \gamma} =  t^{\gamma} (t^{\gamma} - 1)^{-2}  \leq t^{\gamma}\{  \gamma^{-1}\log^{-1}(t)\exp(-\log(t)\gamma) \}^2\leq  t^{- \bgamma}\log^{-2}(1/t)\gamma^{-2} \eqsp,
    \end{equation}
    which completes the proof.
  \end{proof}

    \begin{lemma}
   \label{lemma:jensen}
  For any  probability measures $\mu, \nu $ on $\mcb{\rset^d}$, measurable function $V : \ \rset^d \to \coint{1,+\infty}$ such that $\mu(V) + \nu(V) < +\infty$ and $\expo \in \ooint{0,1}$, we have
   \begin{equation}
     \Vnorm[V^a]{\mu-\nu} \leq 2 \Vnorm{\mu - \nu}^{\expo} \eqsp .
   \end{equation}
 \end{lemma}

 \begin{proof}

\label{lemma:jensen_proof}

 Let $\expo \in \ocint{0,1}$. The statement is trivial if $\mu = \nu$. We just need to consider the case where $\mu \neq \nu$. Define $\xi = \abs{\mu - \nu} / (\abs{\mu - \nu}(\rset^d))$. Using \cite[Definition D.3.1]{douc:moulines:priouret:soulier:2018} we get that 
 \begin{align}
   \Vnorm[V^{\expo}]{\mu -\nu } &= (1/2) \xi(V^{\expo}) \times \abs{\mu - \nu}(\rset^d) \\ 
                                &\leq (1/2) \xi(V)^{\expo} \times \abs{\mu - \nu}(\rset^d) \\
   &\leq 2^{a-1} \Vnorm{\mu - \nu}^{\expo} \times [\abs{\mu - \nu}(\rset^d)]^{1-\expo} \eqsp ,
 \end{align}
   which concludes the proof using that $\expo \leq 1$.
 \end{proof}

Jensen's inequality implies that \Cref{assum:condition_majo_V}-\ref{assum:condition_majo_V_i} holds for $V \leftarrow V^{\expo}$ with $\expo \in \ocint{0,1}$ since $A_1 \geq 1$. \Cref{lemma:jensen} implies that \Cref{assum:condition_majo_V}-\ref{assum:condition_majo_V_ii} holds replacing $V$ by $V^{\expo}$, $\rho$ by $\rho^{\expo}$ and $A_2$ by $2A_2$. Similarly \Cref{assum:condition_majo_V}-\ref{assum:condition_majo_V_iii} holds replacing $V$ by $V^{\expo}$ and $\Psibf(\gamma)$ by $2 \Psibf(\gamma)$.
 
 \subsection{Proof of \Cref{thm:salem_cv}}
\label{thm:salem_cv_proof}

Consider $(\eta_n)_{n \in \N}$ defined for any $n \in \nset$ by 
\begin{equation}
  \eta_n = m_n^{-1} \sum_{k=1}^{m_n} \defEns{H_{\theta_n}(X_k^n) - \pi_{\theta_n}(H_{\theta_n})} \eqsp.
  \label{eq:error_term}
\end{equation}
The proof of \Cref{thm:salem_cv} relies on the two following lemmas.
We consider the following decomposition for any $n \in \nset$, $\eta_n = \eta_n^{(1)} + \eta_n^{(2)}$, where
\begin{equation}
  \label{eq:decompo_1}
  \eta_n^{(1)} = \CPE{\eta_n}{\mathcal{F}_{n-1}} \eqsp , \qquad \eta_n^{(2)} = \eta_n - \CPE{\eta_n}{\mathcal{F}_{n-1}}\eqsp .
\end{equation}
We now give upper bounds on $\expeLigne{\normLigne{\eta_n^{(1)}}}$, $\expeLigne{\normLigne{\eta_n^{(1)}}^2}$ and $\expeLigne{\normLigne{\eta_n^{(2)}}^2}$.
  \begin{lemma}
    \label{lemma:error_bound}
  Assume \tup{\Cref{assum:theta_compact}}, \tup{\Cref{assum:f_diff}}, \tup{\Cref{assum:grad_expec}}, \tup{\Cref{assum:condition_majo_V}} and that for any $\theta \in \Theta$ and $x \in \rset^d$, $\normLigne{H_{\theta}(x)} \leq V^{1/2}(x)$. Then we have for any $n \in \nset$
  \begin{equation}
    \begin{aligned}
      &\expe{\normLigne{\eta_n^{(1)}} }   \leq  B_1 \expe{V^{1/2}(X_0^0)}/(m_n \gamma_n) +   \Psibf(\gamma_{n})  \eqsp ; \\ &\expe{\normLigne{\eta_n^{(1)}}^2 }   \leq  2B_1^2 \expe{V(X_0^0)}/(m_n \gamma_n)^2 +  2 \Psibf(\gamma_{n})^2  \eqsp ,
    \end{aligned}
\end{equation}
with
\begin{equation}
  B_1 =  2A_1 A_2 \rho^{-\bgamma} / \log(1/\rho) \eqsp .
\end{equation}
\end{lemma}
\begin{proof}
  Using the definition of $(\mcf_n)_{n \in \nset}$, see \eqref{eq:def_F_n}, the Markov property, \Cref{assum:condition_majo_V}-\ref{assum:condition_majo_V_ii}-\ref{assum:condition_majo_V_iii}, \Cref{lemma:jensen}, Jensen's inequality and that for any $\theta \in \Theta$ and $x \in \rset^d$, $\normLigne{H_{\theta}(x)} \leq V^{1/2}(x)$, we have for any $n \in \nsets$
  \begin{align}
    &    \| \CPE{ \eta_{n} }{\mathcal{F}_{n-1} } \| \leq m_{n}^{-1} \sum_{k=1}^{m_{n}} \norm{ \Kker_{\gamma_{n},\theta_{n}}^k H_{\theta_{n}}(X_{0}^{n}) - \pi_{\theta_{n}} \left(H_{\theta_{n}}\right) } \\
    &\qquad \leq m_{n}^{-1} \sum_{k=1}^{m_{n}} \norm{ \abs{\updelta_{X_0^n}\Kker_{\gamma_{n},\theta_{n}}^k  - \pi_{\theta_{n}}} \left(H_{\theta_{n}}\right) } \\
    &\qquad \leq m_{n}^{-1} \sum_{k=1}^{m_{n}} \abs{\updelta_{X_0^n}\Kker_{\gamma_{n},\theta_{n}}^k  - \pi_{\theta_{n}}} \left(\norm{H_{\theta_{n}}}\right)  \\    
    &\qquad \leq m_{n}^{-1} \sum_{k=1}^{m_{n}} \defEns{\Vnorm[V^{1/2}]{ \updelta_{X_{0}^n}\Kker_{\gamma_{n},\theta_{n}}^k  - \pi_{\gamma_{n}, \theta_{n}}  }}  + \Vnorm[V^{1/2}] {\pi_{\gamma_{n}, \theta_{n}} - \pi_{\theta_{n}} } \\
                                   &\qquad  \leq m_{n}^{-1} \sum_{k=1}^{m_{n}} \defEns{2A_2\rho^{k \gamma_{n}}V^{1/2}(X_{m_n}^n) + \Psibf(\gamma_{n})}  \\
                                              & \qquad  \leq \frac{2A_2\rho^{-\bgamma}V^{1/2}(X^{n}_{m_{n}})}{\log(1/\rho)\gamma_{n}m_{n}} +  \Psibf(\gamma_{n}) \eqsp , \label{eq:error_cond} 
  \end{align}
  where for the last inequality we have used \Cref{lem:tech_sum_finie}. In a similar manner, we have
  \begin{equation}
    \norm{\CPE{ \eta_0 }{X_0^0}}  \leq \frac{2A_2\rho^{-\bgamma}V^{1/2}(X^{0}_{0})}{\log(1/\rho)\gamma_{0}m_{0}} +  \Psibf(\gamma_{0}) \eqsp .
  \end{equation}
  We conclude using \Cref{assum:condition_majo_V}-\ref{assum:condition_majo_V_i} and that $(a+b)^2 \leq 2a^2 + 2b^2$ for $a,b \in \rset$.
\end{proof}

\begin{lemma}
  \label{lemma:error_variance}
 Assume \tup{\Cref{assum:theta_compact}}, \tup{\Cref{assum:f_diff}}, \tup{\Cref{assum:grad_expec}}, \tup{\Cref{assum:condition_majo_V}} and that for any $\theta \in \Theta$ and $x \in \rset^d$, $\normLigne{H_{\theta}(x)} \leq V^{1/2}(x)$. Then we have for any $n \in \nset$
  \begin{equation}
   \expe{ \normLigne{\eta_n^{(2)}}^2}\leq B_2 m_n^{-2} \gamma_n^{-1} \parenthese{m_n + \gamma_n^{-1} \expe{V(X_0^0)}}) \eqsp ,
 \end{equation}
 with $B_2 = 2(1 + \bgamma)^2 \max (B_{2,1}, B_{2,2})$ and
 \begin{equation}
   \begin{aligned}
     B_{2,1} &= 24A_2^2(1-\rho^{1/2})^{-2}A_3 \eqsp , \\
     B_{2,2} &= 4 A_1 \parentheseDeux{1 + 6A_2^2(1-\rho^{1/2})^{-2} \defEns{A_2(1-\rho)^{-1} + 2} + A_2^2 \log^{-2}(1/\rho) + A_3^2 } \eqsp .
     \end{aligned}
 \end{equation}

\end{lemma}

\begin{proof}
      Let $n \in \nsets$. We have using the Cauchy-Schwarz inequality
  \begin{align}
    &\expe{\norm{ \sum_{k=1}^{m_n} \defEns{H_{\theta_n}(X_k^n) - \CPE{H_{\theta_n}(X_k^n)}{\mathcal{F}_{n-1}}}}^2}\\ & \qquad \leq 2 \expe{\norm{ \sum_{k=1}^{m_n} \defEns{H_{\theta_n}(X_k^n) - \pi_{\gamma_n, \theta_n}(H_{\theta_n}) }}^2}  \\ & \qquad + 2\expe{\norm{ \sum_{k=1}^{m_n}  \defEns{\CPE{H_{\theta_n}(X_k^n)}{\mathcal{F}_{n-1}} - \pi_{\gamma_n,\theta_n}(H_{\theta_n})}}^2} \label{eq:decompo_1}
  \end{align}
Using the Markov property, \Cref{assum:condition_majo_V}-\ref{assum:condition_majo_V_i}-\ref{assum:condition_majo_V_ii}, \Cref{lemma:jensen}, \Cref{lem:tech_sum_finie} and that for any $\theta \in \Theta$ and $x \in \rset^d$, $\normLigne{H_{\theta}(x)} \leq V^{1/2}(x)$ 
we obtain that
\begin{align}
  &\expe{\norm{ \sum_{k=1}^{m_n}  \defEns{\CPE{H_{\theta_n}(X_k^n)}{\mathcal{F}_{n-1}} - \pi_{\gamma_n, \theta_n}(H_{\theta_n})}}^2}  \\ 
  &\qquad \qquad \qquad \qquad  \leq \expe{\abs{\sum_{k=1}^{m_n} \CPE{\Vnorm[V^{1/2}]{\updelta_{X_0^n}\Rker_{\gamma_n, \theta_n} - \pi_{\gamma_n, \theta_n}}}{\mathcal{F}_{n-1}}}^2} \\
  &\qquad \qquad \qquad \qquad  \leq 4A_2^2 \expe{ \abs{\CPE{V^{1/2}(X_0^n)}{\mathcal{F}_{n-1}}\sum_{k=1}^{m_n} \rho^{k\gamma_n/2}}^2}
  \\ & \qquad \qquad \qquad \qquad \leq  4 A_1 A_2^2 \gamma_n^{-2}\rho^{-2 \bgamma} \log^{-2}(1/\rho) \expe{V(X_0^0)} \eqsp . \label{eq:ineq_0}
\end{align}
We now give an upper-bound on the first term in the right-hand side of \eqref{eq:decompo_1}.
Consider for any $n \in \nset$ the Euclidean division of $m_n$ by $\ceil{1/\gamma_n}$ there exist $q_n \in \nset$ and $r_n \in \{0, \dots, \ceil{1/\gamma_n}-1\}$ such that $m_n = q_n\ceil{1 / \gamma_n} + r_n$. Therefore using the Cauchy-Schwarz inequality we can derive the following decomposition
\begin{align}
  \expe{\norm{\sum_{k=1}^{m_n} H_{\theta_n}(X_k^n) - \pi_{\gamma_n, \theta_n}(H_{\theta_n})}^2} &\leq 2\expe{\norm{\sum_{j=1}^{r_n} H_{\theta_n}(X_{j + q_n \ceil{1/\gamma_n}}^n) - \pi_{\gamma_n, \theta_n}(H_{\theta_n})}^2} \\ &\qquad +  2\expe{\norm{\sum_{j=1}^{\ceil{1/\gamma_n}}\sum_{k=0}^{q_n-1} H_{\theta_n}(X_{j+ k \ceil{1/\gamma_n}}^n) - \pi_{\gamma_n, \theta_n}(H_{\theta_n})}^2}                                                                                                                               \\ &\leq 2\expe{\norm{\sum_{j=1}^{r_n} H_{\theta_n}(\bar{X}_{q_n}^{j,n}) - \pi_{\gamma_n, \theta_n}(H_{\theta_n})}^2} \\ &\qquad +  2\ceil{1/\gamma_n}\sum_{j=1}^{\ceil{1/\gamma_n}}\expe{\norm{\sum_{k=0}^{q_n-1} H_{\theta_n}(\bar{X}_k^{j,n}) - \pi_{\gamma_n, \theta_n}(H_{\theta_n})}^2}
       \label{eq:decompo}
\end{align}
Setting for any $j \in \{1, \dots, \ceil{1/\gamma_n} \}$ and $k \in \lbrace{0, \dots, q_n - 1\rbrace}$, $\bar{X}_{k}^{j,n} = X_{j+k\ceil{1/\gamma_n}}^n$. We now bound the two terms in the right-hand side.
First, using the Cauchy-Schwarz inequality and \Cref{assum:condition_majo_V}-\ref{assum:condition_majo_V_i}-\ref{assum:condition_majo_V_iii}, the fact that $r_n \leq \ceil{1/\gamma_n}$ and that for any $\theta \in \Theta$ and $x \in \rset^d$, $\normLigne{H_{\theta}(x)} \leq V^{1/2}(x)$  we have
\begin{align}
  \expe{\norm{\sum_{j=1}^{r_n} H_{\theta_n}(\bar{X}_{q_n}^{j,n}) - \pi_{\gamma_n, \theta_n}(H_{\theta_n})}^2} &\leq r_n  \sum_{j=1}^{r_n} \expe{\norm{H_{\theta_n}(\bar{X}_{q_n}^{j,n}) - \pi_{\gamma_n, \theta_n}(H_{\theta_n})}^2} \\ &\leq \ceil{1/\gamma_n}^2 \parenthese{2 A_1 \expe{V(X_0^0)} + 2 A_3^2} \label{eq:reste} \eqsp .
\end{align}
Now consider the solution of the \emph{Poisson equation} \cite[Section 17.4.1]{meyn1993markov} associated with $\Kker_{\gamma_n, \theta_n}^{\ceil{1/\gamma_n}}$, $x \mapsto \poiss{n}(x)$ defined for any $x \in \rset^d$ by
\begin{equation}
  \label{eq:poisson_def}
  \poiss{n}(x) = \sum_{\ell \in \nset} \left(  \Kker_{\gamma_n, \theta_n}^{\ell \ceil{1/\gamma_n}}H_{\theta_n}(x)  - \pi_{\gamma_n, \theta_n}(H_{\theta_n}) \right) \eqsp .
\end{equation}
Note that by \Cref{assum:condition_majo_V}-\ref{assum:condition_majo_V_ii}, \Cref{lemma:jensen} and since for any $\theta \in \Theta$ and $x \in \rset^d$, $\normLigne{H_{\theta}(x)} \leq V^{1/2}(x)$, we have that for any $x \in \rset^d$
\begin{equation}
  \norm{\poiss{n}(x)} \leq 2A_2  (1 - \rho^{1/2})^{-1} V^{1/2}(x)\eqsp , \label{eq:majo_poiss}
\end{equation}
and in addition for any $x \in \rset^d$
\begin{equation}
 \poiss{n}(x) - \Kker^{\ceil{1/\gamma_n}}_{\gamma_n, \theta_n} \poiss{n}(x) = H_{\theta_n}(x) - \pi_{\gamma_n, \theta_n}(H_{\theta_n}) \eqsp .
\end{equation}
Therefore, we have for any $j \in \{1, \dots, \ceil{1/\gamma_n} \}$
\begin{align}
  \sum_{k=0}^{q_n - 1} \left( H_{\theta_n}(\bar{X}_k^{j,n}) - \pi_{\gamma_n, \theta_n}(H_{\theta_n})\right) &= \sum_{k=0}^{q_n - 1} \left( \poiss{n}(\bar{X}_{k}^{j,n}) - \Kker_{\gamma_n, \theta_n}^{\ceil{1/\gamma_n}}\poiss{n}(\bar{X}_{k}^{j,n}) \right)\\ &= \sum_{k=0}^{q_n - 2} \left( \poiss{n}(\bar{X}_{k+1}^{j,n}) - \Kker_{\gamma_n, \theta_n}^{\ceil{1/\gamma_n}}\poiss{n}(\bar{X}_{k}^{j,n}) \right) \\ & \qquad  + \poiss{n}(\bar{X}_0^{j,n}) - \Kker_{\gamma_n, \theta_n}^{\ceil{1/\gamma_n}} \poiss{n}(\bar{X}_{q_n-1}^{j,n})  \eqsp . \label{eq:decompo_2}
\end{align}
Combining the Cauchy-Schwarz inequality and \eqref{eq:decompo_2}  we obtain that 
\begin{align}
  \expe{\norm{\sum_{k=0}^{q_n - 1} H_{\theta_n}(\bar{X}_k^{j,n}) - \pi_{\gamma_n, \theta_n}(H_{\theta_n})}^2} \leq 3(\mathrm{C}_1 + \mathrm{C}_2) \eqsp , \label{eq:decompo_3}
\end{align}
with
\begin{equation}
  \begin{aligned}
    &\mathrm{C}_1 = \expe{\norm{\poiss{n}(\bar{X}_0^{j,n})}^2 + \Kker_{\gamma_n, \theta_n}^{\ceil{1/\gamma_n}}\norm{\poiss{n}(\bar{X}_{q_n- 1}^{j,n})}^2} \eqsp ;   \\ 
    &\mathrm{C}_2 = \expe{\norm{\sum_{k=0}^{q_n - 2} \poiss{n}(\bar{X}_{k+1}^{j,n}) - \Kker_{\gamma_n, \theta_n}^{\ceil{1/\gamma_n}}\poiss{n}(\bar{X}_{k}^{j,n})}^2} \eqsp .
  \end{aligned}
\end{equation}
First, using \eqref{eq:majo_poiss} and \Cref{assum:condition_majo_V}-\ref{assum:condition_majo_V_i} we get that
\begin{align}
  \mathrm{C}_1 &\leq 4A_2^2(1 - \rho^{1/2})^{-2} \defEns{\expe{V(X_j^n)} + \expe{\Kker_{\gamma_n, \theta_n}V(X_{q_{n}+j-1}^n}}  \\
  &\leq 8A_1A_2^2(1 - \rho^{1/2})^{-2}\expe{V(X_0^0)} \eqsp . \label{eq:inte_1_majo}
\end{align}
We now give an upper-bound on $\mathrm{C}_2$.
For any $j \in \lbrace 1, \dots, r_n \rbrace$ let $(\mathcal{G}_{j,k})_{k \in \lbrace 0, q_n-2\rbrace}$ generated by $\mathcal{F}_{n-1}$ and the sequence of random variables $X_0^n, \dots, X_{k \ceil{1/\gamma_n}+j}^n$. Using the Markov property we have for any $k \in \lbrace 0, \dots, q_n -2 \rbrace$ and $j \in \lbrace 1, \dots, r_n \rbrace$
\begin{equation}
    \CPE{\poiss{n}(X_{k+1}^{j,n})}{\mathcal{G}_{j,k}} = \Kker_{\gamma_n, \theta_n}^{\ceil{1/\gamma_n}}\poiss{n}(X_k^{j,n}) \eqsp .
  \end{equation}
  Therefore, for any $j \in \lbrace 1, \dots, r_n \rbrace$, $\poiss{n}(X_{k+1}^{j,n}) - \Kker_{\gamma_n, \theta_n}^{\ceil{1/\gamma_n}}\poiss{n}(X_k^{j,n})$ is a martingale increment with respect to $(\mathcal{G}_{j,k})_{k \in \lbrace 0, q_n-2\rbrace}$,
  Combining this result with the Markov property implies that for any $k \in \lbrace 0, \dots, q_n -2 \rbrace$ and $j \in \lbrace 1, \dots, r_n \rbrace$, 
\begin{align}
  \mathrm{C}_2 &= \sum_{k=0}^{q_n -2} \expe{\Kker_{\gamma_n, \theta_n}^{\ceil{1/\gamma_n}} \norm{\poiss{n}(\bar{X}_k^{j,n}) -  \Kker_{\gamma_n, \theta_n}^{\ceil{1/\gamma_n}} \poiss{n}(\bar{X}_k^{j,n})}^2}
  \\ &= \sum_{k=0}^{q_n -2} \expe{\Kker_{\gamma_n, \theta_n}^{\ceil{1/\gamma_n}} \norm{\poiss{n}(\bar{X}_k^{j,n})}^2 - \norm{ \Kker_{\gamma_n, \theta_n}^{\ceil{1/\gamma_n}} \poiss{n}(\bar{X}_k^{j,n})}^2}\eqsp . \label{eq:martingale_increment}
\end{align}
Define for any $x \in \rset^d$, $g_n(x) = 
\normLigne{\poiss{n}(x)}^2$.
Using \eqref{eq:martingale_increment}, 
\Cref{assum:condition_majo_V}-\ref{assum:condition_majo_V_ii}-\ref{assum:condition_majo_V_iii} and \eqref{eq:majo_poiss} we obtain that
\begin{align}
  \mathrm{C}_2  &= \sum_{k=0}^{q_n -2} \expe{\Kker_{\gamma_n, \theta_n}^{\ceil{1/\gamma_n}} \norm{\poiss{n}(\bar{X}_k^{j,n})}^2 - \norm{ \Kker_{\gamma_n, \theta_n}^{\ceil{1/\gamma_n}} \poiss{n}(\bar{X}_k^{j,n})}^2} \\
  &\leq \sum_{k=0}^{q_n -2} \expe{\Kker_{\gamma_n, \theta_n}^{\ceil{1/\gamma_n}} \norm{\poiss{n}(\bar{X}_k^{j,n})}^2} \\
  &\leq \expe{\sum_{k=0}^{q_n - 2} \CPE{\Kker_{\gamma_n, \theta_n}^{(k+1) \ceil{1/\gamma_n}} g_n(\bar{X}_0^{j,n}) - \pi_{\gamma_n, \theta_n}(g_n)}{\mathcal{G}_{j, 0}} } + \sum_{k=0}^{q_n - 2} \pi_{\gamma_n, \theta_n} (g_n) \\
    &\leq \frac{4A_2^2}{(1-\rho^{1/2})^{2}}\defEns{\sum_{k=0}^{q_n - 2} \expe{\CPE{\Vnorm{\updelta_{X_j^n} \Kker_{\gamma_n, \theta_n}^{(k+1)\ceil{1/\gamma_n}} - \pi_{\gamma_n, \theta_n}}}{\mathcal{G}_{j,0}}} + \sum_{k=0}^{q_n - 2} \pi_{\gamma_n, \theta_n} (V)} \\
  &\leq 4A_2^2 (1 - \rho^{1/2})^{-2}\defEns{  A_2 (1-\rho)^{-1}\expe{V(X_j^n)} + q_n  A_3 } \\
    &\leq 4A_2^2 (1 - \rho^{1/2})^{-2}\defEns{  A_1 A_2 (1-\rho)^{-1}\expe{V(X_0^0)} + q_n  A_3 } \eqsp . \label{eq:ineq_2}
\end{align}
Therefore, using \eqref{eq:inte_1_majo} and \eqref{eq:ineq_2} in \eqref{eq:decompo_3} we obtain that
\begin{multline}
  \expe{\norm{\sum_{k=0}^{q_n - 1} H_{\theta_n}(\bar{X}_k^{j,n}) - \pi_{\gamma_n, \theta_n}(H_{\theta_n})}^2} \\ \leq 12A_2^2 (1 - \rho^{1/2})^{-2}\parentheseDeux{\defEns{  A_1 A_2 (1-\rho)^{-1}\expe{V(X_0^0)} + q_n  A_3 } + 2 \expe{V(X_0^0)}} \eqsp . \label{eq:majoun}
\end{multline}
As a consequence, using \eqref{eq:reste} and \eqref{eq:majoun} in \eqref{eq:decompo} we get that
\begin{align}
  &\expe{\norm{ \sum_{k=1}^{m_n} H_{\theta_n}(X_k^n) - \pi_{\gamma_n, \theta_n}(H_{\theta_n}) }^2} 
   \leq 4 \ceil{1/\gamma_n}^{2}(A_1\expe{V(X_0^0)} + A_3^2)  \\ &\phantom{aa}+
24\ceil{1/\gamma_n}^2A_2^2(1-\rho^{1/2})^{-2}\defEns{A_1\expe{V(X_0^0)}(A_2(1-\rho)^{-1} + 2) + q_n A_3}                                                                                  \\
  & \leq  \left[\gamma_n^{-2} \parenthese{A_1\expe{V(X_0^0)} \parentheseDeux{24A_2^2(1-\rho^{1/2})^{-2}\defEns{A_2(1-\rho)^{-1} + 2} + 4} + 4A_3^2} \right.\\ &\left. + 24A_2^2(1-\rho^{1/2})^{-2}A_3m_n/\gamma_n\right](1 + \bgamma)^2 \label{eq:ineq_final_1}
\end{align}
Combining \eqref{eq:ineq_0} and \eqref{eq:ineq_final_1} in \eqref{eq:decompo_1} we obtain that
\begin{align}
  &\expe{\norm{ \sum_{k=1}^{m_n} H_{\theta_n}(X_k^n) - \expe{H_{\theta_n}(X_k^n)}}^2}  \leq 8 \gamma_n^{-2} A_1A_2^2 \rho^{-2\bgamma}\log^{-2}(1/\rho)  \expe{V(X_0^0)}
  \\
  &+ 2\left[\gamma_n^{-2} \parenthese{A_1\expe{V(X_0^0)} \parentheseDeux{24A_2^2(1-\rho^{1/2})^{-2}\defEns{A_2(1-\rho)^{-1} + 2} + 4} + 4A_3^2} \right.\\ &\left. + 24A_2^2(1-\rho^{1/2})^{-2}A_3m_n/\gamma_n\right](1 + \bgamma)^2
  \\ &\leq 2 (1+\bgamma)^2\left(A_1\expe{V(X_0^0)} \left[24A_2^2(1-\rho^{1/2})^{-2}\defEns{A_2(1-\rho)^{-1} + 2} \right. \right.\\  &\left. \left.+ 4\defEns{1 + A_2^2 \log^{-2}(1/\rho)}\right] + 4A_3^2\right)\gamma_n^{-2}  +  48A_2^2(1-\rho^{1/2})^{-2}A_3(1 + \bgamma)^2(m_n/\gamma_n) \eqsp ,
\end{align}
which concludes the proof for $n \neq 0$. The same inequality holds in the case where $n = 0$.
\end{proof}

We now turn to the proof of \Cref{thm:salem_cv}.
\begin{proof}[Proof of \Cref{thm:salem_cv}]
  The proof is an application of \cite[Theorem 2, Theorem 3]{atchade2017perturbed}.
  \begin{enumerate}[wide, labelwidth=!, labelindent=0pt, label=(\alph*)]
  \item \label{item:aa}
To apply  \cite[Theorem 2]{atchade2017perturbed}, it is enough to show that the following series converge \as
  \begin{equation}
    \sum_{n =0}^{\plusinfty} \delta_{n+1} \langle \Pi_{\Theta} ( \theta_{n} - \delta_{n+1} \nabla f (\theta_{n})), \eta_{n}^{(i)} \rangle \eqsp , \quad \sum_{n=0}^{\plusinfty} \delta_{n+1}  \eta_n^{(i)}   \eqsp , \quad \sum_{n=0}^{\plusinfty} \delta_{n+1}^2 \| \eta_n^{(i)} \|^2  \eqsp .
  \end{equation}
  where $i \in \{1, 2\}$ and the sequences $(\eta_n^{(1)})_{n \in \nset}$ and $(\eta_n^{(2)})_{n \in \nset}$ are given in \eqref{eq:decompo_1}.

  In the case where $i=1$, since $(\Pi_{\Theta} ( \theta_n - \delta_{n+1} \nabla f(\theta_n)))_{n \in \nset}$ is bounded, we are reduced to proving that \as \ $\sum_{n=0}^{\plusinfty} \delta_{n+1} \| \eta_n^{(1)} \| < + \infty$.
  Using \eqref{eq:condition_cv}, \Cref{lemma:error_bound} and Fubini-Tonelli's theorem we obtain that
  \begin{equation}\expe{\sum_{n \in \nset}\delta_{n+1} \| \eta_n^{(1)} \|} = \sum_{n \in \nset}\delta_{n+1} \expe{\| \eta_n^{(1)} \|} < +\infty \eqsp . \label{eq:cv_1}
  \end{equation}

  We consider the case where $i=2$. Let $(S_n)_{n \in \nset} $ and $(T_n)_{n \in \nset}$ be defined for any $n \in \nset$ by $S_n = \sum_{k=0}^n \delta_{k+1} \langle \Pi_{\Theta} ( \theta_{k} - \delta_{k+1} \nabla f (\theta_{k})), \eta_{k}^{(2)} \rangle$ and $T_n = \sum_{k=0}^n \delta_{k+1} \eta_n^{(2)}$ are $(\mathcal{F}_{n})_{n \in \nset}$-martingale by definition of $(\eta_n^{(2)})_{n \in \nset}$ in \eqref{eq:decompo_1} and $(\mathcal{F}_n)_{n \in \nset}$ in \eqref{eq:def_F_n}. Therefore, using \cite[Section 12.5]{williams1991prob}, the Cauchy-Schwarz inequality and that the sequence $(\Pi_{\Theta} ( \theta_n - \delta_{n+1} \nabla f(\theta_n)))_{n \in \nset}$ is bounded, it suffices to show that $\sum_{n=0}^{+\infty} \delta_{n+1}^2 \expeLigne{\| \eta_n^{(2)} \|^2} < +\infty$. Using \Cref{lemma:error_variance} we get that
  \begin{equation}
    \sum_{n=0}^{+\infty} \delta_{n+1}^2 \expeLigne{\| \eta_n^{(2)} \|^2} \leq B_2 \parenthese{\sum_{n=0}^{+\infty} \delta_{n+1}^2 / (m_n \gamma_n) + \expe{V(X_0^0)}\sum_{n=0}^{+\infty} \delta_{n+1}^2 / (m_n \gamma_n)^2 } \eqsp .
  \end{equation}
  Combining this result and \eqref{eq:cv_1} implies the stated convergence applying \cite[Theorem 2]{atchade2017perturbed}. 
\item Applying \cite[Theorem 3]{atchade2017perturbed}, the Cauchy-Schwarz inequality and using \Cref{assum:theta_compact} we obtain that \as~for any $n \in \N$
  \begin{multline}
    \label{eq:majo_error}
    \sum_{k=1}^n \delta_{k} \defEns{f(\theta_{k}) - \min_{\Theta} f}   \\
    \leq \frac{\| \theta_0 - \theta^{\star} \|^2}{2} - \sum_{k=0}^{n-1} \delta_{k+1} \langle\Pi_{\Theta} ( \theta_k - \delta_{k+1} \nabla f(\theta_k)) -\theta^{\star}, \eta_k \rangle + \sum_{k=0}^{n-1} \delta_{k+1}^2 \| \eta_k \|^2 \\
    \leq 2\Rtheta^2 - \sum_{i=1}^2\sum_{k=0}^{n-1} \delta_{k+1} \langle\Pi_{\Theta} ( \theta_k - \delta_{k+1} \nabla f(\theta_k)) -\theta^{\star}, \eta_k^{(i)} \rangle + 2\sum_{î=1}^2\sum_{k=0}^{n-1} \delta_{k+1}^2 \| \eta_k^{(i)} \|^2 \eqsp .    
   \end{multline}
   which implies by the proof of \ref{item:aa} that $\sup_{n \in \nset}  [ \sum_{k=1}^n \delta_{k} \defEns{f(\theta_{k}) - \min_{\Theta} f}] < +\infty$ \as . The proof is then completed upon dividing \eqref{eq:majo_error}  by $\sum_{k=1}^n \delta_k$.
 \end{enumerate}
\end{proof}
\subsection{Proof of \Cref{thm:salem_cv_control}}
\label{thm:salem_cv_control_proof}

\begin{proof}
  Taking the  expectation in \eqref{eq:majo_error} and using that $\eta_n^{(2)}$ is a martingale increment with respect to $(\mathcal{F}_n)_{n \in \nset}$, we get that for every $n \in \N$
    \begin{align}
      &\expe{ \sum_{k=1}^{n} \delta_{k} \defEns{f(\theta_k) - \min_{\Theta} f}} \\ &\leq \expe{2\Rtheta^2 - \sum_{k=0}^{n-1} \delta_{k+1} \langle\Pi_{\Theta} ( \theta_k - \delta_{k+1} \nabla f(\theta_k)) -\theta^{\star}, \eta_k \rangle + \sum_{k=0}^{n-1} \delta_{k+1}^2 \| \eta_k \|^2 } \\
      &\leq 2 \Rtheta^2 + 2 \Rtheta \sum_{k=0}^{n-1} \delta_{k+1} \expe{\norm{\eta_k^{(1)}}} + 2 \sum_{k=0}^{n-1} \delta_{k+1}^2 \expe{\norm{\eta_k^{(1)}}^2} + 2 \sum_{k=0}^{n-1} \delta_{k+1}^2 \expe{\norm{\eta_k^{(2)}}^2}
              \label{eq:majo_L1}
    \end{align}
    Combining this result, \Cref{lemma:error_bound} and \Cref{lemma:error_variance} completes the proof.
 \end{proof}



 \subsection{Proof of \Cref{thm:salem_cv_fix}}
 \label{thm:salem_cv_fix_proof}

We now introduce some tools needed for the proof.  By  \Cref{assum:H_lip} and \Cref{assum:condition_majo_V}-\ref{assum:condition_majo_V_i}-\ref{assum:condition_majo_V_ii}, for any
  $\theta \in \Theta$ and $ \gamma \in \ocint{0,\bgamma}$, there exists a
  function $\hat{H}_{\gamma,\theta}: \rset^d \to \rset^{d_{\theta}}$ solution of
  the \textit{Poisson equation},
  \begin{equation}
  \label{eq:poisson_eq}
(\Id-\Kker_{\gamma,\theta}) \hat{H}_{\gamma,\theta} = H_{\theta} -
  \pi_{\gamma, \theta}(H_{\theta})    \eqsp,
  \end{equation}
  defined for any  $x \in \rset^d$ by 
  \begin{equation}
    \label{eq:def_poisson}
       \hat{H}_{\gamma,\theta}(x) = \sum_{j \in \N} \{
  \Kker_{\gamma,\theta}^jH_{\theta}(x) - \pi_{\gamma,
    \theta}(H_{\theta})\} \eqsp .
\end{equation}
Note that using \Cref{assum:condition_majo_V}-\ref{assum:condition_majo_V_ii} and \Cref{lemma:jensen} we have for any $\theta \in \Theta$ and $x \in \rset^d$
\begin{equation}
  \norm{\hat{H}_{\theta}(x)} \leq C_{\hat{H}} \gamma^{-1} V^{1/4}(x)  \label{eq:majo_poisson} \eqsp , \qquad C_{\hat{H}} = 8 A_2 \log^{-1}(1/\rho) \rho^{-\bgamma/4} \eqsp .
\end{equation}
Define for any $n \in \nset$
\begin{equation}
  \label{eq:tetan_def}\teta_n = H_{\theta_n}(\tX_{n+1}) - \pi_{\tilde{\theta}_n}(H_{\tilde{\theta}_n}) \eqsp . \end{equation}
Using \eqref{eq:poisson_eq} an alternative expression of $(\teta_n)_{n \in \nset}$ is given for any $n \in \N$ by 
\begin{align}
  &  \teta_{n} =    \tpoiss{n}(\tX_{n+1})-\tkernel{n} \tpoiss{n}(\tX_{n+1}) + \pi_{\gamma_n, \ttheta_n}(H_{\ttheta_n}) - \pi_{\ttheta_{n}}(H_{\ttheta_{n}}) \\
\label{eq:proof_salem_cv_fix_decomp_eta}
  & \phantom{aa} = \tetaa{n} + \tetab{n} + \tetac{n} + \tetad{n} \eqsp,
\end{align}
where 
\begin{equation}
\label{eq:tetan_def_2}
  \begin{aligned}
  \tetaa{n} &=  \tpoiss{n}(\tX_{n+1}) - \tkernel{n}\tpoiss{n}(\tX_n) \eqsp,
  \\
  \tetab{n} &=  \tkernel{n}\tpoiss{n}(\tX_{n}) - \tkernel{n+1}\tpoiss{n+1}(\tX_{n+1})  \eqsp ,
  \\
  \tetac{n} &= \tkernel{n+1}\tpoiss{n+1}(\tX_{n+1}) -  \tkernel{n}\tpoiss{n}(\tX_{n+1})
  \eqsp ,
  \\
       \tetad{n} &= \pi_{\gamma_n, \ttheta_n}(H_{\ttheta_n}) - \pi_{\ttheta_n}(H_{\ttheta_n}) \eqsp.
    \end{aligned}
  \end{equation}
To establish \Cref{thm:salem_cv_fix} we need to get estimates on moments of $\norm{\teta_n^{(i)}}$ for $i \in \lbrace a,b,c,d \rbrace$. It is the matter of the following technical results.
  
  \begin{lemma}
    \label{lemma:cv_sq_fix}
    Assume \tup{\Cref{assum:theta_compact}}, \tup{\Cref{assum:f_diff}}, \tup{\Cref{assum:grad_expec}}, \tup{\Cref{assum:condition_majo_V}} and that for any $\theta \in \Theta$ and $x \in \rset^d$, $\normLigne{H_{\theta}(x)} \leq V^{1/4}(x)$. Then we have for any $n \in \nset$, 
 $   \expe{ \| \teta_n\|^2}  \leq C_1$, 
  with
  \begin{equation}
    \label{eq:C1}
    C_1 = 2A_1\expe{V^{1/2}(\tX_0)}  +  2 \sup_{\theta \in \Theta} \| \nabla f(\theta) \|^2  \eqsp .
  \end{equation}
\end{lemma}

  \begin{proof}
    Using \eqref{eq:tetan_def}, that $\norm[2]{x+y} \leq 2 (\norm[2]{x} + \norm[2]{y})$ for any $x,y \in \rset^d$, \Cref{assum:theta_compact}, \Cref{assum:f_diff}, \Cref{assum:grad_expec} and \Cref{assum:condition_majo_V}-\ref{assum:condition_majo_V_i} and that for any $\theta \in \Theta$ and $x \in \rset^d$, $\normLigne{H_{\theta}(x)} \leq V^{1/2}(x)$, we get for any $ k \in \N$,
    \begin{align}
      \expe{\| \teta_k \|^2} &\leq 2 \expeLigne{\normLigne{H_{\ttheta_k}(\tX_{k+1})}^2} + 2\parentheseDeux{\pi_{\ttheta_k}\parentheseLigne{\normLigne{H_{\ttheta_k}}}}^2
      \\ &\leq 2 A_1 \expe{V^{1/2}(\tX_0)} + 2 \sup_{\theta \in \Theta} \| \nabla f(\theta) \|^2 \eqsp .
    \end{align}

\end{proof}

\begin{lemma}
  \label{lemma:cv_norm_a}
    Assume \tup{\Cref{assum:theta_compact}}, \tup{\Cref{assum:f_diff}}, \tup{\Cref{assum:grad_expec}}, \tup{\Cref{assum:H_lip}}, \tup{\Cref{assum:condition_majo_V}}, \tup{\Cref{assum:condition_kernel_fix}} and that for any $\theta \in \Theta$ and $x \in \rset^d$, $\normLigne{H_{\theta}(x)} \leq V^{1/4}(x)$. Then we have for any $n \in \nset$, 
    $\expe{\norm{\tetaa{n}}^2} \leq \tilde{C_1} \gamma_n^{-2}$, 
  with
  \begin{equation}
    \label{eq:C0}
    \tilde{C_1} = A_1 C_{\hat{H}}^2 \expe{V^{1/2}(\tX_0)} \eqsp .
  \end{equation}
\end{lemma}

\begin{proof}
By \eqref{eq:tetan_def_2}, using \eqref{eq:majo_poisson} and \Cref{assum:condition_majo_V}-\ref{assum:condition_majo_V_i} we get that for any $n \in \nsets$
\begin{align}
  &\expe{\CPE{\norm{\teta_n^{(a)}}^2}{\mathcal{F}_{n}}} \\ &\leq \expe{\CPE{\norm{\tpoiss{n}(\tX_{n+1})}^2}{\mathcal{F}_n}}  - \expe{\norm{\tkernel{n}\tpoiss{n}(\tX_n)}^2} \\
  &\leq A_1 C_{\hat{H}}^2 \gamma_n^{-2} \expe{V^{1/2}(\tX_0)} \eqsp ,
\end{align}
which concludes the proof.
\end{proof}

\begin{lemma}
  \label{lemma:cv_norm_b}
    Assume \tup{\Cref{assum:theta_compact}}, \tup{\Cref{assum:f_diff}}, \tup{\Cref{assum:grad_expec}}, \tup{\Cref{assum:condition_majo_V}} and that for any $\theta \in \Theta$ and $x \in \rset^d$, $\normLigne{H_{\theta}(x)} \leq V^{1/4}(x)$. Then the following statements hold.
    \begin{enumerate}[label=(\alph*), wide, labelwidth=!, labelindent=0pt]
    \item \label{lemma:cv_norm_b_i}There exists $C_3 \geq 0$ such that for any $n \in \nset$ and $\theta \in \Theta$
      \begin{equation}
        \expe{\norm{\sum_{k=0}^n \delta_{k+1} \langle a_{k+1}, \teta_k^b \rangle}} \leq C_3 \parentheseDeux{\sum_{k=0}^n \abs{\delta_{k+1} - \delta_k} \gamma_k^{-1} +  \sum_{k=0}^n \delta_{k+1}^2\gamma_k^{-1} +  (\delta_{n+1} / \gamma_{n+1} - \delta_1 / \gamma_1)} \eqsp .
      \end{equation}
      with $a_{k+1} = \Pi_{\Theta} \parentheseDeux{\ttheta_k - \delta_{k+1} \nabla f(\ttheta_k)} - \theta^{\star}$, $\theta^{\star} \in \argmin_{\Theta} f$  and 
      \begin{equation}
        \label{eq:C3}
        C_3 =  A_1C_{\hat{H}} (4 \Rtheta + \sup_{\theta \in \Theta} \norm{\nabla f(\theta)} + 1 + \delta_1 L_f) \expe{V^{1/4}(\tX_0)} \eqsp .
      \end{equation} 
    \item \label{lemma:cv_norm_b_ii} If \eqref{eq:condition_cv_fix_1} holds then $\sum_{k=0}^n \delta_{k+1}\langle a_{k+1},  \tetab{k} \rangle$ converges \as .
    \end{enumerate}
  \end{lemma}

  \begin{proof}

    By \eqref{eq:tetan_def_2} we have for any $n \in \nset$ and $\theta \in \Theta$ 
    \begin{align}
     &\sum_{k=0}^n  \delta_{k+1} \langle a_{k+1}, \teta_k^{(b)}\rangle  \\ &= \sum_{k=0}^n \langle \delta_{k+1} a_{k+1},   \tkernel{k}\tpoiss{k}(\tX_k) -  \tkernel{k+1}\tpoiss{k+1}(\tX_{k+1}) \rangle \\ &=  \sum_{k=1}^{n} \langle \delta_{k+1}a_{k+1} - \delta_{k}a_{k}, \tkernel{k} \tpoiss{k} (\tX_{k}) \rangle  \\ &- \langle \delta_{n+1}a_{n+1}, \tkernel{n+1} \tpoiss{n+1} (\tX_{n+1}) \rangle \\ &+ \langle \delta_1 a_{1}, \tkernel{0}\tpoiss{0} (\tX_0) \rangle  \eqsp ,\label{eq:ipp}
    \end{align}
    In addition, we have for any $n \in \nset$, $\theta \in \Theta$ using \Cref{assum:theta_compact}, \Cref{assum:f_diff}, that $\Pi_{\Theta}$ is non-expansive, \eqref{eq:algo_x_fix}, \Cref{assum:condition_majo_V}-\ref{assum:condition_majo_V_i} and that for any $\theta \in \Theta$ and $x \in \rset^d$, $\normLigne{H_{\theta}(x)} \leq V^{1/4}(x)$
    \begin{align}
      &\norm{\delta_{n+1} a_{n+1} - \delta_n a_{n}} \leq 2 \Rtheta \abs{\delta_{n+1} - \delta_n} + \delta_{n+1} \norm{a_{n+1} - a_{n}}\\
      & \qquad \leq 2 \Rtheta \abs{\delta_{n+1} - \delta_n} + (1 + \delta_n L_f) \norm{\theta_{n+1} - \theta_n} + \abs{\delta_{n+1} - \delta_n} \norm{\nabla f(\theta_{n+1})} \\
      & \qquad \leq (2 \Rtheta + \sup_{\theta \in \Theta} \norm{\nabla f(\theta)}) \abs{\delta_{n+1} - \delta_n} + \delta_{n+1}^2 (1+ \delta_{n+1} L_f)V^{1/4}(\tX_{n+1}) \eqsp . \label{eq:inter_1}
    \end{align}
\begin{enumerate}[wide, labelwidth=!, labelindent=0pt, label=(\alph*)]
  
\item Combining \eqref{eq:ipp}, \eqref{eq:inter_1}, \eqref{eq:majo_poisson}, the Cauchy-Schwarz inequality and \Cref{assum:condition_majo_V}-\ref{assum:condition_majo_V_i} we get that
  \begin{align}
    \expe{\norm{\sum_{k=0}^n  \delta_{k+1} \langle a_{k}, \teta_k^{[b)}\rangle}} &\leq (2\Rtheta + \sup_{\theta \in \Theta} \norm{\nabla f(\theta)})A_1 C_{\hat{H}} \expe{V^{1/4}(\tX_0)} \sum_{k=0}^n \abs{\delta_{k+1} - \delta_k} \gamma_k^{-1} \\ &+ A_1 C_{\hat{H}} (1+\delta_1 L_f)\expe{V^{1/4}(\tX_0)} \sum_{k=0}^n \delta_{k+1}^2\gamma_k^{-1} \\ &+ 2 A_1\Rtheta C_{\hat{H}} \expe{V^{1/4}(\tX_0)} \defEns{\delta_{n+1}/\gamma_{n+1} + \delta_1 / \gamma_1} \eqsp ,
  \end{align}
  which concludes the proof of \Cref{lemma:cv_norm_b}-\ref{lemma:cv_norm_b_i}.
\item Assume now
  \eqref{eq:condition_cv_fix_1}. We show that \as \ the first term in \eqref{eq:ipp} is absolutely convergence and the second term converges to $0$.
    
  
Using \eqref{eq:inter_1}, \eqref{eq:majo_poisson}, the Cauchy-Schwarz inequality and 
\eqref{eq:condition_cv_fix_1} we get that
\begin{align}
&\expe{\sum_{k=1}^{+\infty} \abs{\langle \delta_{k+1}a_{k+1} - \delta_{k}a_{k}, \tkernel{k} \tpoiss{k} (\tX_{k}) \rangle} } \\ &\leq (2\Rtheta + \sup_{\theta \in \Theta} \norm{\nabla f(\theta)})A_1C_{\hat{H}}\expe{V^{1/4}(\tX_0)} \sum_{k=0}^{+\infty} \abs{\delta_{k+1} - \delta_k} \gamma_k^{-1}  \\ & \qquad + A_1 C_{\hat{H}}(1 + \delta_1 L_f) \expe{V^{1/4}(\tX_0)} \sum_{k=0}^{+\infty} \delta_{k+1}^2 < +\infty\eqsp ,
\end{align}
which implies that $(\langle \delta_{k+1}a_{k+1} - \delta_{k}a_{k}, \tkernel{k} \tpoiss{k} (\tX_{k})\rangle)_{k \in \nset}$ is absolutely convergent \as . 

We have that $ \tkernel{n+1} \|\tpoiss{n+1} (\tX_{n+1}) \|$ is upper-bounded using \eqref{eq:majo_poisson} by $\gamma_{n+1}^{-1} C_{\hat{H}} \tkernel{n+1} V^{1/4}(\tX_{n+1})$. 
It follows that we have for any $\theta \in \Theta$, $\vareps >0$,  using the Markov inequality, the Cauchy-Schwarz inequality, \eqref{eq:majo_poisson} 
and \eqref{eq:condition_cv_fix_1}
\begin{align}
  &\sum_{n \in \N} \proba{\norm{a_{n+1}} \delta_{n+1}  \tkernel{n+1} \|\tpoiss{n+1} (\tX_{n+1})\| \geq \vareps} \\
  & \qquad  \leq 
    \sum_{n \in \N} \proba{2C_{\hat{H}}\Rtheta\,  \delta_{n+1} \, \gamma_{n+1}^{-1} \, \tkernel{n+1}V^{1/4}(\tX_{n+1})  \geq \vareps}\\
  &  \qquad  \leq  4\vareps^{-2} \Rtheta^2 C_{\hat{H}}^2 A_1 \expe{V^{1/2}(\tX_0)}   \sum_{n \in \N} \delta_n^2 \gamma_n^{-2}  < +\infty\eqsp,
\end{align} 
Using the Borel-Cantelli lemma, we get $\lim_{n \to +\infty} \langle \delta_n a_{n} \tkernel{n} \tpoiss{n}(\tX_n) \rangle = 0$ \as .
This completes the proof of convergence of the series $\sum_{k \in \nset} \delta_{k+1} \langle a_{k+1}, \teta_k^{(b)} \rangle $ for any $\theta \in \Theta$.

\end{enumerate}
\end{proof}

\begin{lemma}
  \label{lemma:cv_norm_c_d}
    Assume \tup{\Cref{assum:theta_compact}}, \tup{\Cref{assum:f_diff}}, \tup{\Cref{assum:grad_expec}}, \tup{\Cref{assum:H_lip}}, \tup{\Cref{assum:condition_majo_V}}, \tup{\Cref{assum:condition_kernel_fix}} and that for any $\theta \in \Theta$ and $x \in \rset^d$, $\normLigne{H_{\theta}(x)} \leq V^{1/4}(x)$. Then we have for any $n \in \nset$
    \begin{equation}
      \expe{\norm{\tetac{n}}} \leq C_2 \gamma_{n+1}^{-1}\parentheseDeux{ \gamma_{n+1}^{-1} \defEns{\Lambdabf_1(\gamma_{n}, \gamma_{n+1})  + \Lambdabf_2(\gamma_{n}, \gamma_{n+1})\delta_{n+1}} + \delta_{n+1}} \eqsp ,
  \end{equation}
  with
  \begin{equation}
    \label{eq:C2}
    C_2 = 4 A_1A_2 \log^{-1}(1/\rho) \rho^{-\bgamma/2} \,  \max\parentheseDeux{L_H, C_{c,1} +  2 A_2    \log^{-1}(1/\rho)\rho^{-\bgamma/2}  } \eqsp ,
  \end{equation}
  where $C_{c,1}$ is given by
  \begin{equation} \label{eq:Cc1} C_{c,1} =  4 A_1  A_2 \log^{-1}(1/\rho) \rho^{-\bgamma/2}\expe{V(\tX_0)} \eqsp . \end{equation}
\end{lemma}

\begin{proof}
 We start by giving an upper-bound on $\Vnorm[V^{1/2}]{ \pi_{\gamma_1, \theta_1} - \pi_{\gamma_2, \theta_2}}$ for $\gamma_1, \gamma_2 \in \ocint{0,\bgamma}$ with $\gamma_1 > \gamma_2$ and, $\theta_1, \theta_2 \in \Theta$. Let $g: \rset^d\to \rset^{d_{\theta}}$ be a measurable function satisfying $\sup_{x \in \rset^d} \{\norm{g(x)}/V^{1/2}(x) \} \leq 1$. Using \Cref{assum:condition_majo_V}-\ref{assum:condition_majo_V_i}-\ref{assum:condition_majo_V_ii}, \Cref{assum:condition_kernel_fix}, \Cref{lem:tech_sum_finie} and \Cref{lemma:jensen},  we get that for any $\gamma_1, \gamma_2 \in \ocint{0,\bgamma}$ with $\gamma_1 > \gamma_2$, $\theta_1, \theta_2 \in \Theta$ and $\ell \in \nsets$
  \begin{align}
    &\expe{\norm{\Kker_{\gamma_1,\theta_1}^\ell g(\tX_0) - \Kker_{\gamma_2,\theta_2}^\ell g(\tX_0)}} \\ &= \norm{\sum_{j=0}^{\ell-1} \Kker_{\gamma_1,\theta_1}^{j }(\Kker_{\gamma_1, \theta_1}^{} - \Kker_{\gamma_2, \theta_2}^{}) \defEns{\Kker_{\gamma_2, \theta_2}^{(\ell-1-j)}g(x) - \pi_{\gamma_2,\theta_2}(f)}} \\
                                                                            &\leq 2A_2 \sum_{j=0}^{\ell-1} \rho^{(\ell-1-j)\gamma_2/2}\norm{\Kker_{\gamma_1,\theta_1}^{j }(\Kker_{\gamma_1, \theta_1}^{} - \Kker_{\gamma_2, \theta_2}^{})V^{1/2}(x)} \\
    \label{eq:diff_V_norm_noyau_itere}
                                                                            &\leq 2A_2  \sum_{j=0}^{\ell-1}  \rho^{(\ell-1-j)\gamma_2/2}  \left[ \Lambdabf_1(\gamma_1, \gamma_2)  + \Lambdabf_2(\gamma_1, \gamma_2)\| \theta_1 - \theta_2 \| \right]  \sup_{k \in \N} \expe{\Kker^{k}_{\gamma_1,\theta_1} V(\tX_0)} \\ &\leq 4A_1  A_2 \log^{-1}(1/\rho) \rho^{-\bgamma/2} \gamma_2^{-1} \parentheseDeux{ \Lambdabf_1(\gamma_1, \gamma_2)  + \Lambdabf_2(\gamma_1, \gamma_2)\| \theta_1 - \theta_2 \| }  \expe{V(\tX_0)} \eqsp .
  \end{align}
  Taking $\ell \to \plusinfty$ and using \Cref{assum:condition_majo_V}-\ref{assum:condition_majo_V_ii},  we obtain that for any $\theta_1, \theta_2 \in \Theta$ and $\gamma_1, \gamma_2 \in \ocint{0,\bgamma}$ with $\gamma_1 > \gamma_2$,
  \begin{equation}
    \label{eq:distance_pi}
  \Vnorm[V^{1/2}]{\pi_{\gamma_1,\theta_1} - \pi_{\gamma_2,\theta_2} } \leq C_{c,1} \gamma_2^{-1} \parentheseDeux{ \Lambdabf_1(\gamma_1, \gamma_2)  + \Lambdabf_2(\gamma_1, \gamma_2)\| \theta_1 - \theta_2 \| }  \eqsp ,
\end{equation}
with $C_{c,1}$ given by\eqref{eq:Cc1}.
In what follows we give an upper bound on  $\norm{\Kker_{\gamma_1, \theta_1} \poissp{\gamma_1}{\theta_1}(x) - \Kker_{\gamma_2, \theta_2} \poissp{\gamma_2}{\theta_2}(x)}$ for any $\theta_1, \theta_2 \in \Theta$, $\gamma_1, \gamma_2 \in \ocint{0,\bgamma}$ with $\gamma_1 > \gamma_2$ and $x \in \rset^d$.
By \eqref{eq:def_poisson}  we have  for any $\theta_1, \theta_2 \in \Theta$, $\gamma_1, \gamma_2 \in \ocint{0,\bgamma}$ with $\gamma_1 > \gamma_2$ and $x \in \rset^d$,
\begin{align}
&\norm{\Kker_{\gamma_1, \theta_1} \poissp{\gamma_1}{\theta_1}(x) - \Kker_{\gamma_2, \theta_2} \poissp{\gamma_2}{\theta_2}(x)} 
\\
& = \norm{ \sum_{\ell \in \nsets} \defEns{\Kker_{\gamma_1,\theta_1}^\ell H_{\theta_1} (x) - \pi_{\gamma_1,\theta_1}(H_{\theta_1})} -\sum_{\ell \in \nsets} \defEns{\Kker_{\gamma_2,\theta_2}^\ell H_{\theta_2}(x) - \pi_{\gamma_2,\theta_2}(H_{\theta_2})} } \\
\label{eq:bound_eta_c_0}
  &  \leq \sum_{\ell \in \nsets}  \norm{  \defEns{\Kker_{\gamma_1,\theta_1}^\ell H_{\theta_1}(x) - \pi_{\gamma_1,\theta_1}(H_{\theta_1})} -\defEns{\Kker_{\gamma_2,\theta_2}^\ell H_{\theta_2}(x) - \pi_{\gamma_2,\theta_2}(H_{\theta_2})} }  \eqsp .
\end{align}
We now bound each term of the series in the right hand side.
For any measurable functions $g_1,g_2$ with $g_i: \rset^d \to \rset^{d_{\theta}}$ and such that $\sup_{x \in \rset^d} \norm{g_i(x)}/V^{1/4}(x) < + \infty$ with $i \in \{ 1, 2 \}$,  
 $\theta_1, \theta_2 \in \Theta$, $\gamma_1, \gamma_2 \in \ocint{0,\bgamma}$ with $\gamma_1 > \gamma_2$,  $x \in \rset^d$  and $\ell \in \nsets$, it holds that 
\begin{align}
  \label{eq:error_all}
  &\Kker_{\gamma_1,\theta_1}^\ell g_1(x) - \Kker_{\gamma_2,\theta_2}^\ell  g_2(x) = \Kker_{\gamma_1,\theta_1}^\ell g_1(x) - \Kker_{\gamma_2,\theta_2}^\ell  g_1(x) + \Kker_{\gamma_2,\theta_2}^\ell  (g_1(x) - g_2(x)) \\
  &= \sum_{j=0}^{\ell -1}\defEns{\Kker_{\gamma_1,\theta_1}^j - \pi_{\gamma_1,\theta_1}}(\Kker_{\gamma_1,\theta_1} - \Kker_{\gamma_2,\theta_2}) \defEns{\Kker_{\gamma_2,\theta_2}^{\ell -1-j}g_1(x) - \pi_{\gamma_2,\theta_2}(g_1)} \\
                                                             & \phantom{aa} + \sum_{j=0}^{\ell -1}\pi_{\gamma_1,\theta_1}\defEns{\Kker_{\gamma_2,\theta_2}^{\ell -1-j}g_1(x) - \Kker_{\gamma_2,\theta_2}^{\ell -j}g_1(x)} + \Kker_{\gamma_2,\theta_2}^\ell  (g_1(x) - g_2(x)) \\
  &=\sum_{j=0}^{\ell -1}\defEns{\Kker_{\gamma_1,\theta_1}^j - \pi_{\gamma_1,\theta_1}}(\Kker_{\gamma_1,\theta_1} - \Kker_{\gamma_2,\theta_2}) \defEns{\Kker_{\gamma_2,\theta_2}^{\ell -1-j}g_1(x) - \pi_{\gamma_2,\theta_2}(g_1)} \\
  & \phantom{aa} - \pi_{\gamma_1,\theta_1}(\Kker_{\gamma_2,\theta_2}^\ell g_1(x) - g_1(x)) + \Kker_{\gamma_2,\theta_2}^\ell  (g_1(x) - g_2(x)) \eqsp. 
    \end{align}
Setting $g_1 = H_{\theta_1} - \pi_{\gamma_1,\theta_1}(H_{\theta_1})$ and $g_2 = H_{\theta_2} - \pi_{\gamma_2,\theta_2}(H_{\theta_2})$, we obtain that
\begin{multline}
\label{eq:decompo_K_gamma_ell}
  \Kker_{\gamma_1,\theta_1}^\ell  g_1(x) - \Kker_{\gamma_2,\theta_2}^\ell  g_2(x) \\
  =  \sum_{j=0}^{\ell -1}\defEns{\Kker_{\gamma_1,\theta_1}^j - \pi_{\gamma_1,\theta_1}}(\Kker_{\gamma_1,\theta_1} - \Kker_{\gamma_2,\theta_2}) \defEns{\Kker_{\gamma_2,\theta_2}^{\ell -1-j}H_{\theta_1}(x) - \pi_{\gamma_2,\theta_2}(H_{\theta_1})}
                                                                                                              + \Xi_{\ell} \eqsp,
\end{multline}
where
\begin{align}
  \Xi_{\ell} &= - \pi_{\gamma_1,\theta_1}(\Kker_{\gamma_2,\theta_2}^\ell H_{\theta_1}(x) - H_{\theta_1}(x)) \\ &+ \Kker_{\gamma_2,\theta_2}^{\ell}\big[ H_{\theta_1}(x) - H_{\theta_2}(x) + \pi_{\gamma_2,\theta_2}(H_{\theta_2}) - \pi_{\gamma_1,\theta_1}(H_{\theta_1})\big] \\
                                                                                                   & =  - \pi_{\gamma_1,\theta_1}\Kker_{\gamma_2,\theta_2}^\ell H_{\theta_1}(x) + \Kker_{\gamma_2,\theta_2}^{\ell}\big[H_{\theta_1}(x) - H_{\theta_2}(x) + \pi_{\gamma_2,\theta_2}(H_{\theta_2})\big] \\
      & = (\pi_{\gamma_2,\theta_2} - \pi_{\gamma_1,\theta_1})(\Kker_{\gamma_2,\theta_2}^\ell H_{\theta_1}(x)-\pi_{\gamma_2,\theta_2}(H_{\theta_1})) - \pi_{\gamma_2, \theta_2}(H_{\theta_1})\\
      & \qquad \qquad \qquad \qquad +\Kker_{\gamma_2,\theta_2}^{\ell}\big[H_{\theta_1}(x) - H_{\theta_2}(x) + \pi_{\gamma_2,\theta_2}(H_{\theta_2})\big]  \\
  \label{eq:decompo_K_gamma_ell_XI}
  & = (\pi_{\gamma_2,\theta_2} - \pi_{\gamma_1,\theta_1})(\Kker_{\gamma_2,\theta_2}^\ell H_{\theta_1}(x)-\pi_{\gamma_2,\theta_2}(H_{\theta_1}))  \\ & \qquad \qquad \qquad \qquad+ \Kker_{\gamma_2,\theta_2}^{\ell}(H_{\theta_1} - H_{\theta_2})(x) - \pi_{\gamma_2,\theta_2}(H_{\theta_1} - H_{\theta_2}) \eqsp .
\end{align}

For the first term in \eqref{eq:decompo_K_gamma_ell}, using  \Cref{assum:condition_majo_V}-\ref{assum:condition_majo_V_ii}, \Cref{assum:condition_kernel_fix}, \Cref{lemma:jensen} and and that for any $\theta \in \Theta$ and $x \in \rset^d$, $\normLigne{H_{\theta}(x)} \leq V^{1/4}(x)$  we obtain for any $\theta_1, \theta_2 \in \Theta$, $\gamma_1, \gamma_2 \in \ocint{0,\bgamma}$ with $\gamma_1 > \gamma_2$, $x \in \rset^d$ and $\ell \in \nsets$
\begin{align}
  &\norm{\sum_{j=0}^{\ell -1}\defEns{\Kker_{\gamma_1,\theta_1}^j - \pi_{\gamma_1,\theta_1}}(\Kker_{\gamma_1,\theta_1} - \Kker_{\gamma_2,\theta_2}) \defEns{\Kker_{\gamma_2,\theta_2}^{\ell -1-j}H_{\theta_1}(x) - \pi_{\gamma_2,\theta_2}(H_{\theta_1})}} \\
  &
\phantom{aaaa} \leq     2 A_2  \sum_{j=0}^{\ell -1} \rho^{(\ell - 1 -j ) \gamma_1/2} \norm{\defEns{\Kker_{\gamma_1,\theta_1}^j - \pi_{\gamma_1,\theta_1}}(\Kker_{\gamma_1,\theta_1} - \Kker_{\gamma_2,\theta_2})V^{1/2}(x)}
  \\
  & \phantom{aaaa}\leq 4 A_2^2 \parentheseDeux{ \Lambdabf_1(\gamma_1, \gamma_2)  + \Lambdabf_2(\gamma_1, \gamma_2)\| \theta_1 - \theta_2 \| }  \sum_{j=0}^{\ell -1}\rho^{(j  + (\ell -1-j))\gamma_2/2} V^{1/2}(x)\\
  & \phantom{aaaa} \leq 4 A_2^2  \parentheseDeux{ \Lambdabf_1(\gamma_1, \gamma_2)  + \Lambdabf_2(\gamma_1, \gamma_2)\| \theta_1 - \theta_2 \| }  \ell \rho^{(\ell -1) \gamma_2 /2} V^{1/2}(x) \eqsp . \numberthis \label{eq:majo_1_poiss}
\end{align}
For the first term in \eqref{eq:decompo_K_gamma_ell_XI}, using \Cref{assum:condition_majo_V}-\ref{assum:condition_majo_V_ii}, \Cref{lemma:jensen}, \eqref{eq:distance_pi} and that for any $\theta \in \Theta$ and $x \in \rset^d$, $\normLigne{H_{\theta}(x)} \leq V^{1/4}(x) \leq V^{1/2}(x)$, we obtain for any $\theta_1, \theta_2 \in \Theta$, $\gamma_1, \gamma_2 \in \ocint{0,\bgamma}$ with $\gamma_1 > \gamma_2$, $x \in \rset^d$ and $\ell \in \nsets$
\begin{align}
  &\norm{(\pi_{\gamma_1,\theta_1} - \pi_{\gamma_2,\theta_2})(\Kker_{\gamma_2,\theta_2}^\ell H_{\theta_1}(x)-\pi_{\gamma_2,\theta_2}(H_{\theta_1}))}  \\ & \phantom{aaaa} \leq 2 A_2   \rho^{\ell \gamma_2/2} \Vnorm[V^{1/2}]{\pi_{\gamma_1,\theta_1} - \pi_{\gamma_2,\theta_2} } \\  & \phantom{aaaa} \leq 2 A_2  C_{c,1}\, \rho^{\ell \gamma_2/2} \gamma^{-1}_2\defEns{ \Lambdabf_1(\gamma_1, \gamma_2)  + \Lambdabf_2(\gamma_1, \gamma_2)\| \theta_1 - \theta_2 \| } \eqsp. \label{eq:majo_2_poiss}
\end{align}
For the second term in \eqref{eq:decompo_K_gamma_ell_XI},  using \Cref{assum:H_lip},  \Cref{assum:condition_majo_V}-\ref{assum:condition_majo_V_ii} and \Cref{lemma:jensen}, we obtain for any $\theta_1, \theta_2 \in \Theta$, $\gamma_1, \gamma_2 \in \ocint{0,\bgamma}$ with $\gamma_1 > \gamma_2$, $x \in \rset^d$ and $\ell \in \nsets$
\begin{equation}
  \norm{\Kker_{\gamma_2,\theta_2}^{\ell}(H_{\theta_1} - H_{\theta_2})(x) - \pi_{\gamma_2,\theta_2}(H_{\theta_1} - H_{\theta_2})} \leq 2 A_2 L_H \rho^{\ell \gamma_2/2} \| \theta_1 - \theta_2\|V^{1/2}(x) \eqsp . \label{eq:majo_3_poiss} 
\end{equation}
Combining \eqref{eq:decompo_K_gamma_ell_XI}, \eqref{eq:majo_1_poiss}, \eqref{eq:majo_2_poiss}, \eqref{eq:majo_3_poiss} in \eqref{eq:decompo_K_gamma_ell} and using \Cref{lem:tech_sum_finie}, we obtain that  for any $\theta_1, \theta_2 \in \Theta$, $\gamma_1, \gamma_2 \in \ocint{0,\bgamma}$ with $\gamma_1 > \gamma_2$, $x \in \rset^d$ that 
\begin{multline}
  \norm{\Kker_{\gamma_1,\theta_1}\poissp{\gamma_1}{\theta_1}(x) - \Kker_{\gamma_2,\theta_2}\poissp{\gamma_2}{\theta_2}(x)}  \\
  \leq C_{c,2}\, \gamma_2^{-1} \parentheseDeux{\gamma_2^{-1} \defEns{ \Lambdabf_1(\gamma_1, \gamma_2)  + \Lambdabf_2(\gamma_1, \gamma_2)\| \theta_1 - \theta_2 \| } + \| \theta_1 - \theta_2\|}V^{1/2}(x) \eqsp , 
\end{multline}
with
\begin{equation}
  C_{c,2}\, = 4 A_2 \log^{-1}(1/\rho) \rho^{-\bgamma/2} \,  \max\parentheseDeux{L_H, C_{c,1} +  2 A_2    \log^{-1}(1/\rho)\rho^{-\bgamma/2}  } \eqsp .
\end{equation}
Since for any $k \in \nset$, $\norm{\ttheta_{k+1}-\ttheta_k} \leq \delta_{k+1}  V^{1/2}(\tX_{k+1})$ by \eqref{eq:algo_x_fix} and the fact that for any $\theta \in \Theta$ and $x \in \rset^d$, $\normLigne{H_{\theta}(x)} \leq V^{1/2}(x)$ and that $\Pi_{\Theta}$ is non-expansive, we get that for any $k \in \nset$,
\begin{multline}
  \norm{\Kker_{\gamma_k,\ttheta_k}\poissp{\gamma_k}{\ttheta_k}(\tX_{k+1}) - \Kker_{\gamma_{k+1},\ttheta_{k+1}}\poissp{\gamma_{k+1}}{\ttheta_{k+1}}(\tX_{k+1})} \\  \phantom{aaaa} \leq C_{c,2}  \gamma_{k+1}^{-1} \parentheseDeux{\gamma_{k+1}^{-1} \defEns{ \Lambdabf_1(\gamma_k, \gamma_{k+1})  + \Lambdabf_2(\gamma_k, \gamma_{k+1})\delta_{k+1}} + \delta_{k+1}}V(\tX_{k+1}) \eqsp ,
\end{multline}
which implies by \eqref{eq:tetan_def_2} and using \Cref{assum:condition_majo_V}-\ref{assum:condition_majo_V_i}  that
\begin{equation}\expe{\norm{\teta^{(c)}}} \leq C_2 \gamma_{k+1}^{-1} \parentheseDeux{\gamma_{k+1}^{-1} \defEns{ \Lambdabf_1(\gamma_k, \gamma_{k+1})  + \Lambdabf_2(\gamma_k, \gamma_{k+1})\delta_{k+1}} + \delta_{k+1}}\eqsp , \end{equation} with $C_2$ given by \eqref{eq:C2}.

\end{proof}

\begin{lemma}
  \label{lemma:cv_norm_d}
  Assume \tup{\Cref{assum:theta_compact}}, \tup{\Cref{assum:f_diff}}, \tup{\Cref{assum:grad_expec}}, \tup{\Cref{assum:condition_majo_V}} and that for any $\theta \in \Theta$ and $x \in \rset^d$, $\normLigne{H_{\theta}(x)} \leq V^{1/4}(x)$. Then we have for any $n \in \nset$
    \begin{equation}
      \expe{\norm{\tetad{n}}} \leq \Psibf(\gamma_n) \eqsp .
  \end{equation}
\end{lemma}

\begin{proof}
  By a straightforward application of \Cref{assum:condition_majo_V}-\ref{assum:condition_majo_V_iii} and by \eqref{eq:tetan_def_2} along with the fact that for any $\theta \in \Theta$ and $x \in \rset^d$, $\normLigne{H_{\theta}(x)} \leq V^{1/4}(x)$  we have for any $n \in \nset$, $\expe{\norm{\teta_n^d}} \leq \Psibf(\gamma_n)$.
\end{proof}

We now turn to the proof of \Cref{thm:salem_cv_fix}.  
  \begin{proof}[Proof of \Cref{thm:salem_cv_fix}]
      \begin{enumerate}[wide, labelwidth=!, labelindent=0pt, label=(\alph*)]
      \item To apply \cite[Theorem 2]{atchade2017perturbed}, it is enough to
        show that the following series converge \as
        \begin{equation}
    \sum_{n =0}^{\plusinfty} \delta_{n+1} \langle \Pi_{\Theta} ( \theta_{n} - \delta_{n+1} \nabla f (\theta_{n})) -\theta^{\star}, \teta_{n}^{(i)} \rangle \eqsp , \quad \sum_{n=0}^{\plusinfty} \delta_{n+1}^2 \| \teta_n \|^2  \eqsp ,
  \end{equation}
  with $\theta^{\star} \in \argmin_{\theta \in \Theta}f(\theta)$.  $\sum_{n=0}^{\plusinfty} \delta_{n+1}^2 \| \teta_n \|^2< +\infty$ \as \ by \Cref{lemma:cv_sq_fix} since $\sum_{n \in \nset} \delta_{n+1}^2 < +\infty$. Since $(\langle \Pi_{\Theta} ( \theta_{n} - \delta_{n+1} \nabla f (\theta_{n})) -\theta^{\star}, \tetaa{n} \rangle)_{n \in \nset}$ is a $(\tilde{\mathcal{F}}_n)_{n \in \nset}$-martingale increment, see \eqref{eq:def_tmcf} and by \Cref{lemma:cv_norm_a} and $\sum_{n \in \nset} \delta_{n+1}^2 / \gamma_n^2 < +\infty$
  \begin{equation}
    \expe{\sum_{n =0}^{+\infty} \delta_{n+1}^2 \langle \Pi_{\Theta} ( \theta_{n} - \delta_{n+1} \nabla f (\theta_{n})) -\theta^{\star}, \tetaa{n} \rangle^2} < +\infty \eqsp ,
  \end{equation}   we obtain using \cite[Section 12.5]{williams1991prob} that $\sum_{n =0}^{\plusinfty} \delta_{n+1} \langle \Pi_{\Theta} ( \theta_{n} - \delta_{n+1} \nabla f (\theta_{n})) -\theta^{\star}, \tetaa{n} \rangle$ converges \as . Using \Cref{assum:theta_compact}, \eqref{eq:condition_cv_fix_1} and \Cref{lemma:cv_norm_c_d} and \Cref{lemma:cv_norm_d} we get that $\sum_{n =0}^{\plusinfty} \delta_{n+1} \langle \Pi_{\Theta} ( \theta_{n} - \delta_{n+1} \nabla f (\theta_{n})) -\theta^{\star}, \teta_{n}^{(i)} \rangle$ is absolutely convergent \as \ for $i \in \{c, d \}$. Finally $\sum_{n =0}^{\plusinfty} \delta_{n+1} \langle \Pi_{\Theta} ( \theta_{n} - \delta_{n+1} \nabla f (\theta_{n})) -\theta^{\star}, \tetab{n} \rangle$ converges \as \ by \Cref{lemma:cv_norm_b}-\ref{lemma:cv_norm_b_ii}.
\item The proof of is identical to the one of \Cref{thm:salem_cv}-\ref{thm:salem_cv_ii}.
\end{enumerate}
\end{proof}

\subsection{Proof of \Cref{thm:salem_cv_control_fix}}
\label{thm:salem_cv_control_fix_proof}

  The proof is similar to the one of \Cref{thm:salem_cv_control}, using \Cref{lemma:cv_sq_fix}, \Cref{lemma:cv_norm_b}, \Cref{lemma:cv_norm_c_d}, \Cref{lemma:cv_norm_d} and the fact that $\tetaa{n}$ is a $(\tilde{\mathcal{F}}_n)_{n \in \nset}$-martingale increment, see \eqref{eq:def_tmcf}.


  \subsection{Proof of \Cref{thm:salem_langevin_cv}}
  \label{sec:proofs-crefs-meth}
  In this section, we give the proof of \Cref{thm:salem_langevin_cv}
  by showing that \Cref{assum:condition_majo_V} holds.  First of all in
  \Cref{sec:fost-lyap-drift}, we
  establish under \Cref{assum:glip} and \Cref{assum:curvature} stability results uniform in the
  parameter $\theta \in \Theta$ for the Langevin diffusion
 \eqref{eq:langevin} and the associated Euler-Maruyama discretization
  \eqref{eq:euler_maruyama_langevin} based on a Foster-Lyapunov drift
  condition with constants independent of $\theta$.  Then, in
  \Cref{sec:check-cref}, we show that the
  stability conditions that we derive, are sufficient to prove that
  \Cref{assum:condition_majo_V} holds. The proof of \Cref{thm:salem_langevin_cv} then consists in
  combining all these results and is presented in \Cref{sec:proof-crefthm:s-3}.
  
  Under \Cref{assum:glip} and \Cref{assum:curvature}, for any $\theta \in \Theta$,
\eqref{eq:langevin} defines a Markov semi-group
$(\Pker_{t, \theta})_{t \geq 0}$ for any $x \in \rset^d$ and
$\msa \in \mcb{\rset^d}$ by
$\Pker_{t,\theta}(x,\msa) = \probaLigne{Y_{t}^{\theta} \in \msa}$ where
$(Y_{t}^{\theta})_{t \geq 0}$ is the solution of \eqref{eq:langevin}
with $Y_{0}^{\theta} = x$. Consider now the generator of
$(\Pker_{t, \theta})_{t \geq 0}$ for any $\theta \in \Theta$, defined for
any $f \in  \rmc^2(\rset^d)$ by
\begin{equation}
\label{eq:def_generator_theta}
  \generator_{\theta} f = -\ps{\nabla_x f}{\nabla_x U_{\theta}(x)} + \Delta_x f \eqsp.
\end{equation}


We say that a Markov kernel $\Rker$ on $\rset^d\times \mcb{\rset^d}$ satisfies a discrete Foster-Lyapunov drift condition
\hypertarget{assum:drift_discrete}{$\bfDd(V,\lambda,b)$} if there exist $\lambda \in (0,1)$, $b\geq0$ and a measurable function $V: \rset^d \to \coint{1,+\infty}$ such that for all $x \in \rset^d$
\begin{equation}
  \label{eq:discrete_drift}
  \Rker V(x) \leq \lambda V(x) + b \eqsp.
\end{equation}

We say that a Markov semi-group $(\Pker_t)_{t \geq 0}$ on $\rset^d \times \mcb{\rset^d}$ with extended infinitesimal generator $(\mathcal{A},\domain(\generator))$ (see \eg~\cite{meyn1993criteria_iii} for the definition of $(\generator,\domain(\generator))$) satisfies a continuous drift condition \hypertarget{assum:drift_continuous}{$\bfDc(V,\zeta,\beta)$} if there exist $\zeta >0$, $\beta \geq 0$ and a measurable function $V: \rset^d \to \coint{1,+\infty}$ with $V \in \domain(\generator)$ such that for all $x \in \rset^d$
\begin{equation}
  \label{eq:continuous_drift}
  \mathcal{A}V(x) \leq - \zeta V(x) + \beta \eqsp .
\end{equation}

\subsubsection{Foster-Lyapunov drift conditions uniform on $\theta$}
\label{sec:fost-lyap-drift}
Define $\Ve: \rset^d \to \coint{1,+\infty}$  for all $x \in \rset^d$ by
\begin{equation}
\label{eq:def_V_e}
  \Ve(x) = \exp(\tmttdeux \phi(x)) \eqsp, \quad \text{
    with $\phi(x) = \sqrt{1 + \|x\|^2 }$ and $\tmttdeux = \mttdeux /4$ } \eqsp.
\end{equation}
\begin{proposition}  
  \label{propo:lyap_condition}
  Assume \tup{\Cref{assum:glip}} and \tup{\Cref{assum:curvature}}.     Let $\bar{\gamma} < \min(1,2  \mtttrois) $. 
Then there exist $\lambda_{\rme} \in \ooint{0,1}$ and $b_{\rme} \geq 0$ such that for all $\gamma \in \ocint{0,\bgamma}$ and $\theta \in \Theta$  the Markov kernel $\Rker_{\gamma, \theta}$ associated with  the recursion \eqref{eq:euler_maruyama_langevin} satisfies the discrete drift condition \hyperlink{ass:drift_discrete}{$\bfDd(V,\lambda^{\gamma},b\gamma)$}, \ie \ for all $x \in \rset^d$
  \begin{equation}
    \label{eq:drift_super_exp_log_sob}
    \Rker_{\gamma, \theta}\Ve(x) \leq\lambda_{\rme}^{\gamma}\Ve(x) + b_{\rme} \gamma  \1_{\ball{0}{r_{\rme}}}(x)  \eqsp,
  \end{equation}
  with
  \begin{align}
    \lambda_{\rme} &= \rme^{-\tmttdeux^2(2^{1/2} - 1)} \eqsp , \quad r_{\rme} = \max(1,2(d+\ccur)/\mttdeux,\Rdeux) \eqsp,
    \\
    b_{\rme} &= \tmttdeux (d + c + 2^{1/2} \tmttdeux) \exp\parentheseDeux{\tmttdeux \defEns{(d + c+ \tmttdeux) \bgamma + \sqrt{1+r_{\rme}^2} }} 
     \eqsp .
  \end{align}
\end{proposition}

\begin{proof}
  Since $\phi$ is $1$-Lipschitz, by the log-Sobolev inequality \cite[Proposition 5.4.1]{bakry:gentil:ledoux:2014},  we have for any $x \in \rset^d $ and $\theta \in \Theta$,
  \begin{align}
\label{eq:appli_log_sob}
    \Rker_{\gamma,\theta}\Ve(x) &\leq \exp \parentheseDeux{\tmttdeux \Rker_{\gamma,\theta}\phi(x) + \tmttdeux^2\gamma} \\ &\leq\exp\parentheseDeux{\tmttdeux \sqrt{\| x - \gamma \nabla_x U_{\theta}(x)\|^2 + 2\gamma d +1} + \tmttdeux^2\gamma} \eqsp,
  \end{align}
  where we have used Jensen's inequality in the last line. 
  Second, using \Cref{assum:curvature} and $\gamma < 2\mtttrois$, we obtain that for any $x \in \rset^d$ and $\theta \in \Theta$,
  \begin{align}
    \| x - \gamma \nabla_x U_{\theta}(x&) \|^2 \leq\| x \|^2 -2 \gamma \langle x, \nabla_x U_{\theta}(x) \rangle + \gamma^2\| \nabla_xU_{\theta}(x) \|^2 \\
                               &\leq\| x \|^2 - 2 \mttdeux \gamma \| x \|\1_{\ball{0}{\Rdeux}^{\complementary}}(x)  + \gamma (\gamma - 2 \mtttrois) \| \nabla_xU_{\theta}(x) \|^2+2 \gamma \ccur \\
                               &\leq\|x\|^2 -2\mttdeux \gamma \|x\| \1_{\ball{0}{\Rdeux}^{\complementary}}(x)  + 2 \gamma \ccur \eqsp .
  \end{align}
  Therefore,   using  for any $a >0$, $\sqrt{1+a}-1 \leq a/2$, we get for any $x \in \rset^d$ and  $\theta \in \Theta$,
  \begin{align}  
    &      \sqrt{\| x - \gamma \nabla_xU_{\theta}(x) \|^2 + 2\gamma d + 1} - \phi(x)  \\
\label{eq:appli_log_sob_2}
   &\leq\phi(x) \left\lbrace \sqrt{1 +2\gamma\phi^{-2}(x) (d+ \ccur- \mttdeux\|x\| \1_{\ball{0}{\Rdeux}^{\complementary}}(x) )} - 1\right\rbrace  \\ &\leq\gamma \phi^{-1}(x)(d + \ccur- \mttdeux \| x \| \1_{\ball{0}{\Rdeux}^{\complementary}}(x) )  \eqsp.
  \end{align}
  Therefore, combining this result with \eqref{eq:appli_log_sob} and using that for any $\tx \in \cball{0}{r_{\rme}}^{\complementary}$, $\phi(\tx)^2/\norm[2]{\tx} \leq 2$ and $d+c \leq \mttdeux \norm{x} /2$, we obtain for any $x \in \cball{0}{r_{\rme}}^{\complementary}$ and $\theta \in \Theta$,
  \begin{multline}
    \Rker_{\gamma,\theta}\Ve(x) \leq \exp \parentheseDeux{\tmttdeux \phi^{-1}(x)(d + \ccur- \mttdeux \| x \|  ) + \tmttdeux^2 \gamma } \Ve(x) \\ \leq     \exp \parentheseDeux{- 2 \tmttdeux^2 \gamma\phi^{-1}(x) \|x \| + \tmttdeux^2 \gamma } \Ve(x) 
     \leq     \lambda_{\rme}^{\gamma} \Ve(x) \eqsp. 
   \end{multline}
Using \eqref{eq:appli_log_sob},  \eqref{eq:appli_log_sob_2}, and the fact that  $\phi(\tx) \geq 1$ for any $\tx \in \rset^d$, we have for any $x \in \ball{0}{r_{\rme}}$ and $\theta \in \Theta$,
  \[ \Rker_{\gamma, \theta}\Ve(x) \leq\lambda_{\rme}^{\gamma}\Ve(x) + \left( \rme^{\tmttdeux (d + c + \tmttdeux) \gamma } - \lambda_{\rme}^{\gamma} \right)\exp\parentheseDeux{\tmttdeux\sqrt{1+r_{\rme}^2}} \eqsp .\]
The proof of \eqref{eq:drift_super_exp_log_sob} for   $x \in \ball{0}{r_{\rme}}$ and $\theta \in \Theta$ is then completed upon using that  $\rme^a - \rme^b \leq (a - b) \rme^a$ for all $a,b \in \rset$ with $a \geq b$.
\end{proof}

\begin{proposition}
  \label{propo:unif_ergo_c}
  Assume \tup{\Cref{assum:glip}} and \tup{\Cref{assum:curvature}}. Then for any $\theta \in \Theta$, $(\Pker_{t,\theta})_{t \geq 0}$ associated with \eqref{eq:langevin} satisfies the continuous drift condition
   \hyperlink{assum:drift_continuous}{$\bfDc(\Ve,\zeta_{\rme},\beta_{\rme})$}  for $\Ve$ defined in \eqref{eq:def_V_e} and 
      \begin{equation}
    \zeta_{\rme} = 3 \tmttdeux^2/2^{1/2} \eqsp , \quad \beta_{\rme} = \tmttdeux \exp\parentheseDeux{\tmttdeux\sqrt{1+\tr_{\rme}^2}} (1+\tmttdeux + \ccur + d )\eqsp, \quad \tr_{\rme} = \max(1,\Rdeux) \eqsp .
  \end{equation}

\end{proposition}

\begin{proof}
  First, by definition, for any $x \in \rset^d$, we have
  \begin{align}
    \nabla_x V(x) &= \tmttdeux x V(x) / \phi(x) \\
    \Delta_x V(x) &= \{\tmttdeux V(x) / \phi(x)\}\{\tmttdeux \norm[2]{x}/\phi(x) + d - \norm[2]{x}/\phi^2(x)\}  \eqsp. 
  \end{align}
  Therefore, by \eqref{eq:def_generator_theta} and \Cref{assum:curvature}, we get for any $\theta \in \Theta$ and $x \in \rset^d$,
  \begin{align}
    &\generator_{\theta} V(x)  = \{\tmttdeux V(x) / \phi(x)\} \parentheseDeux{ - \ps{\nabla_x U_{\theta}(x)}{x} + \tmttdeux \norm[2]{x}/\phi(x) + d - \norm[2]{x}/\phi^2(x)} \\
                             & \leq \{\tmttdeux V(x) / \phi(x)\} \parentheseDeux{ - \mttdeux \norm{x} \1_{\boule{0}{\Rdeux}^{\complementary}}(x) + \ccur  + \tmttdeux \norm[2]{x}/\phi(x) + d - \norm[2]{x}/\phi^2(x)} \\
                             & \leq  \{\tmttdeux V(x) / \phi(x)\} \parentheseDeux{ - (3\mttdeux/4) \norm{x} \1_{\boule{0}{\Rdeux}^{\complementary}}(x) + \ccur  + \tmttdeux \norm{x}\1_{\boule{0}{\Rdeux}}(x) + d } \eqsp. 
  \end{align}
The proof is then complete upon using that for any $x \in \ball{0}{\tr_{\rme}}^{\complementary}$, $\norm{x} /\phi(x) \geq 2^{-1/2}$, for any $y \in \rset^d$, $\norm{y}/ \phi(y) \leq 1$.
\end{proof}
\subsubsection{Checking \Cref{assum:condition_majo_V}}
\label{sec:check-cref}

\begin{lemma}
  \label{lemma:V_control}
  Assume \tup{\Cref{assum:glip}} and let  $V: \ \rset^d \rightarrow [1,+\infty)$ satisfying $\lim_{\| x \| \to +\infty} V(x) = + \infty$ and $V \in \domain(\A_{\theta})$, for any $\theta \in \Theta$, where $\generator_{\theta}$ is defined by  \eqref{eq:def_generator_theta}. .
  \begin{enumerate}[label=(\alph*)]
    \item \label{item_V_control_i} Assume that there exist $\lambda \in \ooint{0,1}$, $b \geq 0$ and $\bgamma >0$ such that for any $\theta \in \Theta$ and $\gamma \in \ocint{0,\bgamma}$, $\Rker_{\gamma,\theta}$ associated with the recursion \eqref{eq:langevin_discrete}, satisfies \hyperlink{ass:drift_discrete}{$\bfDd(V,\lambda^{\gamma},b\gamma)$}. Then for any $\theta \in \Theta$ and $\gamma \in \ocint{0,\bgamma}$, $\Rker_{\gamma, \theta}$ admits an invariant probability measure $\pi_{\gamma, \theta}$ on $(\rset^d, \mathcal{B}(\rset^d))$ and there exists $D_3 \geq 0$ such that for any $x \in \rset^d$ and $k\in \N$
  \begin{equation}
\updelta_x \Rker_{\gamma, \theta}^k V \leq D_3 + V(x)  \eqsp , \qquad \pi_{\gamma, \theta} (V) \leq D_3 \eqsp , \qquad  D_3 = b\lambda^{-\bgamma} /\log(1/\lambda) \eqsp .
\end{equation}
In addition, for all $\theta \in \Theta$ and $x \in \rset^d$, $\lim_{k \to +\infty} \| \updelta_x \Rker_{\gamma, \theta}^k - \pi_{\gamma, \theta} \|_V =0$. 
  \item \label{item_V_control_ii}Assume that there exist $\zeta >0$ and $\beta \geq 0$ such that for any $\theta \in \Theta$, $(\Pker_{t,\theta})_{t \geq 0}$ associated with \eqref{eq:langevin}  satisfies \hyperlink{assum:drift_continuous}{$\bfDc(V,\zeta,\beta)$}. Then for any $\theta \in \Theta$, the diffusion is non-explosive, $\A_{\theta}$ admits $\pi_{\theta}$ as an invariant probability measure and
  \begin{equation}
\pi_{\theta} (V) \leq D_0 \eqsp , \qquad D_0 = \beta / \zeta \eqsp .
\end{equation}
In addition, for all $\theta \in \Theta$ and $x \in \rset^d$, $\lim_{t \to +\infty} \Vnorm{\updelta_x\Pker_{t, \theta} - \pi_{\theta}} = 0$.
    \end{enumerate} 
\end{lemma}

\begin{proof}
  \begin{enumerate}[wide, labelwidth=!, labelindent=0pt, label=(\alph*)]
  \item for any $\gamma \in \ocint{0,\bgamma}$ and $\theta \in \Theta$, $\Rgt$ is irreducible with respect to the Lebesgue measure on $\rset^d$, has the Feller property and satisfies \hyperlink{ass:drift_discrete}{$\bfDd(V,\lambda^{\gamma},b\gamma)$} then \cite[Section 4.4]{meyn1993criteria_i} applies and $\Rker_{\gamma, \theta}$ admits an invariant probability measure $\pi_{\gamma, \theta}$.
  The discrete drift condition and \cite[Lemma 1]{durmus2017unadjusted} give that for any $\gamma \in \ocint{0,\bgamma}$ and $\theta \in \Theta$
  \begin{equation}
    \Rker_{\gamma, \theta}^k V(x) \leq V(x) + b \lambda^{-\bgamma}/\log(1/\lambda) \eqsp , \qquad \pi_{\gamma, \theta} (V) \leq b \lambda^{-\bgamma}/\log(1/\lambda) \eqsp .
  \end{equation}
  We obtain that for all $\theta \in \Theta$ and $x \in \rset^d$, $\lim_{k \to +\infty} \| \updelta_x \Rker_{\gamma,\theta}^{k} - \pi_{\gamma, \theta} \|_V =0$ using \cite[Theorem 16.0.1]{meyn1993markov}.
\item Using \hyperlink{assum:drift_continuous}{$\bfDc(V,\zeta,\beta)$} and \cite[Theorem 2.1]{meyn1993criteria_iii} we get that the diffusion process is non-explosive and thus $(\Pker_{t,\theta})_{t \geq 0}$ is defined for any $\theta \in \Theta$ and $t \geq 0$. Using \cite[Corollary 10.1.4]{stroock1997multi} for any $\theta \in \Theta$, $(\Pker_{t,\theta})_{t \geq 0}$ is strongly Feller continuous, therefore any compact sets is petite for the Markov kernel $\Pker_{h, \theta}$, for any $h >0$ and $\theta \in \Theta$,  by \cite[Theorem 6.0.1]{meyn1993markov}. Using \cite[Chapter 7, Proposition 1.5]{revuz1999continuous}, \cite[Chapter 4, Theorem 9.17]{ethier1986characterization}, and the fact that $\pi_{\theta}(\A_{\theta}f) = 0$ for any $\theta \in \Theta$ and $f \in \mrc_c^2(\rset^d)$, we obtain that for any $\theta \in \Theta$, $\pi_{\theta}$ is an invariant measure for $(\Pker_{t,\theta})_{t \geq 0}$. Using \hyperlink{assum:drift_continuous}{$\bfDc(V,\zeta,\beta)$} and \cite[Theorem 4.5]{meyn1993criteria_iii} we get that for all $\theta \in \Theta$, $ \pi_{\theta}(V) \leq \beta / \zeta$. Finally, the convergence is ensured using \cite[Theorem 5.1]{meyn:tweedie:1993:III}.
  \end{enumerate}
\end{proof}

As an immediate corollary we obtain that under the conditions of \Cref{lemma:V_control} for any $\theta \in \Theta$, $\gamma \in \ocint{0,\bgamma}$ and $k \in \N$,
\begin{equation}\pi_{\theta} \Rker_{\gamma, \theta}^k V \leq \beta / \zeta + b\lambda^{-\bgamma} / \log(1/\lambda) \eqsp . \label{eq:majo_combi}\end{equation}

\begin{lemma}
  \label{lemma:v_norm_soul}
  Let $V: \ \rset^d \rightarrow [1,+\infty)$. Assume there exist $\lambda \in \ooint{0,1}$, $b \geq 0$ and $\bgamma >0$ such that for any $\theta \in \Theta$ and $\gamma \in \ocint{0,\bgamma}$, $\Rker_{\gamma,\theta}$ associated with the recursion \eqref{eq:euler_maruyama_langevin} satisfies \hyperlink{ass:drift_discrete}{$\bfDd(V,\lambda^{\gamma},b\gamma)$}.    Let $(\gamma_n)_{n \in \nset}$,
   $(\delta_n)_{n \in \nset}$ be sequences of non-increasing positive real numbers and $(m_n)_{n \in \nset}$ be a sequence of positive integers satisfying 
   $\sup_{n \in \nset} \gamma_n < \bgamma$. Then, 
 $(X_k^n)_{n \in \N, k \in \lbrace 0, \dots, m_n \rbrace}$ given by \eqref{eq:algo_SOUL} with $\{\Kker_{\gamma,\theta} \, : \, \gamma \in \ocint{0,\bgamma}, \theta \in \Theta\}=\{\Rker_{\gamma,\theta} \, : \, \gamma \in \ocint{0,\bgamma}, \theta \in \Theta\}$ satisfies for all $p, n \in \N$ and $k \in \lbrace 0, \dots, m_n \rbrace$
  \begin{equation}
    \CPE{\Rker_{\gamma_n, \theta_n}^p V(X_k^n)}{X_0^0} \leq D_1 V(X_0^0) \eqsp , \qquad D_1 = 1 + 2b \lambda^{-\bgamma}/\log(1/\lambda) \eqsp .
  \end{equation}
\end{lemma}

\begin{proof}
By induction we obtain that
  \begin{equation}
    \label{eq:drift_check_h1_a}
    \CPE{V(X_k^{n+1})}{\mathcal{F}_n}= \Rker_{\gamma_{n+1},\theta_{n+1}}^kV(X_0^{n+1})    \leq \lambda^{k \gamma_{n+1}}V(X_0^{n+1}) + b \gamma_{n+1} \sum_{i=1}^k \lambda^{\gamma_{n+1}(k- i)} \eqsp,
  \end{equation}
  where $(\mcf_n)_{n \in \nset}$ is defined by \eqref{eq:def_F_n}. 
Similarly,  we obtain for any $k \in \lbrace 0, \dots, m_0 \rbrace$,
  \begin{equation}
    \label{eq:drift_check_h1_b}
    \CPE{V(X_k^{0})}{X_0^0} = \Rker_{\gamma_{0}, \theta_0}^kV(X_0^0) \leq \lambda^{k \gamma_{0}}V(X_0^{0}) + b \gamma_{0} \sum_{i=1}^k \lambda^{\gamma_{0}(k- i)} \eqsp.
      \end{equation}
Define for $\ell \in \nset, k \in \nset$ and $i \in \nsets$, $\mathrm{q}_{\ell,k} = \sum_{j=0}^{\ell-1} m_j + k$, $\mathrm{q}_n = \mathrm{q}_{\ell,0}$ and $\tilde{\gamma}_i = \sum_{j=0}^{+\infty} \gamma_j \1_{(\mathrm{q}_j, \mathrm{q}_{j+1}]}(i) $. In addition, consider for any $\mathrm{p}, \mathrm{q} \in \nsets$,  $\Gamma_{\mathrm{p},\mathrm{q}} = \sum_{i=\mathrm{p}}^{\mathrm{q}} \tilde{\gamma}_i$ and $\Gamma_{\mathrm{p}} = \Gamma_{1,\mathrm{p}}$. Combining \eqref{eq:drift_check_h1_a}, \eqref{eq:drift_check_h1_b} and \Cref{lem:tech_sum_finie} we get for any $n \in \N$ and $k \in \lbrace 0,\dots,m_n \rbrace$
\begin{align}
      \label{eq:drift_check_h1_c}
  \CPE{\Rker_{\gamma_n, \theta_n}^pV(X_k^n)}{X_0^0} &\leq \lambda^{\gamma_n p} \CPE{V(X_k^n)}{X_0^0} + b\log(1/\lambda)\lambda^{-\bgamma}
\\ &\leq \lambda^{\Gamma_{\mathrm{q}_{n,k}}} V(X_0^0) + b \sum_{i=1}^{\mathrm{q}_{n,k}} \tilde{\gamma}_i \lambda^{\Gamma_{i+1,\mathrm{q}_{n,k}}} + b\log(1/\lambda)\lambda^{-\bgamma}\eqsp .
  \end{align}
  Since $(\tilde{\gamma}_i)_{i \in \N}$ is nonincreasing and for all $t \geq 0$, $1 - \lambda^t \geq - t \lambda^t \log(\lambda) $, we have for all $\mathrm{q} \in \nsets$, 
\begin{align}
&\sum_{i=1}^{\mathrm{q}} \tilde{\gamma}_i \lambda^{ \Gamma_{i+1,\mathrm{q}}} \leq \sum_{i=1}^\mathrm{q} \tilde{\gamma}_i \prod_{j=i+1}^\mathrm{q} (1+ \lambda^{\tilde{\gamma}_1}  \log(\lambda) \tilde{\gamma}_j) \\
&\leq (- \lambda^{\tilde{\gamma}_1}\log(\lambda))^{-1}\sum_{i=1}^\mathrm{q} \defEns{\prod_{j=i+1}^\mathrm{q} (1+\lambda^{\tilde{\gamma}_1} \log(\lambda)  \tilde{\gamma}_j) -\prod_{j=i}^\mathrm{q} (1+  \lambda^{\tilde{\gamma}_1} \log(\lambda) \tilde{\gamma}_j)} \\
& \leq  (-\lambda^{\tilde{\gamma}_1}\log(\lambda) )^{-1} \eqsp.
\end{align}
Combining this result and \eqref{eq:drift_check_h1_c} completes the proof. 
\end{proof}

\begin{lemma}
  \label{lemma:kl_error}
Let $V: \ \rset^d \to \coint{1,+\infty}$ measurable and  $M_V \geq 0$ such that $\sup_{x \in \rset^d} \defEns{(1 + \| x \|)^{2}/ V(x) }\leq M_V$.  
  Assume \tup{\Cref{assum:glip}} and that for any $\theta \in \Theta$, $\gamma \in \ocint{0, \bgamma}$ and $k \in \nset$,
  \begin{equation}
    \label{eq:majo_tout_k}
    \pi_{\theta} \Rker_{\gamma, \theta}^{k}(V) \leq \tilde{D}_1 \eqsp , \qquad \pi_{\theta} \Pker_{\gamma m_{\gamma}, \theta} V  \leq \tilde{D}_1 \eqsp ,
  \end{equation}
    with $m_{\gamma} = \ceil{1/\gamma}$.
  Then for any $\theta \in \Theta$ and $\gamma \in \ocint{0, \bgamma}$
  \begin{multline}
    \Vnorm[V^{1/2}]{\pi_{\theta} \Rker_{\gamma, \theta}^{m_{\gamma}} -\pi_{\theta} \Pker_{\gamma m_{\gamma}, \theta}}^2 \\ \leq 2  \tilde{D}_1\LUx^2 \gamma (1 + \bgamma) \defEns{d + 2\bgamma(\sup_{\theta \in \Theta} \norm[2]{\nabla_{x} U_{\theta}(0)} + \LUx^2 M_V \tilde{D}_1)} \eqsp ,
  \end{multline}

\end{lemma}

\begin{proof}
  The proof follows the lines of \cite[Theorem 10]{durmus2017unadjusted}.
  Let $\theta \in \Theta$ and $\gamma \in \ocint{0, \bgamma}$. We have, using a generalized Pinsker inequality \cite[Lemma 24]{durmus2017unadjusted}, that
  \begin{align}
    \Vnorm[V^{1/2}]{\pi_{\theta} \Rker_{\gamma, \theta}^{m_{\gamma}} -\pi_{\theta} \Pker_{\gamma m_{\gamma}, \theta}}^2 &\leq 2 (\pi_{\theta} \Rker_{\gamma, \theta}^{m_{\gamma}}V + \pi_{\theta} \Pker_{\gamma m_{\gamma}, \theta}V) \KL{\pi_{\theta} \Rker_{\gamma, \theta}^{m_{\gamma}}}{\pi_{\theta} \Pker_{\gamma m_{\gamma}, \theta}} \eqsp . \\ &\leq 4 \tilde{D}_1 \KL{\pi_{\theta} \Rker_{\gamma, \theta}^{m_{\gamma}}}{\pi_{\theta} \Pker_{\gamma m_{\gamma}, \theta}} \eqsp .
  \end{align}
  Using \Cref{assum:glip}, \cite[Equation (15)]{durmus2017unadjusted}, \cite[Theorem 4.1, Chapter 2]{kullback1997information}, \eqref{eq:majo_tout_k} and that for any $a,b \in \rset$, $(a+b)^2 \leq 2(a^2 + b^2)$ we obtain that
  \begin{align}
    \KL{\pi_{\theta} \Rker_{\gamma, \theta}^{m_{\gamma}}}{\pi_{\theta} \Pker_{\gamma m_{\gamma}, \theta}} &\leq \LUx^2 m_{\gamma} \gamma^2(d + \bgamma \sup_{k \in \nset} \pi_{\theta} \Rker_{\gamma, \theta}^{k}\norm[2]{\nabla_x U_{\theta}(x)} ) \\
    &\leq \LUx^2 (1+ \bgamma) \gamma (d + 2\bgamma(\sup_{\theta \in \Theta} \norm[2]{\nabla_x U_{\theta}(0)} + \LUx^2 M_V \tilde{D}_1)) \eqsp ,
  \end{align}
  which concludes the proof.
\end{proof}

\begin{proposition}
  \label{propo:discrete_vs_continuous}
  Let  $V: \rset^d \to \coint{1,+\infty}$ measurable and $M_V \geq 0$ such that $\sup_{x \in \rset^d} \defEns{(1 + \| x \|)^{2}/ V(x) }\leq M_V$.
  Assume \tup{\Cref{assum:glip}} and that there exist $\lambda \in \ooint{0,1}$, $b\geq 0$ and $\bgamma >0$ such that for any $\theta \in \Theta$ and $\gamma \in \ocint{0,\bgamma}$ $\Rker_{\gamma,\theta}$ satisifies \hyperlink{ass:drift_discrete}{$\bfDd(V,\lambda^{\gamma},b\gamma)$}. Assume that there exists $D_0 \geq 0$ 
  such that for any $\theta \in \Theta$, $\pi_{\theta}(V) \leq D_0$
  . Then there exists $D_4 \geq 0$ such that for any $\theta \in \Theta$ and  $\gamma \in \ocint{0,\bgamma}$
  \begin{equation} \| \pi_{\gamma,\theta} - \pi_{\theta} \|_{V^{1/2}} \leq D_4 \gamma^{1/2} \eqsp .\end{equation}
\end{proposition}
\begin{proof}
  Using \Cref{lemma:V_control} we obtain that for any $\theta \in \Theta$
  \begin{equation} \underset{k  \rightarrow  +\infty}{\lim} \Vnorm[V^{1/2}]{\pi_{\theta} \Rker_{\gamma,\theta}^k - \pi_{\theta} \Pker_{\gamma k, \theta} } =\Vnorm[V^{1/2}]{\pi_{\gamma,\theta} - \pi_{\theta}} \eqsp . \label{eq:limite}\end{equation}
  We now give an upper bound on $\Vnorm[V^{1/2}]{\pi_{\theta} \Rker_{\gamma,\theta}^k - \pi_{\theta} \Pker_{\gamma k, \theta} }$ for $k = q_{\gamma} m_{\gamma} $ with $m_{\gamma} = \ceil{1/\gamma}$ and $q_{\gamma}\in \N$.
  Using \cite[Theorem 6]{debortoli2018back} and that $\pi_{\theta}$ is invariant for $\Pker_{t, \theta}$ with $t\geq 0$, see \Cref{lemma:V_control}, we obtain for all $\theta \in \Theta$, $\gamma \in \ocint{0,\bgamma}$ and $k \in \N$ 
  \begin{align}
    &\Vnorm[V^{1/2}]{\pi_{\theta} \Rker_{\gamma,\theta}^k - \pi_{\theta} \Pker_{\gamma k ,\theta} } \\ &\leq \sum_{\ell=0}^{q_{\gamma}-1} \Vnorm[V^{1/2}]{\pi_{\theta} \Pker_{\gamma (\ell+1) m_{\gamma}, \theta} \Rker_{\gamma,\theta}^{(q_{\gamma} -(\ell+1))m_{\gamma} } - \pi_{\theta} \Pker_{\gamma \ell m_{\gamma}, \theta} \Rker_{\gamma,\theta}^{(q_{\gamma} -\ell)m_{\gamma} }} \\
                                                                                            &\leq \sum_{\ell=0}^{q_{\gamma} - 1} C \xi^{\gamma m_{\gamma} (q_{\gamma} - (\ell+1))} \Vnorm[V^{1/2}]{\pi_{\theta} \Pker_{\gamma \ell m_{\gamma}, \theta} \Pker_{m_{\gamma} \gamma, \theta}- \pi_{\theta} \Pker_{\gamma \ell m_{\gamma}, \theta} \Rker_{\gamma,\theta}^{m_{\gamma}}} \\
    & \leq \Vnorm[V^{1/2}]{\pi_{\theta}  \Pker_{m_{\gamma} \gamma, \theta}- \pi_{\theta} \Rker_{\gamma,\theta}^{m_{\gamma}}} \sum_{\ell=1}^{q_{\gamma}} C \xi^{\ell \gamma m_{\gamma}}  
      \eqsp ,  \label{eq:sum_decompo}
  \end{align}
where $C \geq 0, \xi \in (0,1)$ are the constants given by \cite[Theorem 6]{debortoli2018back} with minorization condition given by \cite[Proposition 8a]{debortoli2018back} with $\mtt = -\mathtt{L}$ since  \Cref{assum:glip} holds and drift condition \hyperlink{ass:drift_discrete}{$\bfDd(V^{1/2},\lambda^{\gamma},b\lambda^{-\bgamma/2}\gamma/2)$},  since for all $\theta \in \Theta$ and $\gamma \in \ocint{0,\bgamma}$ we have that $\Rker_{\gamma, \theta}$ satisfies \hyperlink{ass:drift_discrete}{$\bfDd(V,\lambda^{\gamma},b\gamma)$} and therefore using Jensen's inequality that $\Rker_{\gamma, \theta}$ satisfies \hyperlink{ass:drift_discrete}{$\bfDd(V^{1/2},\lambda^{\gamma/2},b\lambda^{-\bgamma/2}\gamma/2)$}.

We now give an upper bound on  error $\Vnorm[V^{1/2}]{\pi_{\theta}  \Pker_{m_{\gamma} \gamma, \theta}- \pi_{\theta} \Rker_{\gamma,\theta}^{m_{\gamma}}}$. Indeed, since $\A_{\theta}$ satisfies a \hyperlink{ass:drift_continuous}{$\bfDc(V, \zeta, \beta)$} and
  $\Rker_{\gamma, \theta}$ satisfies \hyperlink{ass:drift_discrete}{$\bfDd(V,\lambda^{\gamma},b\gamma)$} for any $\theta \in \Theta$ and $\gamma \in \ocint{0, \bgamma}$, we obtain using \eqref{eq:majo_combi} that for any $\theta \in \Theta$ and $\gamma \in \ocint{0,\bgamma}$
  \begin{equation}
    \label{eq:V_bound_continuous}
    \pi_{\theta}  \Pker_{\gamma m_{\gamma}, \theta}(V) \leq D_0 \eqsp , \qquad \pi_{\theta}  \Rker_{\gamma,\theta}^{m_{\gamma}}(V) \leq \tilde{D}_1 \eqsp , \qquad \tilde{D}_1 = D_0 + b\lambda^{-\bgamma} \log(1/\lambda)^{-1} \eqsp ,
  \end{equation}
Combining this result and \Cref{lemma:kl_error} we have for any $\theta \in \Theta$ and $\gamma \in \ocint{0, \bgamma}$

  \begin{equation}
    \label{eq:bound_vdemi}
    \Vnorm[V^{1/2}]{\pi_{\theta}  \Pker_{\gamma m_{\gamma}, \theta}- \pi_{\theta}  \Rker_{\gamma,\theta}^{m_{\gamma}}} \leq \tilde{D}_2 \gamma^{1/2} \eqsp , 
  \end{equation}
  with 
  \begin{equation}
    \label{eq:constant_Csecond}
    \tilde{D}_2 = 2\tilde{D}_1^{1/2}(1 + \bgamma)^{1/2}\defEns{d + 2 \bgamma (\LUx^2M_V + \sup_{\theta \in \Theta} \| \nabla_x U_{\theta}(0) \|^2 )  \tilde{D}_1}^{1/2}\LUx \eqsp .
  \end{equation}
  Combining \eqref{eq:sum_decompo} and \eqref{eq:bound_vdemi} we get for any $k \in \N$, $\theta \in \Theta$ and $\gamma \in \ocint{0,\bgamma}$
  \begin{align}
    \Vnorm[V^{1/2}]{\pi_{\theta} \Rker_{\gamma,\theta}^k - \pi_{\theta} \Pker_{\gamma k, \theta} } &\leq C \tilde{D}_2 \sum_{\ell = 1}^{q_{\gamma}} \xi^{\gamma m_{\gamma} \ell} \gamma^{1/2} \leq C \tilde{D}_2 (1-\xi)^{-1}\gamma^{1/2} \eqsp ,
  \end{align}
  where we used that $\xi^{\gamma m_{\gamma}} \leq \xi$. The conclusion follows from this result and \eqref{eq:limite}.
\end{proof}

\subsubsection{Proof of \Cref{thm:salem_langevin_cv}}
\label{sec:proof-crefthm:s-3}

Combining \Cref{propo:lyap_condition} and \Cref{lemma:v_norm_soul} we get that \Cref{assum:condition_majo_V}-\ref{assum:condition_majo_V_i} is satisfied with constant $A_1 \leftarrow D_1$. \Cref{assum:glip}, \Cref{assum:curvature}, \Cref{propo:lyap_condition} and \Cref{lemma:V_control}-\ref{item_V_control_i} ensure that \Cref{assum:condition_majo_V}-\ref{assum:condition_majo_V_ii} is satisfied  
by \cite[Theorem 14]{debortoli2018back} with $A_3 \leftarrow D_3$. \Cref{assum:condition_majo_V}-\ref{assum:condition_majo_V_iii} is satisfied combining \Cref{propo:lyap_condition}, \Cref{propo:unif_ergo_c} and \Cref{propo:discrete_vs_continuous} with $\Psibf(\gamma) \leftarrow D_4 \gamma^{1/2}$.

\subsection{Proof of \Cref{thm:salem_langevin_cv_fixed}}
\label{sec:proof-crefthm:s-1}

We preface the proof by a technical lemma.
\begin{proposition}
  \label{lem:v_kernel_error}
 Let $V: \rset^d \to [1,+\infty)$ and $M_{V,4} \geq 0$ such that $\sup_{ x \in \rset^d} \defEns{(1+\| x \|^4)/V(x)} \leq M_{V,4}$. Assume that there exists $M \geq 1$ such that for any $\theta \in \Theta$, $\gamma \in \ocint{0,\bgamma}$, with $\bgamma >0$ and $x \in \rset^d$, $\Rker_{\gamma, \theta} V(x) \leq M V(x)$. Assume \tup{\Cref{assum:glip}} and \tup{\Cref{assum:glip_theta}}, then we have for any $\theta_1, \theta_2 \in \Theta$, $ \gamma_1, \gamma_2 \in \ocint{0,\bgamma}$ with $\gamma_2 < \gamma_1$, $a \in \ccint{1/4,1/2}$ and $x \in \rset^d$
  \begin{equation}
    \label{eq:V_kernel_error}
    \Vnorm[V^{a}]{\updelta_x \Rker_{\gamma_1, \theta_1} - \updelta_x \Rker_{\gamma_2, \theta_2}}  \leq D_5 \parentheseDeux{ \gamma_1/\gamma_2 -1 + \gamma_2^{1/2} \| \theta_1 - \theta_2 \| }V(x)^{2a}  \eqsp ,
  \end{equation}
  where $\{\Rker_{\gamma,\theta},\gamma \in \ocint{0,\bgamma}, \theta\in \Theta\}$ is the sequence of Markov kernels associated with the recursion~\eqref{eq:euler_maruyama_langevin}
and 
\begin{equation}
  D_5 = \max \left( 2M^{1/2}\parentheseDeux{d/4 + \sup_{\theta \in \Theta} \| \nabla_x U_{\theta}(0) \|^2 + \LUx^2M_{4,V}^{1/2} }^{1/2},  (2 M)^{1/2} L_U \right) \eqsp .
\end{equation}

  \end{proposition}

  \begin{proof}
Let $x \in \rset^d$, $\theta_1, \theta_2 \in \Theta$ and $\gamma_1, \gamma_2 \in \ocint{0,\bgamma}$, $\gamma_2 < \gamma_1$.     Using \cite[Lemma 24]{durmus2017unadjusted} we have that 
    \begin{align}
      &\Vnorm[V^{a}]{\updelta_x \Rker_{\gamma_1, \theta_1} - \updelta_x \Rker_{\gamma_2, \theta_2}} \\
      &\phantom{aaaa} \leq \sqrt{2} \left( \Rker_{\gamma_1, \theta_1}V^{2a}(x) + \Rker_{\gamma_2, \theta_2}V^{2a}(x) \right)^{1/2}  \KL{ \updelta_x \Rker_{\gamma_1, \theta_1}}{ \updelta_x \Rker_{\gamma_2, \theta_2}}^{1/2} \\
       &\phantom{aaaa} \leq 2 M^{a} V^{a}(x) \KL{ \updelta_x \Rker_{\gamma_1, \theta_1}}{ \updelta_x \Rker_{\gamma_2, \theta_2}}^{1/2} \numberthis \label{eq:pinsker}
    \end{align}
Denote for any $\upmu \in \rset^d$ and $\upsigma^2 >0$, $\upgamma_{\upmu,\upsigma^2}$ the $d$-dimensional Gaussian distribution with mean $\upmu$ and covariance matrix $\upsigma^2 \Id$. Using that for any  $\upmu_1, \upmu_2 \in \rset^d$ and $\upsigma_1, \upsigma_2 >0$,
\begin{equation}
      \KL{\Upsilon_{\upmu_1, \upsigma_1 \Id}}{\Upsilon_{\upmu_2, \upsigma_2 \Id}}= \norm{\upmu_1 - \upmu_2}^2/(2\upsigma_2^2) + (d/2)\defEns{-\log(\upsigma_1^2/\upsigma_2^2) - 1 + \upsigma_1^2/\upsigma_2^2} \eqsp .
    \end{equation}
  In addition, if $\upsigma_1 \geq \upsigma_2$
  \begin{equation}
    \KL{\Upsilon_{\upmu_1, \upsigma_1 \Id}}{\Upsilon_{\upmu_2, \upsigma_2 \Id}} \leq \norm{\upmu_1 - \upmu_2}^2/(2\upsigma_2^2) + (d/2)(1 -\upsigma_1^2 / \upsigma_2^2)^2\eqsp .
  \end{equation}
  Therefore, we obtain that 
\begin{equation}
  \label{eq:lem_v_kernel_error_0}
            \KL{ \updelta_x\Rker_{\gamma_1, \theta_1}}{ \updelta_x\Rker_{\gamma_2, \theta_2}} \leq \Xi/(4\gamma_2) + (d/2)(1 - \gamma_1/\gamma_2)^2 \eqsp ,
          \end{equation}
          where $\Xi$ satisfies 
    \begin{align}
\Xi &=      \| \gamma_1 \nabla_x U_{\theta_1}(x) - \gamma_2 \nabla_x U_{\theta_2}(x) \|^2 
                                                       \\ & =\| \gamma_1 \nabla_x U_{\theta_1}(x)- \gamma_2 \nabla_x U_{\theta_1}(x)  + \gamma_2 \nabla_x U_{\theta_1}(x) - \gamma_2 \nabla_x U_{\theta_2}(x) \|^2 \\
                                                      &\leq  2\| \gamma_1 \nabla_x U_{\theta_1}(x) - \gamma_2 \nabla_x U_{\theta_1}(x)\|^2 + 2 \|\gamma_2 \nabla_x U_{\theta_1}(x) - \gamma_2 \nabla_x U_{\theta_2}(x) \|^2 \\
      &\leq   2 (\gamma_1 - \gamma_2)^2 \| \nabla_x U_{\theta_1}(x) \|^2+ 2\gamma_2^2\| \nabla_x U_{\theta_1}(x) - \nabla_x U_{\theta_2}(x) \|^2 \\
      \label{eq:lem_v_kernel_error_1}
      & \leq   2 (\gamma_1 - \gamma_2)^2 \| \nabla_x U_{\theta_1}(x) \|^2+ 2\gamma_2^2  L_U^2 \| \theta_1 - \theta_2 \|^2 V^{2a}(x)\eqsp,                                                                                \end{align}
    where we have used \Cref{assum:glip_theta} in the last line.
Using \Cref{assum:glip_theta} again and that $\sup_{\theta \in \Theta} \norm{\nabla_x U_{\theta}(0)} < \plusinfty $ by \Cref{assum:glip}, we get for any $a \in \ccint{1/4,1/2}$
\begin{equation}
  \| \nabla_x U_{ \theta}(x) \|^2 \leq   2 ( \norm{\nabla_x U_{ \theta}(x) - \nabla_x U_{ \theta}(0)}^2 + \sup_{\theta \in \Theta} \norm{\nabla_x U_{\theta}(0)}^2)
\leq   C_{\Theta} V^{2a}(x) \eqsp ,
\end{equation}
with $C_{\Theta} =   2 \sup_{\theta \in \Theta} \| \nabla_x U_{\theta}(0) \|^2 + 2\LUx^2M_{4,V}^{1/2}$.
Combining this result, $\log(\gamma_2/\gamma_1) \leq 0$ and and \eqref{eq:lem_v_kernel_error_1} in \eqref{eq:lem_v_kernel_error_0}, it follows that                                    \begin{align}                                                                             
                                                                                                                                                                                                                                                                                          & \KL{\updelta_x \Rker_{\gamma_1, \theta_1}}{\updelta_x \Rker_{\gamma_2, \theta_2}}
                                                                                                                                                                                                                                                                                          \leq d (1 - \gamma_1/\gamma_2)^2/2  \\ &+ \gamma_2^{-1} (\gamma_1 - \gamma_2)^2 \| \nabla_x U(\theta_1,x) \|^2/2 +\gamma_2 L_U^2 \| \theta_1 - \theta_2 \|^2 V^{2a}(x) /2 \\
                                                                                                                                                                                                                &\leq \parentheseDeux{d \gamma_2^{-1}(1 - \gamma_2/\gamma_1)/4  + \gamma_2^{-1} (\gamma_1 - \gamma_2)^2 C_{\Theta} /2 +\gamma_2L_U^2 \| \theta_1 - \theta_2 \|^2/2}V^{2a}(x) \eqsp .
\end{align}                                                                                                                                                              This result substituted in  \eqref{eq:pinsker} completes the proof with the fact that  for any $a,b \in \rset_+$, $(a+b)^{1/2} \leq a^{1/2} + b^{1/2}$.
  \end{proof}

  \begin{proof}[Proof of  \Cref{thm:salem_langevin_cv_fixed}]
       \Cref{assum:glip} and \Cref{assum:curvature} ensure a uniform drift condition on $\Rker_{\gamma, \theta}$, see \Cref{propo:lyap_condition} . Note that the Lyapunov function $V$ defined by \Cref{propo:lyap_condition} satisfies  $\sup_{x \in \rset^d} (1+\| x \|^4)/ V(x) < +\infty$.
   \Cref{assum:condition_kernel_fix} 
   is then a direct consequence of \Cref{lem:v_kernel_error}
  \end{proof}


\end{document}